\documentclass[12pt, a4paper]{article}

\usepackage{dsfont}
\usepackage{fullpage,graphicx}
\usepackage{amsmath,amssymb,amsfonts,amsthm,mathtools}
\usepackage{natbib,float,enumerate,tabularx,booktabs,url,comment,enumitem}
\usepackage{geometry}
\usepackage[onehalfspacing]{setspace}
\usepackage[titletoc,toc]{appendix}
\usepackage{minitoc}
\renewcommand \thepart{}
\renewcommand \partname{}
\usepackage[flushleft]{threeparttable}
\usepackage{ragged2e}
\usepackage{subcaption}
\usepackage{tikzit}

\tikzstyle{new style 0}=[fill=white, draw=red, shape=circle]

\tikzstyle{cutoff}=[-, draw=red]
\tikzstyle{arrow R}=[->]
\tikzstyle{arrow L}=[<-]
\tikzstyle{fading}=[-, draw={rgb,255: red,191; green,191; blue,191}, dashed]

\usetikzlibrary{chains, positioning, arrows}
\usepackage{lscape}
\usepackage{placeins}
\usepackage{epigraph}

\usepackage{etoolbox}

\makeatletter
\patchcmd{\epigraph}{\@epitext{#1}}{\itshape\@epitext{#1}}{}{}
\makeatother

\usepackage{xcolor} 
\usepackage{colortbl}
\usepackage{booktabs}
\usepackage{multirow}

\usepackage[thinc]{esdiff}
\usepackage{bm,bbm}

\DeclareMathOperator*{\sst}{ss}

\newtheoremstyle{newthmstyle}%
  {\topsep}
{\topsep}
{\itshape}
{0pt}
{\bfseries\color{blue}}
{. }
{0pt}
{}
\theoremstyle{newthmstyle}

\newtheorem{lemma}{{\color{blue}Lemma}}
\newtheorem{definition}{Definition}
\newtheorem{assumption}{Assumption}
\newtheorem{remark}{{\color{blue}Remark}}
\newtheorem{prop}{{\color{blue}Proposition}}
\newtheorem{thm}{{\color{blue}Theorem}}

\newtheorem{appprop}{{\color{blue}Proposition}}

\newtheorem{appremark}{{\color{blue}Remark}}
\newtheorem{applemma}{{\color{blue}Lemma}}

\theoremstyle{definition}

\allowdisplaybreaks
\renewcommand\Pr{\mathbb{P}}
\renewcommand\Re{\mathbb{R}}
\newcommand\E{\mathbb{E}}

\renewcommand\hat[1]{\widehat{#1}}
\renewcommand\tilde[1]{\widetilde{#1}}

\newcommand{\g}{\mid}%

\renewcommand{\d}{\ \mathrm{d}}

\newcommand{\ulz}{\underline{z}}
\newcommand{\indic}[1]{ \mathbb{I}_{\left\{z \geq #1\right\}}(z) }
\newcommand{\indicalt}[1]{ \mathbb{I} \left\{z \geq #1\right\}}
\newcommand{\indictwo}[2]{ \mathbb{I}_{\left\{#2 \geq #1\right\}}(#2)}

\usepackage{tocvsec2}
\usepackage{xpatch}
\usepackage{booktabs}

\makeatletter
\xpatchcmd{\sv@part}{\huge \bfseries \partname \nobreakspace \thepart \par \vskip 20\p@ \fi \Huge \bfseries #2}{\fi \Huge \bfseries \thepart. #2}{}{}
\makeatother

\renewcommand \thepart{}
\renewcommand \partname{}

\newcommand{\subfloat}[2]{%
    \begin{minipage}{0.33\textwidth}
        \includegraphics[width=\textwidth]{ #1 } \newline
        \mbox{\hspace{0.5\textwidth}} #2
    \end{minipage}%
}

\renewcommand{\Pr}{\mathbb{P}}
\renewcommand{\Re}{\mathbb{R}}
\renewcommand{\d}{\ \mathrm{d}}

\newcommand{\ela}{\mathcal{E}}

\renewcommand{\hat}[1]{\widehat{#1}}
\renewcommand{\tilde}[1]{\widetilde{#1}}
\newcommand{\olpsi}{\overline{\psi}}

\renewcommand{\olpsi}{\overline{\psi}}

\newif\ifhidelongversion
\hidelongversiontrue 
\newcommand{\versionlong}[1]{\ifhidelongversion\else#1\fi}

\newif\ifhideshortversion
\hideshortversionfalse 
\newcommand{\versionshort}[1]{\ifhideshortversion\else#1\fi}
\begin{document}

\doparttoc 
\faketableofcontents 


\def \thetitle{Not-So-Cleansing Recessions}
 \title{Not-So-Cleansing Recessions\footnote{We are thankful to Edoardo Acabbi, Isaac Baley, David Baqaee, Florin Bilbiie, Timo Boppart, Paco Buera, Ariel Burstein, Andrea Caggese, Vasco Carvalho, Matteo Escud\'e, Gene Grossman, David H\'emous, Felix Kübler, Omar Licandro, Luca Mazzone, Mathieu Parenti, Lorenzo Pesaresi, Michael Peters, Francisco Queir\'os, Edouard Schaal, Florian Scheuer, Adam Spencer, Vincent Sterk, and Jaume Ventura as well as seminar participants at UZH, Nova SBE, Crei-UPF, Bristol, Oxford, Queen Mary University London, UCL, SFI, and presentations at the Esade Macro Meetings, 2025 SED Summer Meeting, and Firm Heterogeneity and Macro workshop for their feedback. We thank Andrea Caggese and Andrea Sy for sharing the data. Igli Bajo gratefully acknowledges financial support from the Swiss National Science Foundation.  Alessandro Ferrari gratefully acknowledges financial support from the Severo Ochoa Programme for Centres of Excellence in R\&D (Barcelona School of Economics
CEX2024-001476-S). Linen Yu provided excellent research assistance.}}
\def \theshorttitle {} 
\author{ \color{blue}\large{Igli Bajo}\color{black}
\\ \small \hspace{-8pt} University of Zurich \& SFI  \and %
\color{blue}\large{Frederik H. Bennhoff
}\color{black}
\\ \small \hspace{-8pt} University of Zurich    \and
\color{blue}\large{Alessandro Ferrari}\color{black}
\\ \small \hspace{-8pt} UPF, CREi, BSE \& CEPR
}
\def \thedate {\today} 

\maketitle
\thispagestyle{empty} 

\vspace{-30pt}
\begin{abstract} 
Recessions are periods in which the least productive firms in the economy exit, and as the economy recovers, they are replaced by new and more productive entrants. These \textit{cleansing effects} improve the average firm productivity. At the same time, recessions induce a loss of varieties. In an economy with Homothetic Single Aggregator technology, we show that their long-run welfare effects trade off these two forces. This trade-off is governed by \textit{love-of-variety} and the elasticity of substitution in aggregate production. If industry output is aggregated using the standard CES aggregator, recessions do not improve long-run GDP or welfare. If the economy features more \textit{love-of-variety} than CES, the social planner optimally subsidizes economic activity both in steady state and even more so in recessions to avoid firm exit. We use the model and quasi-exogenous variation in demand to estimate love-of-variety. We find it to be significantly higher than implied by CES aggregation, suggesting that even the long-run effects of recessions are negative. Finally, we quantitatively characterize the optimal policy response both along the transition and in the steady state.

\end{abstract}

\begin{raggedright} Keywords: cleansing effects of recessions, business cycle, love-of-variety.\\
JEL Codes: E32, D31, L11\\
\end{raggedright}
\newpage
\setcounter{page}{1}
\onehalfspacing

\epigraph{\justifying This is really at the bottom of the recurrent troubles of capitalist society. They are but temporary. They are the means to reconstruct each time the economic system on a more efficient plan. But they inflict losses while they last, drive firms into the bankruptcy court, throw people out of employment, before the ground is clear and the way paved for new achievement of the kind which has created modern civilization and made the greatness of this country.}{--- \textup{Joseph A. Schumpeter} on the Great Depression, 1934}

\section{Introduction}
Recessions are periods of increased reallocation of economic activity. The \textit{liquidationist} view of business cycles has long posited that during downturns, unproductive economic units exit, and, as the economy recovers, they are replaced by new and more productive firms. \cite{stiglitz1993endogenous} refers to this effect as the \textit{``silver lining of economic recessions''}. The intellectual roots of this argument can be traced back to \cite{schumpeter1934}, who viewed recessions as a necessary short-term pain through which the economy improves its long-run outcomes. 

Figure \ref{fig:loss_of_m_es} provides the empirical motivation for our analysis. It shows the evolution of the number of active firms in the administrative data from Spain (SABI). Relative to a linear trend estimated through 2006, the number of firms exhibits a large and persistent decline following the 2007-08 recession. Figure \ref{fig:tfp_entrant_exiters}, based on the same data, shows that while in \textit{normal times} the estimated TFP distributions of entering and exiting firms are very similar, this is not the case during recessions. We find that the distribution of entering firms first-order stochastically dominates that of exiting ones. This represents prima facie evidence of a positive reallocation during recessions: exiting firms are replaced by, on average, better ones, consistent with the creative destruction in Schumpeterian theories. However, we argue that this reallocation may not necessarily yield output or welfare gains even in the long run, when the recession has subsided.

\begin{figure}[htb]
    \centering
    \includegraphics[width=0.65\linewidth]{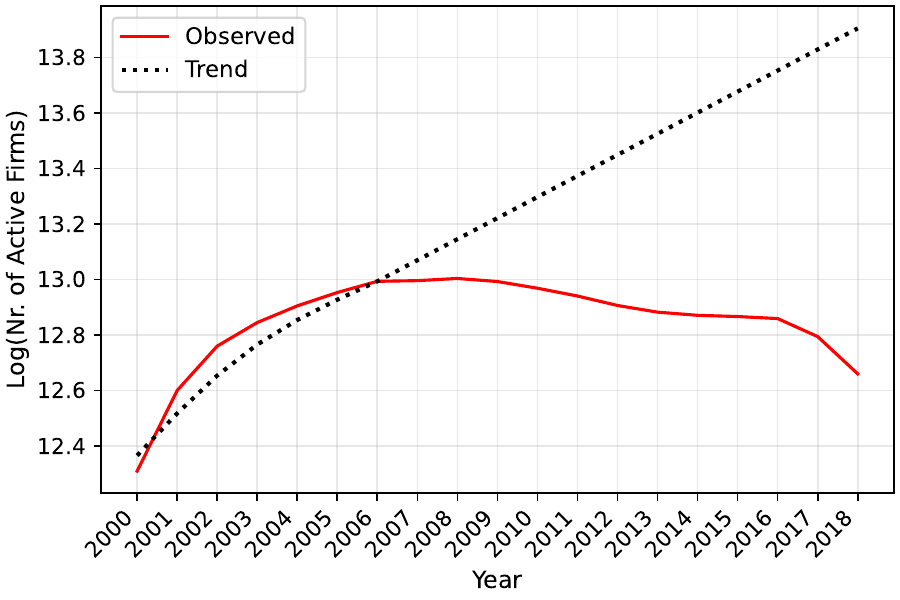}
        \caption{Number of Firms in Spain}
    \label{fig:loss_of_m_es}
\end{figure}

\begin{figure}[htp]
  \centering
  \begin{minipage}{0.49\linewidth}
    \centering
    \includegraphics[width=\linewidth]{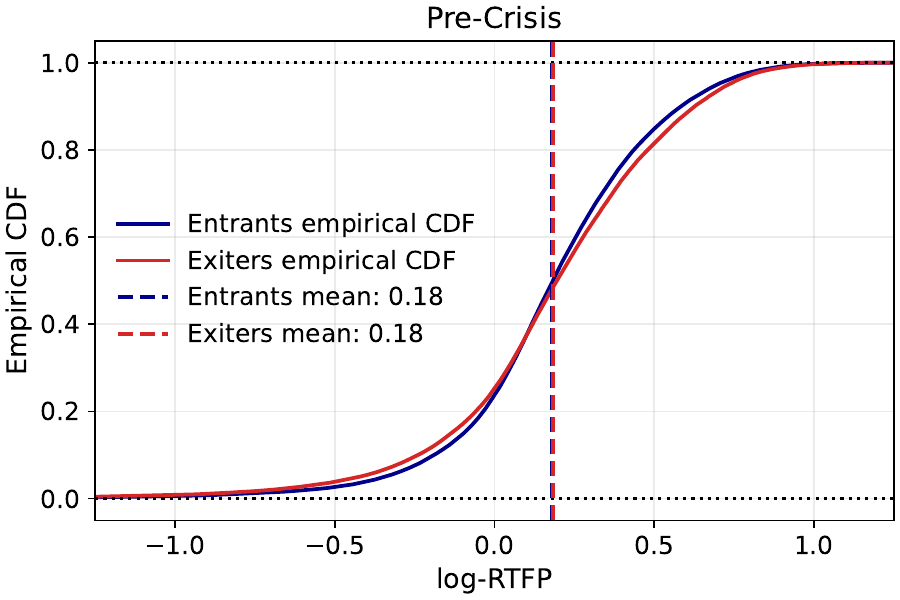}
  \end{minipage}\hfill
  \begin{minipage}{0.49\linewidth}
    \centering
    \includegraphics[width=\linewidth]{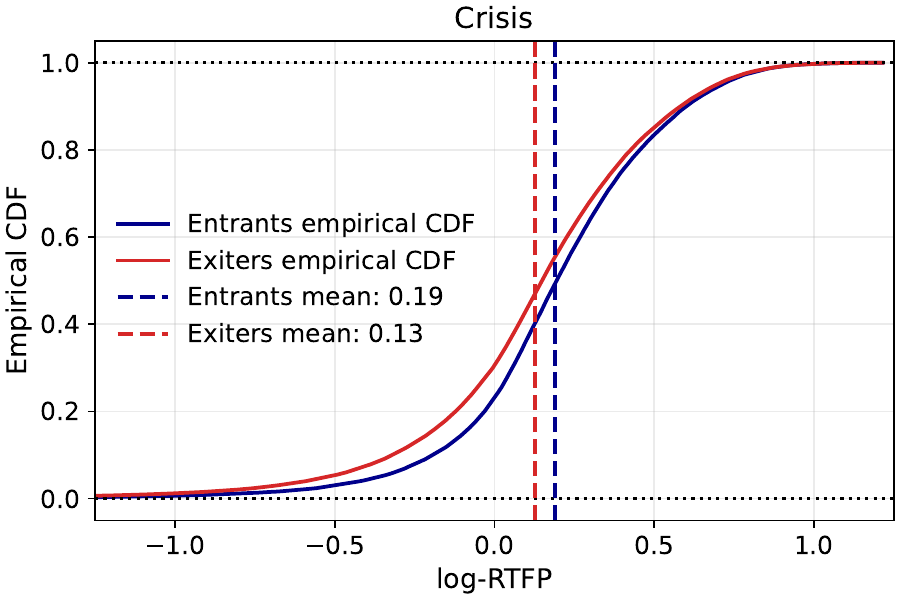}
  \end{minipage}
   \caption{Empirical CDFs (ECDFs) of log Revenue TFP (RTFP) for Entrants and Exiters. The estimation procedure for RTFP is described in Appendix \ref{app:sabi_tfp_estim}. 
  In the period preceding the Global Financial Crisis (2002-2006), the entrant and exiter ECDFs overlap and cross, and we fail to reject equality of means at the 5\% level. During the Global Financial Crisis (2007-2009), the entrant ECDF lies below the exiter ECDF throughout the support---consistent with first-order stochastic dominance of entrants---and we reject equality of means at the 5\% level.}

  \label{fig:tfp_entrant_exiters}
\end{figure}

In this paper, we revisit the role of recessions as moments of \textit{cleansing} and creative destruction and, in particular, their long-run effects on output and welfare. 
We consider a parsimonious model of firm dynamics with heterogeneous firms \`a la \cite{hopenhayn1992entry}-\cite{melitz_impact_2003} and study the cleansing effect of business cycles driven by fluctuations in the fixed cost of production.  In our model, temporary fluctuations can permanently change the firm size distribution, generating history dependence.  In this context, we show that increases in fixed costs have cleansing effects through Schumpeterian forces: after a recession, fewer firms operate, and they are, on average, more productive than before the downturn. However, these productivity gains do not necessarily translate into higher long-run output and welfare. 

When products are differentiated, the output and welfare effects of changes in the number of varieties and in average productivity are governed by different forces. Changes in the number of firms translate into output and welfare according to the \textit{love-of-variety} (LoV) of the aggregator. Conversely, changes in the average productivity of firms affect the economy through the elasticity of demand. We formalize this trade-off in a partial equilibrium economy where a final good producer has a \textit{Homothetic, Single Aggregator (HSA)} technology \citep{matsuyama2023non}. 
We characterize the effects of recessions driven by temporary increases in the fixed cost of production in terms of three forces: i) the external effects of the change in the number of available varieties; ii) the change in a sufficient statistic for the state of competition in the economy; and iii) the average surplus difference between exiting and entering firms. If the presence of more varieties generates positive externalities in production via LoV, the first effect is negative since recessions tend to reduce the number of firms. The second effect is generally zero since the improvement in average efficiency is exactly offset by the reduction in the number of firms. Finally, the contribution of the third effect fundamentally depends on whether the demand elasticity is increasing in price (and, by extension, decreasing in productivity), a property known as \textit{Marshall's Second Law of Demand} \citep[][]{mayer2021product,matsuyama2023non}. If this property holds, the third effect is positive since average productivity increases after the recession.

While this general characterization clarifies the fundamental forces shaping long-run recession effects, it remains intractable for analytical and empirical work because the aggregator is only implicitly defined. We make further headway by considering a simple special case of HSA:
\cite{dixit1975monopolistic} aggregation. This is a generalized CES where \textit{love-of-variety} $q$ and the elasticity of substitution $\sigma$ are governed by different parameters.\footnote{See \cite{brakman2001monopolistic} for a reprint of \cite{dixit1975monopolistic}. See also \cite{ethier1982national} and \cite{benassy1996taste} for similar aggregation.} In this special case, we show that the long-run welfare effect of recession is positive if and only if $q<1/(\sigma-1)$ while it is negative if and only if $q>1/(\sigma-1)$. The classical \cite{dixit1977monopolistic} CES aggregation represents a special case of our setting, where $q=1/(\sigma-1)$, in which the long-run effect of recessions is exactly zero. In this knife-edge case, the effect of the variety loss and the average productivity gains exactly offset each other as they are governed by the same elasticity $1/(\sigma-1)$. We conclude that, in this setting, \textit{cleansing effects} in terms of average productivity are neither necessary nor sufficient to determine the long-run welfare effects of business cycles. 

Next, we show that this intuition extends to general equilibrium with a minor caveat: in GE, the cleansing effect frees up labor resources. After a recession, the economy features fewer, more productive firms. As a consequence, less of the finite labor endowment is used to pay the fixed costs of production. This implies that fixed cost cycles induce higher long-run welfare even in CES economies. Nonetheless, we show that for any fixed cost cycle there exists a unique level of LoV $q^\star$, such that long-run output and welfare are identical to their pre-cycle levels. 

We conclude our theoretical analysis with two normative results. We show that the market allocation is constrained efficient if and only if the economy features CES aggregation. This extends results by \cite{dixit1977monopolistic} and \cite{dhingra2019monopolistic} to economies that also consider the presence of incumbent firms. We then show that when the economy values varieties more or less than implied by CES, the market equilibrium features either too few or too many firms, and characterize the optimal policy intervention. The planner can correct this steady-state inefficiency with a single instrument, changing the relative entry and fixed costs. 

Importantly, we show that in economies in which output is aggregated with stronger LoV than in the CES benchmark, under mild assumptions on the productivity distribution, the social planner finds it optimal to increase fixed cost subsidies during recessions. This policy is motivated by inducing a smaller amount of exit during bad times than implied by the market equilibrium. This result speaks to frequent efforts by governments to reduce churn during bad times, as seen during the Covid-19 recession.\footnote{See \cite{kozeniauskas2022cleansing} for an empirical evaluation of such programs.}

Both our positive and normative analyses highlight the importance of the parameters governing love of variety $q$ and the elasticity of demand $\sigma$. To inform our quantitative evaluation, in Section \ref{sec:empirics}, we set out to estimate these key parameters. We do so using two equilibrium relations in our model. First, in equilibrium $\sigma$ uniquely determines the profit rate of the economy. We use the WIOD data \citep{timmer2015illustrated} to obtain the distribution of profit rates for the 56 industries in the sample across 43 countries.\footnote{WIOD reports capital expenses together with economic profits, we, therefore, make the conservative assumption that all capital expenses are, in fact, profits, leading to an upward bias of the profit rate and a downward bias in $\sigma$.} We find a median elasticity of substitution of 5.8, associated with a profit rate of $17.2\%$. To estimate $q$, we make use of another structural relation in our model: LoV governs the extent of external returns in our economy. We show that, in our setting, the elasticity of output to an increase in downstream expenditure is equal to $1+q$, where the ``1'' obtains from our constant technological returns to scale assumption. We use the quasi-exogenous shifts in foreign income from the instrumental variable developed by \cite{ferrari2023inventories}. This instrument combines shifts in destination markets for WIOD industries via their network exposure shares. We obtain an estimate of $q$ between .51 and .63, depending on the specification. Comparing these estimates with the love-of-variety in CES implied by our distribution of $\sigma$, we conclude that all but 2 out of the 56 industries in our sample have a stronger love-of-variety than implied by CES. The direct consequence of our empirical finding is that, for most of these industries, recessions are costly even in the long run, despite their cleansing effects on the productivity distribution.

We conclude by studying a quantitative version of our model. This allows us to evaluate the overall effect of recessions, accounting for transitional dynamics. We use our estimated parameters for love-of-variety and the elasticity of substitution, and calibrate the rest of the model to moments of the Spanish firm-level administrative data.
First, we show that, given the estimated love-of-variety and elasticity of substitution, the planner would like to introduce a steady-state subsidy to the fixed cost. As highlighted by our normative results, the economy features too few firms, and the planner wants to induce a higher number of varieties. Quantitatively, we find that this subsidy can increase welfare by $55.30\%$ in consumption equivalent variation (CEV). Next, we study business cycles in our quantitative model. We engineer a temporary increase in the fixed cost so that in the trough, 20.44\% of firms exit, as in Spain in 2007-09. Absent any policy intervention, this recession induces a welfare loss of $-5.28\%$ CEV. This loss is largely driven by a very persistent loss of varieties ($-9.1\%$ after 15 years), consistent with the deviation from trend observed in the data. When we allow the social planner to intervene, it does so by subsidizing firms during the recession to minimize the exit of varieties. The planner fully shifts the absorption mechanism of the recession: it induces firms not to exit and reallocates labor from production to fixed costs payments. This active policy makes the recession much shorter-lived and eliminates any long-run effect. Depending on whether the economy starts from a distorted or an undistorted steady-state, the cycle intervention can substantially reduce the welfare cost of the recession (up to 88\% of its CEV cost). This result underscores the importance of the intervention in reducing the cost of the recession. 

\paragraph{Related Literature}
This paper draws insights from the literature on the optimal number of varieties \citep{dixit1977monopolistic,dhingra2019monopolistic, matsuyama2020does, matsuyama2023love} to discuss the cleansing effects of recessions.\footnote{This problem has been the subject of a large literature which includes \cite{spence1976product,dixit1977monopolistic,mankiw1986free,zhelobodko2012monopolistic,dhingra2019monopolistic,matsuyama2020does}. See \cite{brakman2001monopolistic} for an early review.} We build on \cite{jovanovic1982selection}, \cite{hopenhayn1992entry}, and \cite{melitz_impact_2003} to study the effect of recessions in economies with love-of-variety and firm dynamics.
The notion that recessions can generate long-run benefits thanks to an acceleration of creative destruction dates back to \cite{schumpeter1934,schumpeter1939business} and has been formalized in \cite{caballero1994cleansing}. 
They study an embedded capital model in which recessions induce firm exit, thereby accelerating the speed of modernization of installed capital. We study an economy in which there is no capital, but producers are heterogeneous. During recessions, the least productive firms exit and are subsequently replaced by entrants with higher average productivity. Importantly, we show that this is not enough to infer the behavior of output or welfare, as the equilibrium number of firms drops after a recession. The patterns predicted by our model are consistent with the large empirical evidence on the procyclical properties of entry and exit, both at the firm and product level \citep{dunne1989firm,davis1992gross,broda2010product,kehrig2015cyclical,lee2015entry,argente_innovation_2018,tian2018firm,argente2024life}. 

Theoretically, our work is related to the literature on the aggregate consequences of entry, exit, and firm heterogeneity, such as \cite{chatterjee1993entry,bilbiie2012endogenous,clementi2016entry,bilbiie2019monopoly,carvalho2019large,bilbiie2020aggregate,gouin2022productivity,ferrari2022firm,bendetti2024hetereogeneous,collard2025neoclassical}, and the literature on external effects in macroeconomics following \cite{baxter1991productive}. Related to our conclusion, \cite{hamano_endogenous_2017, hamano_monetary_2022} highlight the existence, in a CES economy, of the welfare tradeoff between higher average productivity and fewer available varieties. We show that for a CES economy, the two effects perfectly offset each other in PE. In GE, the only welfare effect is the labor savings and the income windfall that this generates. The tradeoff between better firm selection and loss of variety becomes welfare-relevant only away from CES. In particular, when love-of-variety is larger than implied by the CES benchmark. \cite{barlevy2002sullying}, \cite{caballero2005cost}, \cite{ouyang2009scarring}, \cite{moreira2016firm}, and \cite{acabbi2022human} suggest that the cleansing effects may be lower than expected. \cite{barlevy2002sullying} and \cite{acabbi2022human} argue that procyclical job match quality dominates the cleansing role of recessions. \cite{ouyang2009scarring} considers a setting in which recessions may halt the entry of very productive firms and, thereby, reducing long-run growth. \cite{moreira2016firm} shows that firms born in recessions remain smaller than their counterparts born in good times. Our main result is related but distinct: even when measured productivity improves during recessions, this may induce a loss of varieties that manifests in lower output and welfare. 

Finally, our paper is closely related to \cite{ardelean_how_2006}, \cite{baqaee2023supplier}, and \cite{poilly}, who provide estimates for love-of-variety. Using firm-to-firm transactions, \cite{baqaee2023supplier} estimate the elasticity to new suppliers to be .3. Using establishment entry, \cite{poilly} estimate an elasticity of .49. Our estimates based on foreign demand shocks suggest that the aggregate love-of-variety effect is approximately .5-.6, broadly consistent with the range implied by \citet{baqaee2023supplier} and \citet{poilly}, and significantly larger than implied by a CES aggregator.

\section{Model}
\label{sec:pe}
\subsection{Monopolistic Competition with HSA Demand Structure}\label{sec:hsa_setup}
We analyze a standard model of firm dynamics \citep{hopenhayn1992entry,melitz_impact_2003} with monopolistically competitive firms \'a la \cite{dixit1977monopolistic}. We begin with partial equilibrium, where a single representative household spends a fixed budget, $\mathcal{I}$, on a composite good. The good is produced by a perfectly competitive downstream firm. This intermediary takes a bundle of intermediate products, produced by upstream firms, as inputs to its \textit{homothetic, single aggregator} (HSA) technology.

\paragraph{Households} There is a representative household with utility $\mathcal U(Y)$, $\mathcal U^\prime>0$, who spends a fixed amount of income $\mathcal{I}$ on a final consumption good $Y$. 

\paragraph{Production}
Let $\bm\Omega$ be the index set of all active `upstream' firms. Each firm, $\omega\in \bm\Omega$, is a monopolistic producer of its unique variety. To produce, a firm must pay overhead $f^c > 0$ in units of labor. The firm's production function is given by $y_\omega = z_\omega l_\omega$, where $z_\omega$ is its idiosyncratic productivity and $l_
\omega$ the amount of labor used in production. Given an output target $y_\omega$, the total labor demand is $l_\omega(y_\omega) = y_\omega / z_\omega + f^c$. Setting labor as the numeraire ($w=1$), this is also the firm's wage bill. We denote by $m(z)$ the productivity density of active firms and by $M\equiv \int m(z)dz$ the number of available varieties.\footnote{As a notational convention, we use $m(z)$ when referring to density functions, and $\mu(z)$ when referring to probability density functions that integrate to 1, i.e. $\mu(z):=m(z)/M$.}

Firms operate in monopolistic competition. Let $\bm{p} = \{p_\omega : \omega\in\Omega\}$ be the vector of prices of all firms in the market, and let the demand faced by the firm producing variety $\omega $ be $D(p_\omega, \bf{p})$. Then, the firm chooses its price to maximize profits
\begin{align} \label{eq:hsa_profit_max}
    \pi \equiv \,\max_{p_\omega \geq 0} \; D\big( p_\omega,\bm{p}\big) (p_\omega - 1/z_\omega)-f^c.
\end{align}
The demand function for intermediate goods is symmetric across all $\omega \in \Omega$. It is derived from the profit maximization problem of the final good producer, which we describe next.

\paragraph{Industry Output} A perfectly competitive firm buys varieties and combines them into a homogeneous consumption good, $Y$, sold to the household. We assume that this intermediary operates a \textit{homothetic, single aggregator (HSA)} technology:  $Y(\bm{y})$, where $\bm y$ is a product bundle of intermediate inputs. The key feature of the \textit{HSA} class of aggregators is the existence of a single aggregate price $A(\bm{p})$ that summarizes the competitive state of the economy.\footnote{In general, the unit-cost index $P(\bm{p})$ and the aggregate price index $A(\bm{p})$ are separate objects, and they do not need to coincide.} 
This aggregate price statistic is implicitly defined via the adding-up constraint (AUC)%
\begin{align}\label{eq:hsa_addup}\tag{AUC}
    1 = \int_{\bm \Omega} s\left( \frac{p_\omega}{A(\bm{p})}\right) \d \omega = \int s \left( \frac{p(z)}{A(\bm{p})}\right) m(z) \d z,
\end{align} %
where $s(x)$ is the market-share function, which is strictly decreasing in \textit{normalized prices} $x = p/A(\bm{p})$. The market-share function is the defining primitive of any HSA demand structure, and \eqref{eq:hsa_addup} can be thought of as a market clearing condition. It implies the demand function in eq. \eqref{eq:hsa_profit_max} through the relationship $s_\omega = (p_\omega D_\omega)/\mathcal{I}$.

We denote by $P(\bm{p})$ the cost index of the final good producer with unit productivity. This index solves the cost-minimization problem at prices $\bm{p}$. Combining the adding-up constraint with the market share identity, $s\left(\frac{p_\omega}{A(\bm{p})}\right) = \frac{\partial \ln P(\bm{p})}{\partial \ln p_\omega},$ we obtain the following relation between the unit cost function $P(\bm{p})$ and the market share function 
    \begin{align}
    P(\bm{p}) = A(\bm{p}) \exp \left[ -\int s\left(\frac{p(z)}{A(\bm{p})}\right) \Phi\left(\frac{p(z)}{A(\bm{p})}\right) m(z)\d z \right], \text{ where } \Phi(x) \equiv \frac{1}{s(x)}\int_{x}^{\infty} \frac{s(\lambda)}{\lambda} \d \lambda > 0. \label{eq:hsa_price}
\end{align}%
The cost index $P(\bm p)$ depends positively on the aggregate price $A(\bm p)$ and negatively on $\Phi(x)$, which is the productivity gain created by the product sold at normalized price $x$, and can be interpreted as a measure of \textit{internal returns to variety} \citep{MatsuyamaUshchev2023_LoveForVariety}. The final producer's marginal cost of production is then composed of the cost index $P(\bm p)$ and its productivity $\upsilon$. We allow productivity to depend directly on the number of varieties $M$ to capture \textit{external returns to variety}. These are any TFP-gain of the downstream firm which i) is caused by having more varieties as production inputs and ii) does not distort the relative demand for any two given varieties. Formally, the final good producer's marginal cost is given by $1/\upsilon(M)\cdot P(\bm p)$.
We obtain aggregate output by combining productivity $\upsilon(M)$, the cost index, $P$ of eq. (\ref{eq:hsa_price}), and the expenditure equation $YP=\mathcal{I}$ %
\begin{align}\label{eq:hsa_output}
    Y &= \upsilon(M) \cdot \mathcal{I} \cdot A(
\bm{p})^{-1}\cdot \exp \left[ \int s\left(\frac{p(z)}{A(\bm{p})}\right) \Phi\left(\frac{p(z)}{A(\bm{p})}\right) m(z) \d z \right]. 
\end{align}

\paragraph{Love-of-Variety} Love-of-Variety (LoV) of an aggregator, as defined in \cite{benassy1996taste}, is the elasticity of output $Y$ with respect to $M$, keeping the amount of inputs constant: $\mathcal{L}(M)=d\log Y/d\log M$. LoV is designed to capture the variety gains in eq. \eqref{eq:hsa_output}.\footnote{Formally, if $\omega$ indexes firms on the interval $[0, M]$ without any loss of generality, then\begin{align*}
    \mathcal{L}(M) = \frac{\partial \log Y\Big(\big\{ \frac{1}{M}\ \vert\ \omega \in [0, M] \big\}\Big)}{\partial \log M} = \frac{\partial \log Y(\bm{1})}{\partial \log M} - 1,
\end{align*} %
where $\bm{1}$ is a vector of ones, and constant returns to scale of $Y$ allows factoring out $1/M$ from $Y$, which explains the $(-1)$ term on the right side.} 
From \cite{MatsuyamaUshchev2023_LoveForVariety}, we note that LoV in our HSA economy can be written as \begin{align}\label{eq:hsa_lov_def}
    \mathcal{L}(M) \equiv \underbrace{(\ela_M \upsilon)(M)}_{\mathcal{L}^{ext}(M)} + \underbrace{\Phi \left( s^{-1}\left(M^{-1}\right)\right)}_{\mathcal{L}^{int}(M)}
\end{align} where $\ela_M \upsilon$ is the elasticity of TFP $\upsilon$ with respect to $M$ and captures the external return to varieties $\mathcal{L}^{ext}$, whereas $\mathcal{L}^{int}$ summarizes the internal gains from varieties through the demand function. The classical \cite{dixit1977monopolistic} CES aggregator is a special case of this formulation. Formally, it is the only HSA aggregator such that $\mathcal{L}^{int}(M)$ is constant.

\paragraph{Pricing}

Firms maximize profits in eq. \eqref{eq:hsa_profit_max} by choosing relative prices $x \equiv p/A$ given their normalized marginal cost, $z^{-1}/A$. As a technical condition, we assume that marginal revenues are strictly increasing in the relative price along the demand curve.\footnote{A sufficient condition for this assumption to hold is that the aggregator satisfies Marshall's Second Law of Demand \citep{MatsuyamaUshchev2022_SelectionSorting}.} This gives rise to a strictly increasing pricing function, $X$: %
\begin{align}
    x(z) \equiv \frac{p(z)}{A} = X\bigg( \frac{z^{-1}}{A} \bigg), \quad X' > 0,
\end{align} %
and implies that profits can be written as a function of only normalized marginal cost: $\pi\big( \frac{z^{-1}}{A} \big)$. 

\paragraph{Entry and Exit}
Ex-ante homogeneous upstream firms pay a fixed cost of entry $f^e>0$ in units of labor to enter and draw their productivity level $z$ from a baseline probability distribution $\mu^E$ with support $\bm Z \subseteq (0, \infty)$. Then, they choose to start producing and pay the fixed cost $f^c$ if and only if $\pi(z_\omega^{-1}/A) \geq 0$. This fixed cost implies the existence of a threshold productivity $\underline z$ such that the zero-profit condition
\begin{align}
    \pi\big(\underline{z}^{-1}/A\big) = 0 \label{eq:hsa_zpc}\tag{ZPC}
\end{align}
holds. Firms enter until the expected net profit equals the cost of entry, i.e., the following free-entry condition is satisfied: 
\begin{align}
    f^e \geq \int_{\bm Z} \max \{ \pi\big(z^{-1}/A\big), 0 \}\mu^{E}(z)\d z = \int_{\underline{z}}^\infty \pi\big(z^{-1}/A\big)\mu^{E}(z)\d z . \label{eq:hsa_free_entry}\tag{FEC}
\end{align}%
\paragraph{Dynamics} Finally, even if the firm pricing and entry problems are static, periods are linked through the presence of incumbent firms. For any given period $t$, incumbents are firms that were active in the economy at $t-1$. Their entry cost is sunk, and their productivity is known. We denote the distribution of incumbent firms by $m_{t-1}$. Then, the distribution of active firms at time $t$ is  given by the law of motion%
\begin{align}\label{eq:hsa_lom}\tag{LOM}
    m_t(z) =  \mathbbm{1}_{[\underline{z}_t, \infty)}(z)\cdot m_{t-1}(z) + \mathbbm{1}_{[\underline{z}_t, \infty)}(z) \cdot E_t \cdot \mu^{E}(z), \quad (z \geq 0),
\end{align} %
for some arbitrary incumbent distribution $m_{t-1}$. The first term represents the surviving incumbents and the second term the surviving entrants. 

\subsection{Equilibrium}\label{sec:hsa_eq}
An equilibrium in the economy is given by a triplet $(E_t, \underline{z}_t, A_t)$ of entrants, marginal firm, and aggregate price. The economy admits three types of equilibria: i) \textit{equilibrium with entry}; ii) \textit{equilibrium without entry}; and iii) \textit{knife-edge equilibrium}. All equilibria are formalized in the following definition. Finally, in our economy, a \textit{steady-state} is an equilibrium where the distribution does not change over time.

\begin{definition}[Equilibrium] \label{definition:HSA_equilibrium}
   Given parameters $\bm{\phi}_t=(f^c, f^e, \mathcal{I})$, and incumbents $m_{t-1}$:
    \begin{enumerate}
    \item[(1)] An equilibrium with entry is a triplet $(\underline{z}^\star_t, A^\star_t, E^\star_t)$ satisfying \eqref{eq:hsa_zpc}, \eqref{eq:hsa_free_entry} holding with equality, \eqref{eq:hsa_addup}, and \eqref{eq:hsa_lom}  with $E^\star_t>0$. 
    
    \item[(2)] An equilibrium without entry is a triplet $(\underline{z}^\star_t, A^\star_t, 0)$ satisfying \eqref{eq:hsa_zpc}, \eqref{eq:hsa_free_entry} holding with strict inequality, \eqref{eq:hsa_addup}, and \eqref{eq:hsa_lom} with $E^\star_t=0$.
    
    \item[(3)] A knife-edge equilibrium is a triplet $(\underline{z}^\star_t, A^\star_t, 0)$ satisfying \eqref{eq:hsa_zpc}, \eqref{eq:hsa_free_entry} holding with equality, \eqref{eq:hsa_addup}, and \eqref{eq:hsa_lom}  with $E^\star_t=0$.
    \end{enumerate}
    A time-$t$ equilibrium is a steady-state if at any future time $t+s$ $\bm\phi_t = \bm \phi_{t+s}$ and $(\underline{z}^\star_t, A^\star_t, E^\star_t)=(\underline{z}^\star_{t+s}, A^\star_{t+s}, E^\star_{t+s}),\,\forall s\in \mathbb{N}$.
\end{definition}


We note that the equilibrium conditions do not depend on $\mathcal{L}^{ext}$. External returns to varieties only affect the production of the downstream intermediary firm. This change in production does not translate into higher demand for the upstream firms, because total expenditure on the composite good is fixed to $\mathcal{I}$. Economies with higher LoV have larger final good output and a proportionately lower price index such that the revenue of the intermediary remains fixed at $\mathcal{I}=PY$. 

From the definition, we derive the property that any two equilibria with entry that share the same parameter vector $\bm \phi$, but differ in their incumbent distributions, have the same equilibrium productivity cutoff and aggregate price index. This is true even though the firm distributions can be different in either equilibrium. All proofs are relegated to Appendix \ref{app:proofs}.

\begin{lemma}[Invariance with Entry]\label{lemma:hsa_invariance_A_z}
    Let $\bm{\phi}$ be a parameter vector and $m_{t-1}$ and $m_{t-1}'$ be two distributions of incumbents. Suppose, the economy given $(\bm \phi, m_{t-1})$ and the economy given $(\bm \phi, m_{t-1}^\prime)$ both attain an equilibrium with entry. Then, the aggregate price statistic and cutoff coincide in both economies: $A = A^\prime  \text{ and }  \underline{z} = \underline{z}^\prime.$
\end{lemma}

This result reveals a fundamental symmetry in our model: when entry is active, the competitive state of the economy (captured by $A$ and $\ulz$) depends only on 
parameters, not on history, encoded in $m$. The intuition is as follows: firm profits depend 
on i) fixed costs $f_c$, ii) market size $\mathcal I$, and iii) competitive pressure $A$. Since the zero-profit and free-entry conditions together determine $\underline z$ and $A$ using only these variables, the equilibrium values are independent of the incumbent distribution $m_{t-1}$. Entry acts as the equilibrating force: more or fewer entrants adjust $A$ to its unique equilibrium level, regardless of how many incumbents are already present. Because entrants and incumbents affect $A$ symmetrically through the pricing equation, incumbents have no competitive advantage—their presence simply crowds out some entry that would otherwise occur.

In contrast, in an equilibrium without entry, the cutoff depends on the incumbents' distribution. Eq. \eqref{eq:hsa_free_entry} is slack, the \eqref{eq:hsa_zpc} determines the cutoff $\ulz$, and eq. \eqref{eq:hsa_addup} determines the aggregate price $A$. If the incumbent distribution $m_{t-1}$ is weighted towards low-productivity firms, the aggregate price $A$ is high due to weak competition. The opposite holds if $m_{t-1}$ is weighted toward high-productivity firms. Then, from the \eqref{eq:hsa_zpc}, the cutoff has to be different in the two competitive settings. Different histories produce different competitive pressures, making $\underline z$ and $A$ path-dependent.

These results yield a striking asymmetry in business cycle dynamics. During any boom with positive entry, the measures of competition ($A$) and minimum productivity ($\underline z$) are identical across economies, regardless of the incumbent distributions. The composition of firms may differ, but the economic outcomes do not. Recessions break this symmetry. When entry stops, the shape of the inherited distribution $m_{t-1}$ determines how competitive the market becomes, what cutoff emerges, and which firms survive. History matters, but only during contractions. Paraphrasing Tolstoy: all expansions are alike, but each recession is different in its own way.

Now, the natural questions are: i) how do business cycle dynamics shape the incumbent distribution over time, and ii) how do such changes affect output and welfare? We study these questions in the next section.

\subsection{Business Cycles} \label{sec:hsa_business_cycles}
To isolate the long-run consequences of temporary fluctuations, we compare steady states before and after shocks. We return to the full dynamics, including transitions, at the end of the section. We define business cycles as one-time, unexpected shocks to the fixed costs of production, $f^c$. Such cycles could, for example, be driven by varying financing conditions in working capital constraints.\footnote{An alternative interpretation of these shocks is that the fixed costs are produced in a different sector $c$ from labour and that downturns are triggered by reductions in the TFP of sector $c$.} We focus on this kind of recession as they have the highest potential to generate \textit{cleansing effects} since fixed costs directly govern selection. We consider alternative sources of business cycle fluctuations in Section \ref{sec:otherbc}. Our overarching goal is to characterize the effects of recessions in terms of cleansing effects, defined as changes in the productivity distribution, as in Figure \ref{fig:tfp_entrant_exiters}, and overall output effects.

Formally, consider an economy running through three phases, capturing the dynamics of crises. In phase $\tau = 1$, fixed costs equal $f_l^c$, then unexpectedly increase to $f_h^c>f_l^c$ in phase $\tau = 2$, and subsequently revert to $f_l^c$ in phase $\tau = 3$; subscripts $l$ and $h$ refer to low and high.\footnote{Note that considering the phase 3 reversion to $f_l^c$ is equivalent to studying the limit of a slow-moving process of mean-reversion of the fixed cost to the long-run mean of $f_l^c$. We characterize this problem in Appendix \ref{prop:smooth_reversion_pareto_smooth_fc}.} Our comparison of interest is between phases 1 and 3, in which the vector of parameters characterizing the economy is the same, but the history is different. As the economy does not feature slow-moving dynamics, the three phases can be thought of as three long-run steady states or as three consecutive periods.

Suppose that in phase 1, the economy achieves a steady state characterized by $\underline{z}_1$, $A_1$, and the distribution $m_1$. The distribution at the end of phase 1 is given by $m_1 = m_0(z) \mathbb{I}_{\{z\geq \underline z_1\}} + E_1\mu^E(z)\mathbb{I}_{\{z\geq \underline z_1\}}$ for some initial distribution $m_0$ such that $E_1>0$.
In phase 2, the rise of fixed costs implies that the cutoff, determined by the zero-profit condition (eq. \ref{eq:hsa_zpc}), increases to $\underline z_2 > \underline z_1$. This generates the new distribution $m_2 = m_1(z) \mathbb{I}_{\{z\geq \underline z_2\}}$, forcing some of the incumbents to exit.
We formalize this process in Proposition \ref{prop:hsa_phase2} and panels A and B of Fig. \ref{fig:entry_exit}.

\begin{prop}[Equilibrium During the Recession]\label{prop:hsa_phase2}
    Let the phase-1 equilibrium given incumbents $m_0$ and $\bm \phi_1 =(f_l^c, f^e,\mathcal{I})$ be characterized by $(\underline z_1, A_1, E_1)$ with $E_1>0$ and distribution $m_1$. Consider the phase-2 equilibrium under $\bm \phi_2 =(f_h^c, f^e,\mathcal{I})$ with incumbents $m_1$ and $f^h>f^c$. Then, in phase 2, the cutoff productivity increases $\underline z_2 > \underline z_1$, the aggregate price increases $A_2 > A_1$, and there is no entry $E_2 = 0$.%
    \end{prop}
In phase 3, the fixed costs revert to $f_l^c$. Survivors of the crisis are now incumbents and new firms enter. Equations \eqref{eq:hsa_addup}, \eqref{eq:hsa_zpc} and \eqref{eq:hsa_free_entry} hold with equality, and $E_3 > 0$. Panel C of Figure \ref{fig:entry_exit} shows the post-recession distribution $m_3$. Note that the productivity of the least productive active firm is the same before and after the cycle, since by Lemma \ref{lemma:hsa_invariance_A_z}, the cutoff $\underline z_3$ equals the pre-crisis cutoff $\underline z_1$. Yet, the distribution of firms differs and $m_1 \neq m_3$ due to the presence of incumbents. 

This surprising result can be understood by noting that in this model, there is an important asymmetry between entry and exit. Entry occurs under the veil of ignorance: firms do not know their productivity when they choose to enter. On the other hand, exit occurs when firms already know their type $z$, which implies that exit is \textit{selected}. In other words, entry occurs along the entire support above the cutoff $\underline z$, while exit always happens at the bottom of the distribution. 

\begin{figure}[htbp]
    \centering
    \includegraphics[width=1\linewidth]{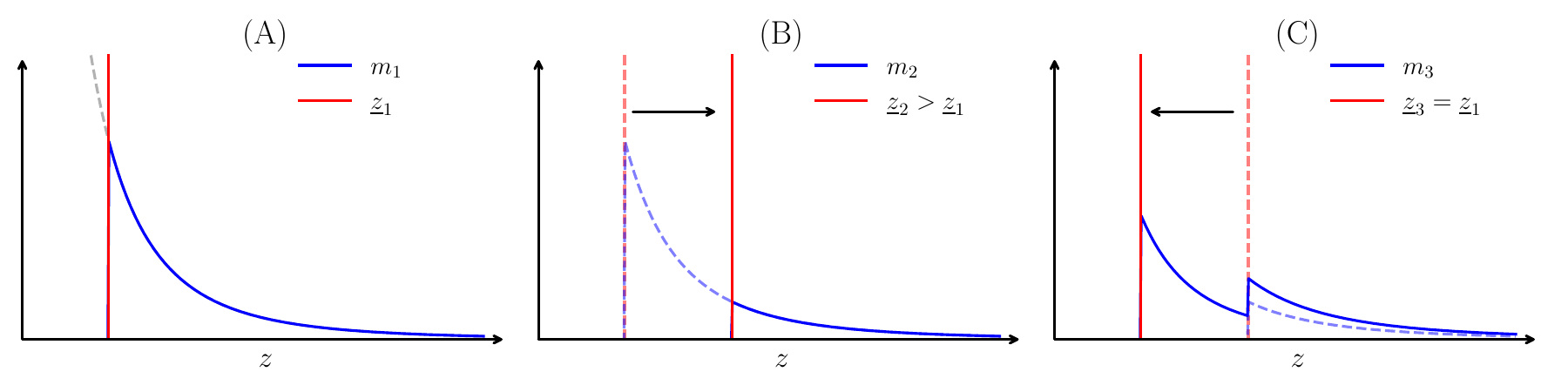}
    \caption{The figure shows the entry and exit dynamics over the business cycle. Panel (A) shows the distribution $m_1$ before the shock hits. Upon impact, the left tail of firms with productivity less than $\underline z_2$ leave, creating distribution $m_2$ (B). Finally, after fixed costs return to pre-shock levels and new firms drawn from the baseline distribution ($\mu^E$) enter, $m_3$ becomes the distribution of productivities in the market (C). The dashed light-blue line refers to $m_{t-1}$.}
    \label{fig:entry_exit}
\end{figure}

The asymmetry between unselected entry and selected exit means that recessions and booms can change the makeup of the productivity distribution and the equilibrium number of active firms. In Theorem \ref{thm:hsa_tradeoff}, we formalize the trade-off between the \textit{cleansing effects} and the \textit{variety effects} of recessions.

\begin{thm}[Fundamental Trade-Off of Recessions]\label{thm:hsa_tradeoff}
Let
$$
\zeta_E(z)\equiv \frac{\mu^E(z)\,\mathbb{I}\{z\ge \ulz_1\}}{\Pr_{\mu^E}(z\ge \ulz_1)}
\qquad\text{and}\qquad
\zeta_X(z)\equiv \frac{m_1(z)\,\mathbb{I}\{\ulz_1\le z\le \ulz_2\}}{\int_{\ulz_1}^{\ulz_2} m_1(z)\,\d z}
$$
denote the successful-entrant and exiter probability densities, respectively, and let $
\mu_1(z)\equiv \frac{m_1(z)}{M_1}$ and $ \mu_3(z)\equiv \frac{m_3(z)}{M_3}
$ denote the pre- and post-recession productivity probability densities, respectively. Then, the following holds

\begin{itemize}
    \item 
 If $\zeta_E$ strictly first-order stochastically dominates $\zeta_X$: $\zeta_E\succ_{\mathrm{1}}\zeta_X$, then the post-recession mass of varieties falls and the productivity distribution mass shifts right:
\begin{align}\label{hsa:loss_of_variety_right_shift_pdist}
    \zeta_E \ \succ_{\mathrm{1}}\ \zeta_X 
\quad\Longrightarrow\quad
\Delta M<0\ \ \text{and}\ \ \mu_3 \ \succ_{\mathrm{1}}\ \mu_1.
\end{align}
\item  Conversely, if $\zeta_X$ strictly first-order stochastically dominates $\zeta_E$: $\zeta_X\succ_{\mathrm{1}}\zeta_E$, then the post-recession mass of varieties rises, and the productivity distribution mass shifts left:
\begin{align}
    \zeta_X \ \succ_{\mathrm{1}}\ \zeta_E
\quad\Longrightarrow\quad
\Delta M>0\ \ \text{and}\ \ \mu_1 \ \succ_{\mathrm{1}}\ \mu_3.
\end{align}
\end{itemize}
\end{thm}
Theorem \ref{thm:hsa_tradeoff} highlights that a recession may either reduce the set of available varieties and improve the post-recession productivity distribution, or expand the set of available varieties and deteriorate the post-recession productivity distribution. If entrants are better than exiters, in the FOSD sense, then the number of active firms falls. As a consequence, we argue that, if a recession has \textit{cleansing effects}, then it induces a loss of variety. To sign the output and welfare implications of recessions requires asking which of these two effects dominates and fundamentally depends on the relative value assigned to variety versus shifts in the productivity distribution. In the remainder of the paper, we focus on the empirically relevant case in which recessions reduce variety and shift the productivity distribution rightward, i.e. eq. \eqref{hsa:loss_of_variety_right_shift_pdist}, consistent with the evidence for Spain shown in Figure \ref{fig:tfp_entrant_exiters}. A simple sufficient condition for $\zeta_E \; \succ_{\mathrm{1}} \; \zeta_X$ is that incumbents are distributed proportionally to the entrants distribution on the exit region, i.e.\ $m_{t-1}(z)=c\,\mu^E(z) \; \forall z\in[\ulz_1,\ulz_2]$ for some constant $c>0$.

To characterize the behavior of output and welfare pre- and post-cycle (i.e., $Y_1$ versus $Y_3$), we use eq. \eqref{eq:hsa_output} and introduce three auxiliary objects. First, the sales-share function $r(z^{-1}/A) = (s \circ X) (z^{-1}/A)$ maps normalized marginal costs to revenue shares. Second, the function $\theta_1(z) \equiv r(z^{-1}/A_1)m_1(z)$ is the equilibrium sales share density of phase 1. Integrating it over $[s_1, s_2]$ yields $\Theta_1(s_1, s_2)\equiv \int_{s_1}^{s_2} \theta_1(z) \d z $, the sales share of all varieties produced with productivity $s_1\leq z\leq s_2$. %
Third, the function $\theta_1^E(z) \equiv \frac{r\left( \frac{z^{-1}}{A_1} \right)\mu^E(z)}{\int_{\ulz_1}^{\infty} r\left(\frac{z^{-1}}{A_1}\right) \mu^E(z) \d z}$ is the entry sales share density. It gives the sales share attributable to new entrants with productivity $z$ relative to the total sales share of all new entrants. Both densities have support $[\ulz_1, \infty)$. Theorem \ref{thm:hsa_cleansing} characterizes the effect of a recession on long-run output.

\begin{thm}[Cleansing Effects of Cycles]\label{thm:hsa_cleansing}
 Let $m_1$ be the equilibrium distribution of an equilibrium with entry. Let $\Delta \log Y\equiv \log Y_3 - \log Y_1$. 
 Then:
    \begin{align}\label{eq:hsa_cycle_Y}
       \hspace{-1.3cm}\Delta \log Y &= \underbrace{\Delta \log \upsilon(M)}_{\Delta^{ext}} - \underbrace{\Delta\log A(p)}_{\substack{\text{change in} \\ \text{ aggregate price}}} + \notag \\ & \underbrace{\overbrace{\Theta_1(\ulz_1, \ulz_2)}^{\text{rev. share weight}} \bigg[ \overbrace{\E_{\theta_1^E} \Big[(\Phi\circ X) \Big(\frac{1}{zA_1}\Big) \Big]}^{\text{avg. entrant contribution}} - \overbrace{\E_{\theta_1} \Big[  (\Phi\circ X) \Big(\frac{1}{zA_1}\Big) \ \Big|\ \ulz_2 \geq z \geq \ulz_1 \Big]}^{\text{avg. exiter contribution}} \bigg]}_{\Delta^{int} = \substack{ \text{change in average} \\ \text{surplus per variety}}},
    \end{align}
    where the first term, $\Delta^{ext}$, is the effect of the cycle through the variety externality, the second term, $\Delta\log A = 0$ by Lemma \ref{lemma:hsa_invariance_A_z}, is the change in the aggregate price, and the last term can be understood as the change in the average consumer surplus contribution of entrants vs exiters.
\end{thm}
The long-run effect of temporary increases in the fixed cost can be decomposed into three distinct terms. First, the cycle reduces the mass of varieties. Intuitively, because revenue shares are increasing in productivity, replacing the lost revenue share of exiting firms requires fewer successful entrants on average.
If the external effects are increasing in the number of firms, the loss of varieties generates a long-run output loss through $\Delta^{ext}<0$. Next, the \textit{cleansing effects} of the recession improve average productivity. However, this gain is partly offset by lower entry, such that the aggregate price $A(p)$ is unchanged, as stated in Lemma \ref{lemma:hsa_invariance_A_z}. Finally, this replacement process affects the surplus generated in the economy, above and beyond the effects on the aggregate price $A(p)$. This effect is summarized in the last term of eq. \eqref{eq:hsa_cycle_Y}. Under HSA, the elasticity of demand for each variety is free to vary with the price/quantity of the variety itself. The \textit{cleansing effects} imply that, on average, high price (low productivity) firms are replaced by lower price (higher productivity) ones. This replacement can translate into a higher or lower surplus, depending on whether the elasticity of demand is higher or lower for entrants vs exiting firms. Formally, this property is summarized by the monotonicity of $\Phi$ in the relative price $x$. Determining the sign of the change in average surplus $\Delta^{int}$ requires imposing additional properties on $\Phi$, $\theta_1$, and $\theta_1^E$. We do so through some known special cases.

\paragraph{Example 1: CES} First, suppose that $\Phi$ is globally constant. This is the case if and only if the demand system is CES \citep{MatsuyamaUshchev2023_LoveForVariety}. Then, the two expectations in eq. \eqref{eq:hsa_cycle_Y} cancel and $\Delta^{int}=0$. We can easily establish  a partial converse: if the internal effect is zero for all cutoff increases, then the demand system must be CES.

\begin{prop}[No internal effects iff CES]\label{prop:CES_HSA}
    Suppose the demand system is CES, then $\Delta^{int}=0$. Conversely, if $\Delta^{int}=0$ for all $f^c_h>f_l^c$, then the demand system is CES. As a consequence, if the demand system is CES, the long-run effect of recession is fully driven by the external effects $\Delta^{ext}$. Under CES, the long-run effect is zero if and only if there are no external effects $\Delta^{ext}=0$.
\end{prop}

Proposition \ref{prop:CES_HSA} provides a powerful benchmark. If the economy aggregates varieties through a \cite{dixit1977monopolistic} CES aggregator, then the elasticity of demand is constant across the price distribution. Therefore, replacing low-productivity exiters with higher-productivity entrants generates offsetting effects: the surplus gain per variety from better entrants exactly cancels the surplus loss from having fewer varieties. As a consequence, absent external variety effects, there are no long-run effects of business cycles on output. This result is surprising since, as we have shown before, the ex-ante and ex-post productivity distributions differ, and firms are, on average, more productive after a recession. We return to this case later in the paper. 

\paragraph{Example 2: No Incumbents}
Suppose the economy is \textit{empty} right before phase 1, namely, there are no incumbent firms that were selected in previous recessions ($m_0(z)=0,\,\forall z$). This immediately implies that $\theta_1 = \theta_1^E$. Then, the change in average surplus $\Delta^{int}$ simplifies to \begin{align}\label{eq:hsa_cycle_Y_noinc}
    \Delta^{int} = \Theta_1(\ulz_1, \ulz_2) \bigg[ \E_{\theta_1^E} \Big[(\Phi\circ X) \Big(\frac{1}{zA_1}\Big) \Big]
    - \E_{\theta_1^E}  \Big[  (\Phi\circ X) \Big(\frac{1}{zA_1}\Big) \ \Big|\ \ulz_2 \geq z \geq \ulz_1 \Big] \bigg],
\end{align}
and the sign of the internal effects is pinned down by whether $\Phi$ is increasing or decreasing. 
If $Phi$ is decreasing, the economy is characterized by a larger internal-effects contribution from varieties produced by higher-productivity firms. The opposite is true for an \textit{increasing} $\Phi$. As firms are, on average, more productive after the cycle, the internal effects are positive if $\Phi$ is decreasing, and negative if $\Phi$ is increasing. In the absence of external variety effects, the sign of $Phi$ fully governs the long-run output effects of the recession.

\citet{MatsuyamaUshchev2023_LoveForVariety} show that a low-level condition on the demand structure, \textit{Marshall's Second Law of Demand} (MSLD), is sufficient to sign the derivative of $\Phi$. In particular, if the price elasticity of demand for each variety is globally strictly increasing in its own price (MSLD), then $\Phi$ is decreasing. If the opposite holds and the own-price elasticity of demand is decreasing (anti-MSLD), then $\Phi$ is increasing.
Proposition \ref{prop:hsa_cleansing_msld} summarizes these results. 
This result is particularly relevant, given that recent empirical work finds support for Marshall's Second Law \citep{mayer2021product}.

\begin{prop}[Cleansing Effects of Business Cycles under MSLD]\label{prop:hsa_cleansing_msld}
   Suppose the economy starts with no incumbents. Also, suppose that there are no external variety effects $\Delta^{ext}=0$. Then, the following holds:
   \begin{itemize}
       \item if the aggregator satisfies MSLD,  the long-run effect of recessions is positive: $\Delta\log Y>0$.
       \item if it satisfies the anti-MSLD, then the long-run effect of recessions is negative: $\Delta\log Y<0$.
   \end{itemize}
\end{prop}

\paragraph{Example 3: External Effects of Varieties} So far, we have considered special cases that simplify the internal effects of varieties. Here, we study functional forms for the external effects $\Delta^{ext}$. We consider two simple cases: i) the externality is isoelastic in the mass of firms $M$, as in \citet{benassy1996taste}; ii) the externality is exponential in the mass of firms $M$. Proposition~\ref{prop:hsa_cleansing_external_iso} characterizes the change in output after the recession under these two specifications.
\begin{prop}[Special Cases of External Effects] \label{prop:hsa_cleansing_external_iso}
    Suppose that the external effects take the form $\upsilon(M)=e^{qM}$ for some $q\in \Re$. Then, the recession-induced output change is given by $\Delta\log Y = q\Delta M+ \Delta^{int}$ with $\Delta M<0$. 

    Suppose otherwise, that the external effects take the form $\upsilon(M) = M^q$ for some $q\in\Re$. Then, the recession-induced output change is given by $\Delta\log Y = q\Delta\log M+ \Delta^{int}$ with $\Delta \log M<0$. 
\end{prop}
These simple cases highlight how the external effects, provided that $q>0$, reduce post-recession output for any given change in the internal effects. A positive coefficient $q$ characterizes economies in which the existence of multiple input varieties generates a positive externality on the TFP of the aggregating firm.

In summary, temporary recessions driven by fluctuations in the fixed cost of production can have quasi-permanent effects. The \textit{memory} of this economy is encoded in the firm productivity distribution. Business cycles in which the fixed costs temporarily increase feature i) \textit{cleansing effects}, in the sense that the average post-recession productivity increases; ii) a net loss of variety; and iii) ambiguous effects on output. This ambiguity is driven by how changes in the firm productivity distribution translate into output and pricing choices. Depending on the specific structure of the HSA aggregator, the economy might experience a permanent loss of output and welfare even if the average productivity of firms increases. 

Hereafter, we choose a particular HSA aggregator that allows us to characterize the economy analytically. 

\subsection{Industry Output -- Isoelastic Aggregator} \label{sec:hsa_aggregator_ben}
Consider a special case of the \textit{HSA} aggregator defined in Section \ref{sec:hsa_setup}. The final good firm aggregates inputs with a generalized CES function \`a la \cite{dixit1975monopolistic}, \cite{ethier1982national}, \cite{benassy1996taste}:%
\begin{align}\label{eq:aggr}
    Y = M^{q - \frac{1}{\sigma - 1}} \; \left[ \int y(z)^{\frac{\sigma -1}{\sigma}} m(z) dz \right]^{\frac{\sigma}{\sigma-1}}.
\end{align} %
Two key parameters govern the aggregator: the constant elasticity of substitution, $\sigma$, which measures product similarity, and love-of-variety (LoV). We highlight in Remark \ref{rem:q} that in our aggregator, LoV is exactly equal to $q$. 
The classic \cite{dixit1977monopolistic} CES aggregator is a special case where $q = q^{CES} := 1/(\sigma - 1)$. Hence, $q>1/(\sigma-1)$ implies stronger LoV than in the CES case and vice versa.\footnote{Readers more familiar with the \cite{melitz_impact_2003} model aggregation can think of our aggregator as being  $Y = M^{q+1} \; \left[ \int y(z)^{\frac{\sigma -1}{\sigma}} \mu(z) dz \right]^{\frac{\sigma}{\sigma-1}}$, with $\mu(z)=m(z)/M$ being a probability density function.}
\begin{remark}\label{rem:q}
    The aggregator in eq. \eqref{eq:aggr} features constant external variety effect $\mathcal{L}^{ext}=q$ and no internal variety effect $\mathcal L^{int}=0$. As a consequence, the total love-of-variety is $\mathcal L=q$.
\end{remark}
The aggregator in eq. (\ref{eq:aggr}) satisfies a number of desirable properties, striking a balance between the tractability of CES aggregation and the ability to govern separately \textit{love-of-variety} and the price elasticity of demand. In Appendix \ref{ext:hom-aggr} we prove that in the broader class of homothetic preferences described in \cite{matsuyama2023non}, the aggregator in eq. (\ref{eq:aggr}) is the only one that jointly satisfies separability into an iso-elastic variety externality $\mathcal L^{ext}=q$ and an aggregator with zero LoV $\mathcal L^{int}=0$. In this sense, eq. (\ref{eq:aggr}) is the only aggregator in which the variety externality can be cleanly separated from demand aggregation.

\subsubsection{Equilibrium} \label{subsec:2.2}
Given the aggregator in eq. \eqref{eq:aggr}, we can derive the demand for each variety explicitly, which allows a closed-form solution for the equilibrium. Since the variety externality does not affect relative prices, the demand schedule is the same as under standard CES aggregation. As a consequence, the monopolistically competitive intermediate good producers choose a price $p(z)=\frac{\sigma}{\sigma-1}\frac{1}{z}$. At the optimum, firm profits, $\pi$, depend on the firm's productivity $z$ and on the intensity of competition, summarized by the endogenous productivity distribution $m$: 
%
\begin{align} \pi( z, m) = \frac{\mathcal{I}}{\sigma} \frac{z^{\sigma-1}}{\int  z^{\sigma -1} m(z) \d z} - f^c.\label{eq:pi-pe}\end{align}
Since the aggregator belongs to the HSA class, $m$ affects profits only through an aggregate price statistic. In this case, we can write the statistic explicitly: $A(\bm{p})=\frac{\sigma}{\sigma-1}\big(\int z^{\sigma-1}m(z)\d z\big)^{\frac{1}{1-\sigma}}$. 
For expositional convenience, we define a monotone transformation of $A$: the index $\int z^{\sigma -1} m(z)\d z$, which we call \textit{market intensity}. Intuitively, when market intensity is high, the aggregate price $A$ is low. 
We refer to this statistic as \textit{market intensity} since $z^{\sigma-1}/\int z^{\sigma -1} m(z)\d z$ is the market share of firms with productivity z. Hence, a larger \textit{market intensity} implies lower market shares for all firms. Since the statistic is a strictly decreasing transform of $A$, all results stated in terms of $A$ in Sections \ref{sec:hsa_eq} and \ref{sec:hsa_business_cycles} carry over to market intensity, with inequalities reversed.

%

As before, entrants form expectations on the profit they would obtain upon entry. In any equilibrium with entry, firms enter simultaneously until the free entry condition is satisfied.\footnote{We prove in Appendix \ref{ext:entry_equiv} that this is equivalent to the limit of an iterative entry game.} 
All the equilibria types from Definition \ref{definition:HSA_equilibrium} are characterized by the following system
\begin{align}
   & m_t(z) = m_{t-1}(z)\mathbb{I}_{\{z\geq \underline z_t\}} + E_t\mu^E(z)\mathbb{I}_{\{z\geq \underline z_t\}}, \;\; (\forall z \geq 0),\label{eq:mu-lom}\\ 
    &\E_{\mu^E} [\max\{\pi( z, m_t) , 0\} ] \leq f^e , 
    \label{eq:fe-pe} \\
   &  \pi(\underline z_t, m_t)=0. 
    \label{eq:zcp-pe} 
\end{align} 
The equations \eqref{eq:mu-lom}, \eqref{eq:fe-pe}, and \eqref{eq:zcp-pe} are the counterparts of \eqref{eq:hsa_lom}, \eqref{eq:hsa_free_entry}, and \eqref{eq:hsa_zpc}, respectively, for the aggregator in eq. \eqref{eq:aggr}. The adding-up constraint \eqref{eq:hsa_addup} is mechanically satisfied in equilibrium. 

\subsubsection{Business Cycles} \label{sec:business_cycles}
In equilibrium, aggregate output can be rewritten as 
\begin{align}\label{eq:gdp}
     Y_\tau = M_\tau^{q - \frac{1}{\sigma - 1}}  L^p_\tau\: \left(\int z^{\sigma-1} m_\tau(z)\d z\right)^{\frac{1}{\sigma-1}},
\end{align} 
where $L_\tau^p:=\int l(z)m(z) \d z$ is the amount of labor used in production, and $\left(\int z^{\sigma-1} m_\tau(z)\d z\right)^{\frac{1}{\sigma-1}}$ is aggregate technical efficiency, which we refer to as TFP. Proposition \ref{proposition:PE-cycles} characterizes how business cycles affect aggregate output relative to pre-cycle values. This is a special case of Theorem \ref{thm:hsa_cleansing} under the isoelastic aggregator in eq. \eqref{eq:aggr}.

\begin{prop}[Cleansing Effects of Cycles]\label{proposition:PE-cycles}
The change in output after the crisis is given by \begin{align}
    \Delta\log Y = \left(q - \frac{1}{\sigma-1}\right) \Delta\log M + \frac{1}{\sigma-1} \Delta\log \left(\int z^{\sigma-1} m(z)\d z\right).
\end{align} The recession reduces the number of firms but leaves aggregate productivity unchanged:
\begin{align*}
   \Delta\log M < 0 \text{ and } \; \Delta\log \left(\int z^{\sigma-1} m(z)\d z\right)=0.
\end{align*} As a consequence, output and welfare increase if and only if the economy values varieties less than in CES aggregation: 
\begin{align}
    \Delta\log Y\gtreqqless 0 \Leftrightarrow q \lesseqqgtr q^{CES}.
\end{align}
\end{prop} 
Whether long-run output and welfare are larger than before the recession depends uniquely on the LoV parameter $q$. In the long run, an economy that undergoes a temporary increase in the fixed cost of production has fewer but, on average, more productive firms.  
This result comes with several additional important observations. First, our result refers exclusively to long-run effects: any potential welfare loss is not driven by the transition in phase 2. When $q<q^{CES}$, the economy might still be worse off once transition costs are accounted for. What we show is that if $q>q^{CES}$, long-run welfare is lower even without the transition costs. Second, this negative long-run effect arises even if all the classic Schumpeterian forces of recessions are present: in our model, new entrants are, on average, more productive than exiters during the cycle. Recessions induce a \textit{selection effect}. Yet, they may still result in long-run welfare losses. The aggregate consequences of recessions trade off fewer firms, whose effect is governed by LoV $q$, with higher average productivity, whose effect is governed by the elasticity of demand and specifically $q^{CES}=\frac{1}{\sigma-1}$.
Third, we remark that the CES case is a knife-edge parametric restriction where these two coincide:

\begin{remark}[CES Aggregation]
    In the special case of CES aggregation $\left(q=q^{CES}\right)$, long-run output and welfare are unchanged pre- and post-recession:  $\Delta\log Y=0$.
\end{remark}
To build intuition for the result, it is useful to invoke Lemma \ref{lemma:hsa_invariance_A_z}. Since the size of the industry $\mathcal I$ is constant, and there is no strategic advantage of incumbents over entrants, the economy converges back to the original market intensity $\int z^{\sigma-1} m(z)\d z$ in the long-run steady state. However, the recession permanently changes the firms' distribution. Firms entering during the recovery are, on average, more productive than exiting firms. As a consequence, exiters are replaced less than 1-for-1, and the number of firms $M$ drops.


In a CES economy, output is given by $Y = L^p \left(\int z^{\sigma-1} m(z)\d z\right)^{\frac{1}{\sigma-1}}$ and therefore depends uniquely on aggregate TFP and aggregate production labor, $L^p$. The latter, however, does not depend on the fixed cost cycle but exclusively on market size $\mathcal{I}$ and $\sigma$. 
Since aggregate TFP and aggregate labor used in production are unchanged, so is output and, therefore, welfare.

This effect is best understood considering eq. (\ref{eq:gdp}) and rewriting it in terms of the total number of firms and average productivity. Taking logs and differencing pre- and post-crisis values yields %
\begin{align}\label{eq:output_FD_avgprod}
    \Delta \log Y = q \Delta \log M + \Delta \log L^p +  \Delta \log \bar z, 
\end{align} %
where we call $\bar z_\tau := \left(\int z^{\sigma-1}\mu_\tau(z) \d z\right)^{\frac{1}{\sigma-1}}$ \textit{average productivity}. The first term, $q \Delta \log M$, is the \textit{variety effect}, the second term, represents changes in labor used for production, and the third term the \textit{selection effect} of the crisis. The variety effect is driven by changes in the number of firms $M$ and its effect on output is governed by LoV, $q$. The selection effect is channeled through changes in average productivity as exiting firms, on average, are less productive than entrants. This effect is governed by $\frac{1}{\sigma-1}$ since, noting that aggregate technical efficiency is unchanged immediately implies $\Delta\log \bar z = - \frac{1}{\sigma-1}\Delta\log M$. Therefore, in PE, the variety effect dominates the selection effect for $q > q^{CES}$, while for CES economies, the two cancel out exactly. 



\subsubsection{Alternative Sources of Business Cycle Fluctuations}\label{sec:otherbc}

In our analysis so far, we have considered a specific source of business cycle fluctuations: movements in the fixed operating cost. We interpret these costs broadly as stand-ins for operating capital constraints or quasi-fixed factors. Notably, these recessions have the highest chance of generating cleansing effects since they directly affect the selection margin. Here, we consider two alternative sources of fluctuations: aggregate TFP shocks and movements in the entry cost.

\paragraph{Aggregate TFP Cycles} Consider an unexpected temporary decrease in aggregate productivity $\Xi$, such that the productivity of all firms goes from $z$ to $\Xi z<z$. We characterize the behavior of our economy in Proposition \ref{prop:ext_TFP}.
\begin{prop}[Aggregate TFP Cycles]\label{prop:ext_TFP}
    Recessions driven by temporary decreases in aggregate TFP have no long-run effects. 
\end{prop}
To understand this surprising result, note that, in our economy, temporary shocks have long-run effects only if they affect entry and exit decisions. An aggregate TFP shock moves the effective productivity of all firms equally. As a consequence, the relative productivity is unchanged. Entry/exit decisions might still be affected if the size of the industry changes. However, in this partial equilibrium setting the industry size is fixed at $\mathcal I$. Therefore, an aggregate TFP shock does not induce any cleansing effect in this economy. It does, however, affect output, which remains lower than its pre-crisis value as long as $\Xi$ remains lower than 1. Yet, as entry and exit are not affected, if $\Xi$ returns to its initial value, the economy returns to its original steady state.\footnote{ This result extends to the general equilibrium model we consider in the next section under perfectly inelastic labor supply.}

\paragraph{Entry Cost Cycles} We study the effect of changes in the entry cost. To this end, we have to change our model to have a positive entry flow in steady state. Consider a version of our economy where firms exogenously exit at rate $\delta>0$ as in \cite{melitz_impact_2003}. Then the law of motion of the productivity distribution is given by $m_t(z) = [(1-\delta) m_{t-1}(z)+E_t \mu^E(z)]\mathbbm{1}_{\{z\geq\underline z_t\}}$. This standard formulation implies that in the steady state, there is positive entry to exactly offset the churn due to the exit rate $\delta$. Importantly, it also implies that the steady state is unique and history-independent. Denote $\underline z^{\sst}$ and $M^{\sst}$ the steady-state cutoff and number of firms, respectively.

Suppose that the economy, starting from a steady state, experiences an unexpected increase in the entry cost $f^E$. We characterize the behavior of this economy in Proposition \ref{prop:ext_entry}.
\begin{prop}[Entry Cost Cycles]\label{prop:ext_entry}
    Consider a temporary increase in the entry cost. Along this transition, the following holds
    \begin{enumerate}[label = \alph*)]
        \item The number of varieties converges to  $M^{\sst}$ from below.
        \item The average productivity of active firms converges to $\bar z^{\sst}$ from below.
        \item The temporary increase in the entry cost always generates welfare losses, independently of love-of-variety. 
    \end{enumerate}
\end{prop}

Proposition \ref{prop:ext_entry} characterizes the transitional behavior of the economy after a temporary increase in the fixed cost of entry $f^e$. First, when the entry cost is higher, the business dynamism of the economy is hampered. Throughout the transition, the number of firms is below its steady-state value, inducing variety losses. Unlike a fixed-cost recession, an increase in the entry cost is not \textit{cleansing}. Along the transition, selection forces are weakened as potential entrants exert a smaller competitive pressure on unproductive incumbents. As a consequence, the average productivity of firms declines. Interestingly, output does not fall substantially on impact since the output losses are only driven by the steady erosion of the number of firms and the slow decline in average productivity. As the entry cost converges back to its original level, the economy returns to its steady state with a slow increase in the number of firms and average productivity. We conclude that a temporary increase in $f^e$ induces aggregate behavior that is more reminiscent of \textit{stagnation} episodes than of recessions.

Since both the number of firms and average productivity worsen along the transition, the welfare effects are unambiguously negative, independently of how the economy values varieties.

\section{General Equilibrium} \label{sec:ge}
In this section, we extend our result to a general equilibrium setting. Relative to the model in Section \ref{sec:hsa_aggregator_ben}, there are two changes. First, production labor demanded by firms now affects the market-clearing wage. Labor market clearing dictates:
\begin{align}
    L^p + M  f^c + E f^e = \bar L. \label{eq:lmc-ge}
\end{align}%
Second, expenditure on the consumption good $PY$ is equal to the revenues of the firms $R$, so that, setting labor as the numeraire ($w=1$), the budget constraint is $R=\bar L+\Pi$.\footnote{We maintain that labor is inelastically supplied, an assumption we relax in Appendix \ref{ext:elastic-labor} and specifically in Proposition \ref{prop:elast_labor}.} Then, firm profits become:
\begin{align}\label{eq:pi_ge}
    \pi(z, m) = \frac{R}{\sigma} \frac{z^{\sigma-1} }{\int z^{\sigma-1}m(z) \d z}\;  - f^c.
\end{align}%
Entry and exit are still pinned down by the
\textit{free-entry condition}
and \textit{zero-profit condition} (eqs. \ref{eq:fe-pe} and \ref{eq:zcp-pe}).
When entry is strictly positive, eqs. (\ref{eq:fe-pe}) and (\ref{eq:zcp-pe}) can be rearranged to obtain:\footnote{To obtain the PE equilibrium characterization just substitute $R_t$ with $\mathcal{I}$.}
\begin{align}\label{eq:entrymass-ge}
    E_t &= \frac{R_t}{\sigma f^c} \frac{\underline z_t^{\sigma -1}}{\int_{\underline z_t}^\infty z^{\sigma -1} \mu^E(z)\d z} - \frac{\int_{\underline z_t}^\infty z^{\sigma -1} m_{t-1}(z)\d z}{\int_{\underline z_t}^\infty z^{\sigma -1} \mu^E(z)\d z},\\
    \frac{f^e}{f^c} &= \int_{\underline z_t}^\infty \left[\left(\frac{z}{\underline z_t}\right)^{\sigma-1} -1 \right]\mu^E(z)d z. \label{eq:cutoff}
\end{align} 
The equilibrium is fully characterized by eqs. (\ref{eq:lmc-ge}), (\ref{eq:entrymass-ge}), and (\ref{eq:cutoff}). 
When there is no scope for entry, then $E_t = 0$, the \textit{free-entry condition} is slack, and the \textit{zero-profit condition}
\begin{align}\label{eq:zpc-noentry-ge}
    \frac{R_t}{\sigma} \frac{\underline z_t^{\sigma -1}}{\int_{\underline z_t}^\infty z^{\sigma-1}m_{t-1}(z) \d z}\;  - f^c = 0
\end{align}
pins down the mass of active firms in the economy. The equilibrium is fully characterized by eqs. (\ref{eq:lmc-ge}) and (\ref{eq:zpc-noentry-ge}), and $E_t=0$.

\subsection{Business Cycles} \label{sec:fixec_cost_cycles}
We consider the same source of fixed cost variation and establish the effect on long-run output. The following result extends Proposition \ref{proposition:PE-cycles} to general equilibrium:
\begin{prop}[Cleansing Effects of Cycles]\label{prop:output-GE} The change in output after the crisis is given by \begin{align}\label{eq:ge_ratio}
    \Delta\log Y =  \left(q - \frac{1}{\sigma-1}\right)\Delta\log M + \Delta\log L^p +\frac{1}{\sigma-1}\Delta\log\int z^{\sigma-1} m(z) \d z,
\end{align} where \begin{align}\label{eq:ge-clens}
    \Delta\log M<0, \;  \Delta\log L^p>0, \;\Delta\log\int z^{\sigma-1} m(z) \d z>0.
\end{align} There exists a unique $q^\star > q^{CES}$ (which generally depends on the change in fixed cost $\Delta \log f^c$) for which $\Delta\log Y=0$. Furthermore,
\begin{align}
    \Delta\log Y<0 \iff q>q^\star.
\end{align}
\end{prop}

General equilibrium introduces an additional effect of recessions: the \textit{labor-saving effect}. As the long-run economy features fewer operating firms, fewer labor resources are used for fixed costs as opposed to production. As a consequence, the income obtained by the household, both in terms of labor payments and profits, increases. To see this, note that the fixed supply of labor $\bar L$ is used in production $L^p$ and for fixed costs payments $Mf^c$, since, by definition, there is no entry in the steady state.  Combining labor market clearing and the firms' pricing policy, we note that GDP obtains as $PY = \frac{\sigma}{\sigma-1}(\bar L-Mf^c)$. Hence, recessions that reduce the steady-state number of firms increase GDP by reducing the amount of labor used for fixed costs payments and diverting it to production. Importantly, GDP is also the total market size that firms consider upon entry: $PY=R$. As a consequence, the labor-saving effect induces additional entry after the recession, relative to our partial equilibrium framework. 
The direct consequence of the combination of these forces is an output gain, even in the CES case, which is highlighted in Remark \ref{ces-GE}.

\begin{remark}[CES Aggregation in General Equilibrium]\label{ces-GE}
In the special case of CES aggregation $\left(q= q^{CES}\right)$, $\Delta\log Y>0$.
\end{remark} 

Based on the above intuition and remark, it is immediate that when $q>q^{CES}$, the economy values the loss of variety more than in the CES case. As a consequence, there exists a level of love-of-variety such that the household is indifferent between pre- and post-cycle outcomes as the three effects exactly cancel. 

Proposition \ref{prop:output-GE} establishes the properties of an economy experiencing a recession driven by some change in the fixed cost. In the remainder of the section, we discuss many of the assumptions we have made so far and consider recessions of different intensities.

\subsection{Discussion}
We conclude this section with a brief discussion of the most salient model primitives and their interpretations. We provide a formal treatment of these extensions in Appendix \ref{sec:extensions}.

\paragraph{Intensity of Recessions} Throughout our analysis, we have studied recessions characterized by some arbitrary fixed cost increase $f^c_h$. In Appendix \ref{sec:depth}, we ask whether larger increases in fixed costs induce stronger cleansing effects and variety losses. We show that we can partition economies into three sets based on their love-of-variety $q$. First, some economies have a $q$ so high that recessions of any intensity are welfare-reducing. Second, some economies have a $q$ so low that any recession is welfare-improving. Third, there exists a range of love-of-variety $q$ such that small and large increases in the fixed cost increase long-run output and welfare, while medium-sized recessions reduce them. This is driven by a non-monotonic relationship between recession depth and the strength of cleansing effects.

\paragraph{Love-of-Variety} In our model, while firms are heterogeneous in their productivity, they are homogeneous in their contribution to the variety externality. We relax this assumption in Appendix \ref{app:matsuyama}. We characterize the behavior of an economy where firms, upon entry, draw two distinct features: their productivity $z$ and their contribution to the variety externality $\xi$. These can be drawn from a joint distribution with arbitrary correlation. We show that our baseline economy is a special case of this more general setting where the $z$ and $\xi$ draws are independent. Away from this special case, we show in Proposition \ref{prop:het_ext} that the intensity of the correlation between productivity and externality contributions determines the long-run effect of recessions. Intuitively, if productivity and the externality contribution are positively correlated, recessions induce a smaller variety loss than if they are independent: the same cleansing effects that improve average productivity can reduce the cost of the reduction in varieties. 
Conversely, if the correlation between idiosyncratic externality and productivity is negative, the long-run effects of recessions are more negative than under the independence case. In this scenario, the cleansing effect induces a further loss as it negatively selects firms in terms of their externality contributions.

\paragraph{Entry, Exit, and Selection} 

In our model, recessions generate \textit{selected} exit followed by \textit{unselected} entry during recoveries. This property follows from our assumption that firms do not know their productivity upon paying the entry cost. This modeling assumption provides us with analytically tractable result, which qualitatively carries through in less extreme versions of our model. Suppose, for example, that all firms have some uncertainty about their true productivity $z$. As long as firms that have already produced have less uncertainty than prospective entrepreneurs, our results qualitatively hold.
These results would also arise in a context in which potential entrants know their productivity but face uncertainty about the fixed cost.

A related point regards to our assumptions on exit. We have modeled exit as a forced event, triggered by negative current profits $\pi(z, m) < 0$. This assumption can be interpreted as firms defaulting due to working capital constraints and poorly developed financial markets. Suppose otherwise that firms could borrow against future profits. Then, they would make exit decisions by comparing the net present value of fixed cost $f^c/(1-\beta)$ to the NPV of gross profits $\tilde\pi/(1-\beta)$. Under this reading, firm exit is not driven by a constraint, but by a rentability consideration. Due to the unanticipated nature of the shocks we consider, this formulation is equivalent to our static exit choice, as firm owners do not expect parameter changes. In Appendix \ref{ext:fwd-look-exit}, we relax the assumption that firm owners do not expect any future changes in fixed cost, allowing them to forecast the reversion of fixed cost to pre-cycle levels. This extension is consistent with our propositions and does not alter results beyond a change in notation. 

An additional important assumption in our framework is the absence of exit in the steady state. This is often generated by either idiosyncratic fluctuations in productivity \citep{hopenhayn1992entry} or some exogenous exit rate \citep{melitz_impact_2003}. In either of those scenarios, the steady-state is not history dependent. As a consequence, all our results apply along the transition. We return to this point in the quantitative version of our model. 

\paragraph{Aggregate Transitional Dynamics} Proposition \ref{proposition:PE-cycles} provides a characterization comparing the economy before the shock to its long-run post-shock equilibrium. This allows us to isolate the Schumpeterian argument that cleansing recessions are beneficial in the long run. It also conveniently allows us to sidestep the full dynamic behavior of the economy along the transition. We provide an analytical characterization of the transition dynamics in Appendix \ref{app:dynamics} under the assumption that the entrant productivity distribution is Pareto. In Propositions \ref{prop:smooth_reversion_pareto_sequential_fc} and \ref{prop:smooth_reversion_pareto_smooth_fc}, we show that when fixed costs revert at a constant rate, the transition features a constant flow of entrants. The cutoff smoothly converges to its long-run level from above, while the number of firms converges from below. We show that fixed cost increases that slowly subside induce stronger selection effects as shown in Figure \ref{fig:tent}.

\paragraph{Idiosyncratic Fluctuations}
 Since our focus is on the consequences of aggregate shocks, we have considered an economy where firms are not subject to idiosyncratic fluctuations \'a la \cite{hopenhayn1992entry}. Our model can be readily extended to account for idiosyncratic movements in productivity. We extend our results to such an economy in Appendix \ref{app:stoch_prod}.

\paragraph{Varieties and the Number of Firms} In our model, varieties and firms are, by assumption, the same thing. However, our results readily generalize to economies in which firms produce multiple products. We assume that varieties produced by the same firm are more substitutable than varieties across firms, and we can then distinguish between LoV for products within and across firms. We analyse this extension and how our main results change in Appendix \ref{app:multiprod}. Our main conclusions continue to hold.

\paragraph{Fixed Costs Units} In the model discussed in this section, we maintained that the unit of fixed and entry costs is labor. This provides a natural benchmark based on the \cite{dixit1977monopolistic,dhingra2019monopolistic} efficiency results, to which we return in our policy discussion. An alternative assumption is that these costs are in units of the final output good. This has important implications as entry then generates additional external effects \citep[see][]{barro_economic_2004}. We consider this case in Appendix \ref{ext:fixed-cost-in-final-good} by extending our framework to one in which fixed costs are paid in units of a new good, made from a Cobb-Douglas aggregate of output and labor. This framework nests our baseline model and one of output good fixed costs as special cases. We show that our main results carry through. 
\section{Policy} \label{sec:policy}
In this section, we first study a social planner problem for economies with incumbents. The planner chooses the allocation of labor between entry costs, fixed costs, and each individual firm's production, controlling the entry mass $E_{\tau, t}^{SP}$ and the cutoff $\underline z_{\tau, t}^{SP}$ at time $t$ and phase $\tau$. We show that the government can use subsidies to achieve the first-best allocation with a single instrument: a tax/subsidy on $f^e/f^c$ in eq. (\ref{eq:cutoff}),  financing the intervention using a lump-sum tax paid by the households. In the knife-edge case of CES, the market allocation is efficient. Next, we consider the optimal policy problem along the business cycle, as in Section \ref{sec:business_cycles}. We find that the sign of the intervention is the same as the steady state one, but its optimal magnitude is larger: if the planner finds it optimal to subsidize firms in steady state, it should increase the size of the subsidies during recessions.

\subsection{Planner Entry and Exit Choices}
We consider a myopic planner maximizing current output. We show in Appendix \ref{app:proofs} that the market's allocation of labor across firms is efficient. Hence, we allow the planner to choose directly the number of entrants $E^{SP}$ and the cutoff productivity $\underline z^{SP}$. This is equivalent to choosing how much labor to allocate for production, entry costs, and fixed costs payments. Next, we decentralize allocation via a tax/subsidy. The planner's problem is
\begin{align}\label{eq:SP_obj}
    \max_{E_t^{SP},\underline z_t^{SP}} \quad Y_t^{SP} &={(M_t^{SP})}^{q-\frac{1}{\sigma-1}}L^{p,SP} \left[\int  z^{\sigma - 1} m_t^{SP}(z)\d z\right]^{\frac{1}{\sigma-1}} \qquad \text{s.t.} \\\nonumber
    \bar L &\geq L_t^{p,SP} + f^e E_t^{SP} + f^c M_{t}^{SP},\\\nonumber
    m_t^{SP}(z) &= (m_{t-1}(z) + E_t^{SP}\mu^E(z)) \mathbb{I}_{\{ z \geq \underline z_t^{SP}\}}(z) ,\\\nonumber
    \underline z_t^{SP}, E_t^{SP}&\geq 0.
\end{align} 
The myopic planner maximizes $Y_t^{SP}$,  subject to the labor market clearing condition and the law of motion of the firms' productivity distribution. We solve the full dynamic policy problem in the quantitative version of the model. 
For the analytical characterization, we make the following assumptions:
\begin{assumption}\label{as:policy}
    We assume the following holds:
    \begin{enumerate}[label=\alph*)]
        \item The distribution of active firms is a truncation of the entrants' distribution: $m(z)\propto \mu^E(z)\mathbb{I}_{\{z\geq\underline z\}}$ for some $z$.\label{1a}
        \item The expected relative market share of entrants decreases with higher truncations: \\$\mathbb{E}_{\mu^E}\left[(z/\underline z)^{\sigma-1}|z\geq \underline z\right]$ decreases in $\underline z$.\label{1b}
    \end{enumerate}
\end{assumption}

Assumption \ref{as:policy}\ref{1a} amounts to an assumption on the history of the economy. It states that the distribution of active firms is proportional to the entrants' productivity distribution. This is the case whenever the economy has not undergone any recession yet, or it is in the middle of a fixed cost crisis. 
Assumption \ref{as:policy}\ref{1b} is an assumption on the productivity distribution of entrants. It posits that as we truncate the distribution from the left, the expected market share weakly decreases. This is satisfied by many commonly used distributions.\footnote{Distributions satisfying the assumption are, for example, Pareto, Exponential, Weibull (for shape parameter $k>1$), Burr, and Fréchet (for shape parameter $\xi\leq1$). It can be interpreted as requiring that as the condition for survival becomes more stringent, the productivity of marginal exiter and new entrant become more similar on expectation. A sufficient condition for $\E_{\mu^E}\left[\left( z/\underline z\right)^{\sigma-1} \big\lvert  z \geq \underline z\right]$ to be decreasing in the cutoff $\underline z$ is that the elasticity of the anti-cumulative corresponding to $\mu^E$ with respect to the truncation $\underline z$, $\varepsilon_{1-F}(\underline z)$, is globally larger than $-1$, i.e. $\varepsilon_{1-F}(\underline z)\geq -1, \; \forall \underline z \geq 0$.} While we can relax these assumptions for specific results, we maintain Assumption \ref{as:policy} throughout this section to provide a consistent characterization. 

Proposition \ref{prop:SP_solution} characterizes of the solution to the planner problem. 

\begin{prop}[Social Planner Allocation]\label{prop:SP_solution}
    The following holds

\begin{enumerate} [label=\Roman*)]
    \item The planner solution coincides with the market allocation if and only if $q = q^{CES}$.
    \item In a state of entry,
    \begin{enumerate} [label = \alph*)]
        \item The optimal number of entrants $E^{SP}$ is strictly increasing in $q$.
        \item The presence of incumbents with higher average productivity than entrants decreases the socially optimal number of entrants and, therefore, the number of firms. It does not affect the social planner's steady-state cutoff.
        \item If there are no incumbents, the cutoff $\underline z^{SP}$ is decreasing in $q$.
    \end{enumerate}

\item The socially optimal allocation can be achieved in a decentralized equilibrium where the planner sets a subsidy/tax to the fixed cost $\theta^c$ such that firms pay $f^c(1-\theta^c)$, financed by a lump-sum tax on the household. The optimal subsidy/tax $\theta^c$ satisfies
\begin{align}
            1-\theta^c(\underline z^{SP}) = \left[[q(\sigma - 1) - 1]\left(\E_{\mu^E}\left[\left(\frac{ z}{\underline z^{SP}}\right)^{\sigma-1} \mid  z \geq \underline z^{SP}\right]\right)^{-1}+1\right]^{-1},
\end{align}\label{prop:SP_subsidy}
where $\theta^c\gtreqless 0 \iff q\gtreqless q^{CES} \equiv1/(\sigma-1)$. 
    \end{enumerate}
\end{prop}
The result in Proposition \ref{prop:SP_solution} extends previous insights from \cite{spence1976product,dixit1977monopolistic,mankiw1986free,parenti2017toward, bilbiie2019monopoly, dhingra2019monopolistic} and \cite{matsuyama2020does} to economies with firm heterogeneity and incumbents. When the planner is constrained by the presence of already existing incumbent firms, it optimally chooses the same cutoff and mass of entrants as in the market equilibrium if and only if the intensity of love-of-variety is that of CES preferences. 
If LoV is larger, there are inefficiently few firms operating in the market. These potential inefficiencies represent steady-state distortions. For any LoV, the optimal number of entrants is decreasing in the mass incumbent. In our setting, where incumbents are, on average, more productive than successful entrants, the planner optimally includes fewer firms the larger the mass of incumbents. Intuitively, the social planner internalizes the trade-off between allocating more labor to the most productive firms versus producing more varieties through a less productive entry margin. 
Without incumbents, LoV affects the cutoff in an intuitive way: a higher $q$ leads the social planner to decrease $\underline z$ and induce further entry. 

Finally, Proposition \ref{prop:SP_solution}.\ref{prop:SP_subsidy} characterizes the optimal policy intervention. A subsidy/tax to fixed costs is sufficient to restore efficiency. Importantly, whether the planner wants to subsidize or tax the presence of firms in the economy depends exactly on whether $q$ is larger or smaller than $1/(\sigma-1)$. When $q>1/(\sigma-1)$, there are inefficiently few firms, so the planner optimally chooses a subsidy $theta^c>0$.

Next, we characterize the optimal policy intervention during business cycles.

\subsection{Optimal Policy through the Cycle}\label{sec:optim-pttc}
We compare the planner solutions and market allocations subject to a cycle like in Section \ref{sec:business_cycles}. Figure \ref{fig:policy_cycle} illustrates the comparison. 
\begin{figure}[htbp]
    \centering
    \includegraphics[width=\linewidth]{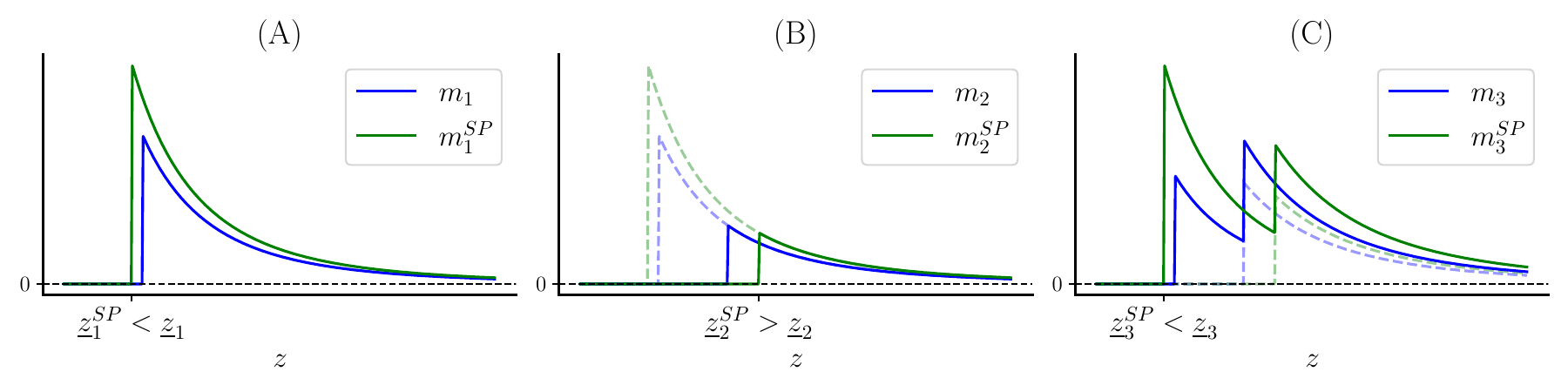}
    \caption{Evolution of $m$ through the business cycle for the market and social planner allocations, assuming that $q > q^{CES}$.}
    \label{fig:policy_cycle}
\end{figure}
Suppose that $q>1/(\sigma-1)$, then in the original allocation there are too few firms and the planner induces additional entry by implementing the optimal policy from Proposition \ref{prop:SP_solution}.\ref{prop:SP_subsidy}. Before the recession, the planner allocation features more firms, as shown in Panel (A). Suppose that the two economies experience an increase in the fixed cost. The planner's economy, by virtue of the larger measure of firms, is effectively hit harder by the rising fixed costs. As a consequence, the cutoff experiences a much larger rightward shift than in the laissez-faire allocation. Importantly, the planner finds it optimal to increase the subsidy during the crisis, which counteracts this force, saving some firms from exit.  Yet, the realized cutoff may be to the right of the laissez-faire economy, because the planned economy enters the crisis with a thicker right tail. When the rise in fixed costs subsides, firms repopulate the economy. Since the planner is instating the steady-state subsidy from Proposition \ref{prop:SP_solution}, more firms enter, and the post-crisis distribution features more varieties than the market allocation would deliver. 
We summarize the optimal policy response in the three phases of the cycle.

\begin{prop}[Optimal Policy over the Cycle]\label{prop:steady-state-subsidy}
    {{The optimal fixed cost taxes/subsidies are counter-cyclical in absolute value:  $|\theta_1^c| \leq |\theta_2^c|$ and $|\theta_3^c| \leq |\theta_2^c|$.}}
\end{prop}

The social planner's optimal policy is to increase the magnitude of its steady-state interventions, determined by Proposition \ref{prop:SP_solution}.\ref{prop:SP_subsidy}. If the economy features $q>1/(\sigma-1)$ then the planner's steady-state policy is a subsidy $\theta_1>0$, in which case during period 2 at the peak of the crisis, the planner optimally increases the subsidy to $\theta_2>\theta_1$. The intuition is that, during the recession, firms' exit permanently lowers the number of available varieties. To avoid this long-run loss, the planner intervenes by softening the blow and avoiding the exit of some of the firms.
As the recession subsides, the planner reduces the subsidy to a lower level $\theta_3<\theta_2$. Conversely, if $q<1/(\sigma-1)$, the optimal steady-state policy is a tax $\theta_1<0$ and the planner increases this tax during the recession. In this scenario, the planner takes advantage of the crisis as a cleansing moment and uses it to \textit{pick winners} in the long-run steady state. 
We conclude that the social planner implements an optimal cyclical policy, which can be inferred from their steady-state behavior. 

Given the importance of \textit{love-for-variety} $q$ and the price elasticity of demand $\sigma$ in shaping the optimal intervention, in the next section, we provide a new approach to identify them separately.
\section{Estimation}\label{sec:empirics}

We set out to estimate our two key parameters: $q$ and $\sigma$, governing love-for-variety and substitutability across varieties. We estimate $q$ and $\sigma$ in a cross-industry version of our model for all the countries available in the World Input-Output Database \citep{timmer2015illustrated}, henceforth WIOD, as well as $\sigma$ in the Spanish administrative firm-level data (SABI). Since we consider a large number of different industries, $i \in\mathbf{I}$, over time, $t\in \mathbf{T}$, and across countries, $c \in \mathbf{C}$, we study the partial equilibrium framework of Section \ref{sec:business_cycles}. We discuss the identification of $\sigma$ and $q$ in turn.

\subsection{Data}
We use two main data sources: the World Input-Output Database (WIOD) and the Spanish administrative balance sheet data (SABI). We use the former to estimate our key parameter of interest and the latter for validation and calibration.
\paragraph{Input-Output Data}
The primary data source for the identification of $q$ is the World Input-Output Database (WIOD) 2016 release, see \cite{timmer2015illustrated}. It contains the Input-Output structure of sector-to-sector flows for $C=43$ countries from 2000 to 2014 yearly. The data is available at the 2-digit ISIC rev-4 level. The number of sectors in WIOD is $I=56$, which amounts to 6,071,296 industry-to-industry flows and 108,416 industry-to-country flows for every year in the sample. The full country and sector coverage is reported in Appendix \ref{app:wiod}.

The World Input-Output Table is an $(I \times C)$ by $(I \times C)$ matrix. Each element $Z_{cd}^{ij}$ denotes the sales of industry $i$ in country $c$ to industry $j$ in country $d$ for intermediate input use. Additionally, the data includes an $(I\times C)$ by $C$ matrix of final use. Each element $F_{cd}^i$ is the value of sales of industry $i$ in country $c$ sold to and consumed by households, government, and non-profit organizations in destination $d$. Denote $F_c^i=\sum_dF_{cd}^i$ the value of output of sector $i$ in country $c$ consumed in any country in the world. The total value of sales is $S_c^i=F^i_c+\sum_d\sum_j Z_{cd}^{ij}$. 

WIOD also includes industry-level deflators. Importantly, the data is also provided at previous-year prices, which allows us to concatenate the data to obtain changes in output from changes in sales at different prices, as we explain in Appendix \ref{sec:salesVoutput}.

We complement the I-O table with the data in the Socio-Economic Accounts, which include information on input usage for all sector-country-year tuples. In particular, we make use of the information on labor used in production by these industries.

\paragraph{Spanish Administrative Data (SABI)} The Sistema de Análisis de Balances Ibéricos (SABI) is a firm-level database maintained by Bureau van Dijk and Informa D\&B. It contains standardized financial statements and company information for Spanish firms, drawing on official sources such as the Boletín Oficial del Registro Mercantil, the Mercantile Registry, and stock exchange filings. The database provides historical annual accounts, including balance sheets and income statements. Our sample covers the period 2000–2018 at an annual frequency and includes 768,296 unique firms, with about 384,571 firms per year on average.
\subsection[Identification of sigma]{Identification of $\sigma$}
In our model, $\sigma$ can be identified directly from the profit share since $\pi/R = 1/\sigma $. However, the WIOD data is built such that the sum of material, labor, and capital expenditure is equal to gross output. This implies lumping together economic profits with the cost of capital. For the purpose of identifying $\sigma$, we make the conservative assumption that all capital expenses are, in fact, profits. Hence, we compute this profit rate as gross output net of worker compensation and intermediate input expenditures divided by gross output: $1/\sigma = (GO-COMP-II)/GO$ for all industry-country-year triplets. Clearly, this implies that we are overstating the profit share in our measurement. As a consequence, our estimated $\sigma$ is a lower bound for the true elasticity of substitution and, therefore, the implied $q^{CES}$ is larger than the true $q^{CES}$. Alongside the WIOD-based measurement, we compute $\sigma$ using the universe of firms in the Spanish SABI database. We apply the same conservative treatment of capital costs as profits, which similarly leads to an upward-biased value of $q^{CES}$.

\subsection[Identification of q]{Identification of $q$}

The estimation of love-for-variety has been extremely elusive for decades. The main difficulty is that it requires exogenous variation in the number of available varieties to estimate its effect on output. This, naturally, requires taking a stance on what a variety is in the data. Our starting point to estimate $q$ is different: we use a structural relation of the model, linking changes in expenditure to changes in output through $q$. We start by taking logs of the aggregate production function in eq. (\ref{eq:gdp}), and differentiating with respect to $\log \mathcal I$ to obtain
\begin{align}\label{eq:empirical_logY}
    \frac{\partial\log Y}{\partial \log \mathcal I}    &= q \frac{\partial\log M}{\partial \log \mathcal I}+ 1 + \frac{\partial\log \bar z}{\partial \log \mathcal I}.  
\end{align}
A change in nominal expenditure of downstream customers on the industry, $\mathcal I$, induces a 1-for-1 direct increase in output as well as two indirect effects. First, higher expenditure induces entry, which in turn increases $M$. This effect generates a change in output that depends only on $q$. Second, by inducing entry, higher $\mathcal I$ changes the average productivity of firms, $\log \bar z$, provided entrants/exiters do not have the same average productivity as incumbents. Since in equilibrium production labor $L^p$ is proportional to $\mathcal I$, the same equations and all following results hold for changes in $L^p$, too. Next, we note that the economy does not respond symmetrically to expansions and contractions. In an expansion, the average productivity of firms in the economy changes if entrants are more or less productive than incumbents on average. Differently, a reduction in $\mathcal I$ always increases average productivity, since exit occurs from the bottom of the distribution. The following Lemma characterizes the elasticities given a \textit{small increase in} $\mathcal I$. 

\begin{lemma}[Elasticities In Output Equation]\label{lem:emp_elasticities}
    Suppose, entrants and incumbents are distributed according to $m^E(\, \cdot\, ) \propto \mu^E(\, \cdot \g z \geq \underline z)$ and $m^I(\, \cdot\, )=E^I\mu^E(\, \cdot \g z \geq z^*), E^I\geq 0$, respectively.%
    \footnote{%
    Here, $E^I$ is the number of `historical' entrants, i.e., the mass of incumbent firms that have tried, and of which $E^I \int_{z^*}^\infty \mu^E(z) dz$ have succeeded in entering the market.} %
    Then, the right-derivatives $\frac{\partial_+\log M}{\partial_+ \log \mathcal I}$ and $\frac{\partial_+\log \bar z}{\partial_+ \log \mathcal I}$ are given by:
    \begin{align}
       & \frac{\partial_+\log \bar z}{\partial_+ \log \mathcal I} = \text{\footnotesize$\frac{1}{\sigma-1}$} \left( 1 -  \frac{\partial_+\log M}{\partial_+ \log \mathcal I}\right),\: \text{ and}\:&
        \frac{\partial_+\log M}{\partial_+ \log \mathcal I} = \frac{E_{\g E^I=0}p_E(\underline z)}{E_{\g E^I=0}p_E(\underline z) + E^I\, k(\underline z, z^*)},
    \end{align}
    where $p_E(\underline z):=\int_{\underline z}^\infty \mu^E(z) dz$, $E_{\g E^I=0}$ is the number of entrants if there are no incumbents, and $k$ is a function satisfying $k(\underline z,z^\ast) \lesseqgtr 0 $ if $\underline z \lesseqgtr z^\ast$.
\end{lemma}

For identification, we make Assumption \ref{ass:entrant_incum_dist}, which states that the entrant distribution and the incumbent distribution have the same cutoff. This has three interpretations. Either i) the industry has not undergone a sequence of selection-inducing fixed-cost recessions yet. Or ii) the industry is in a long-run equilibrium with exogenous exit (see Proposition \ref{prop:ext_idrisk_dyn}), such that the cleansing effects from the last recession have fully subsided. Or, finally, iii) entrants `learn' from incumbents, such that their technologies, conditional on clearing the cutoff, $\underline{z}$, are no worse than those of existing businesses.

\begin{assumption}[Entrant and Incumbent Distribution]\label{ass:entrant_incum_dist}
    $\underline z = z^*.$
\end{assumption}

Under Assumption \ref{ass:entrant_incum_dist}, entry induced by shocks such as income expansions that leave $f^c/f^e$ and hence the cutoff unchanged do not affect average productivity, $\bar z$. As a consequence, we have the following identification result.

\begin{prop}[Identification of $q$]\label{prop:identification}
    Under Assumption \ref{ass:entrant_incum_dist},
    \begin{align}
        \frac{\partial_+\log Y}{\partial_+ \log X} = 1+q, \quad X\in \{\mathcal I, L^p\},
    \end{align} and in the regression equation %
    \begin{align} \label{eq:regression_zbar}
        \Delta\log Y_{i,t} = \beta \Delta \log X_{i,t} +\gamma\Delta \log\bar z_{i,t}+\epsilon_{i,t}, \quad X_{i,t}\in \{\mathcal I_{i,t}, L^p_{i,t}\},
    \end{align} %
    we have %
        $\Delta \log X_{i,t} > 0 \implies \Delta \log\bar z_{i,t} = 0\ \text{ and }\ \beta = 1+q.$
\end{prop}

This result allows us to identify $q$ as the determinant of the elasticity of output to income changes. Underlying this result is that, by the free entry condition, an increase in income $\mathcal I$ increases the number of firms proportionally. The change in $M$ induces a direct change in output with elasticity $q$. In this sense, we estimate $q$ as governing the intensity of external returns to scale over and above the technological returns, which in our model are equal to 1 by the constant returns to scale assumption. Importantly, this only holds for positive changes in income and demanded labor since any negative change necessarily induces exit from the bottom of the distribution and, therefore, a response in $\log\bar z$. We conclude that we can estimate love-of-variety $q$ provided some exogenous increase in employment $L^p$ or expenditure $\mathcal I$ from %
\begin{align}
    \label{eq:regression}
    \Delta\log Y_{i,t} = \beta \Delta \log X_{i,t} +\epsilon_{i,t}, \quad X_{i,t}\in \{\mathcal I_{i,t}, L^p_{i,t}\}.
\end{align}%
Proposition \ref{prop:identification} enables us to estimate eq. (\ref{eq:regression}) with an instrument inducing \textit{positive variation} in $L^p$ or $\mathcal I$, while disregarding the average productivity term. We use this identification strategy as our benchmark and discuss possible threats to identification in the next section. In particular, how eq. (\ref{eq:regression_zbar})  extends to a setting with multiple factor inputs, non-constant returns to scale, and measurement error. Additionally, we discuss how failure of Assumption \ref{ass:entrant_incum_dist} still allows us to identify the sign of $q - q^{CES}$, given that we can identify $q^{CES} = \text{\footnotesize$\frac{1}{\sigma-1}$}$, or a lower bound thereof.

\subsection{Discussion of Identification}
\paragraph{Misspecified Production Function} Part of our identification result in Proposition \ref{prop:identification} relies on the assumption that labor is the only factor of production and that it enters the production function linearly. The constant returns to scale (CRS) underlie the ``1'' in our estimated coefficient of $\beta 
= 1+q$. Clearly, these assumptions are counterfactuals in empirical applications.  Our identification strategy readily extends to settings in which firms have multiple inputs combined under CRS. The ``1'' we identify in our regression should be interpreted as the returns to the optimal input bundle (in our model, trivially made only of labor). Importantly, in the absence of relative input price changes, homogeneity of the production function implies that factor demands are proportional to one another. As a consequence, estimating eq. (\ref{eq:regression}) on a single input (labor) does not affect our finding. Hence, with multiple inputs under CRS, we would still obtain $\beta=1+q$ from estimating eq. (\ref{eq:regression}). If the production function had non-constant returns to scale, we would have an upward bias in our estimated $q$ if the true data-generating process features increasing returns, and a downward bias if it features decreasing returns to scale. We summarize both these results in Remark \ref{identification_homog}.
\begin{remark}[Identification for Homogeneous Production Functions]\label{identification_homog}
    Proposition \ref{prop:identification} holds more generally for any production function that is homogeneous of degree 1. If the production function is instead homogeneous of degree $\chi$, the right derivative with respect to some input $X$ estimates
    \begin{align*}
        \frac{\partial_+\log Y}{\partial_+\log X} = \chi+q.
    \end{align*}
\end{remark}
Recent evidence estimating returns to scale points to small deviations from CRS, with returns to scale of at most $1.05$ \citep[see for example][]{chiavari2024customer}. This upper bound implies a maximal upward bias of just $0.05$ in our estimate of $q$, leaving our findings unaltered.

Importantly, our model-based interpretation of $\beta = 1+q$ also relies on the absence of alternative sources of external returns. Suppose otherwise, that the aggregate production function in equilibrium is given by $Y^\delta = M^q L \bar z$, where $\delta<1$ governs the extent of external returns that are not love-of-variety, such as agglomeration externalities or knowledge spillovers. Then, our estimate is $\beta = \frac{1+q}{\delta}$. As a consequence, we have a biased estimate $\hat q-q = (q+1)\frac{1-\delta}{\delta}>0$. This allows us to bound the extent of increasing returns that would undo our qualitative results. We provide these calculations together with our findings in Section \ref{sec:emp_results}.

\paragraph{Mismeasurement of Price Indices} An important requirement in our estimation approach is the observation of output, rather than sales. In our empirical analysis, we use deflated sales to obtain output. This relies on the correctness of the industry deflators as composites of the standard CES-implied price index $P^{CES}$ and the variety effect $M^{\frac{1}{\sigma-1}-q}$. A natural question when taking our approach to the data is whether the WIOD deflators incorporate the variety effect---and, accordingly, whether and under what assumptions $q$ is identified, and how large the bias can be when one recovers $q$ from Proposition \ref{prop:identification}, i.e. $\hat q = \hat \beta - 1$. We consider this in Remark \ref{rem:ident_pindex}. To fully characterize this problem, which is not always defined when deflators are incorrectly measured, we consider an economy where fixed and entry costs are produced by a Cobb-Douglas bundle of output (share $\alpha$) and labor (share $1-\alpha$) as described in Appendix \ref{ext:fixed-cost-in-final-good}.\footnote{We adopt this formulation because when deflators are incorrectly measured, the limit of our estimator is not informative on $q$, and, therefore, does not have a bounded bias. Using this more general model allows us to characterize the limit behavior of the bias and therefore find its sign.} Denote with $\sim$ objects measured or estimated using the incorrect deflators, which do not account for the variety effect. For example, call $\tilde Y$, sales deflated by $P^{CES}$ rather than $P$. 

\begin{remark}[Mismeasurement in Price Indices]\label{rem:ident_pindex}
Under Assumption \ref{ass:entrant_incum_dist} and positive demand shocks, the regressions
\begin{align}
\Delta\log Y_{i,t}=\beta\,\Delta\log X_{i,t}+\epsilon_{i,t},\qquad
\Delta\log \tilde Y_{i,t}=\tilde\beta\,\Delta\log X_{i,t}+\epsilon_{i,t},\qquad
X_{i,t}\in\{\mathcal I_{i,t},L^p_{i,t}\}
\end{align}
identify
\begin{align}
\beta \;=\; 1-\frac{q}{\alpha q-1},\qquad
\tilde\beta \;=\; 1-\frac{1}{(\sigma-1)(\alpha q-1)}.
\end{align}
Hence, given an estimate $\beta$ or $\tilde \beta$, and $\sigma$, the implied love-of-variety parameter $q$ is
\begin{align}
q(\alpha)=\frac{1-\beta}{\alpha(1-\beta)-1},\qquad
\tilde q(\alpha)=\frac{1}{\alpha}\!\left[1+\frac{1}{(\sigma-1)(1-\tilde\beta)}\right].
\end{align}
Taken together, these relationships deliver two results that allow us to bound our estimate against mismeasurement of price indices.
    \begin{itemize}
        \item[i)] When empirical deflators account for varieties, $q$ is identified for any $\alpha\in[0,1]$. Furthermore, the bias is given by $-\,\frac{\alpha(1-\beta)^2}{\,\alpha(1-\beta)-1\,}$, for $\alpha(1-\beta) \neq 1$.
        \item[ii)] When deflators do not account for varieties, $q$ is identified for any $\alpha\in(0,1]$ conditional on $\sigma$. Furthermore, the bias is given by $\frac{\,\alpha(\sigma-1)(\tilde\beta-1)^2-(\sigma-1)(\tilde\beta-1)+1\,}        {\alpha(\sigma-1)(\tilde\beta-1)}$, for $\alpha>0$ and $\tilde\beta\neq 1$. 
    \end{itemize}


\end{remark}

In summary, identification holds for any $\alpha$ when deflators are correctly reported, and for any $\alpha>0$ (given $\sigma$) when deflators fail to incorporate the variety effect. If one recovers $q$ using the result in Proposition \ref{prop:identification}, in anticipation of the empirical results, the bias is positive but small in the correctly-measured deflator case, while in the mismeasured case the bias is negative for our estimate of $\beta$ and $\sigma$. Our findings remain unaltered when taking into account the bias under both deflators. We provide these bounds together with our main estimates in Section \ref{sec:emp_results}.

\paragraph{Estimates without Assumption \ref{ass:entrant_incum_dist}} Our identification strategy requires that there is limited selection in incumbents. However, even if Assumption \ref{ass:entrant_incum_dist} fails entirely and $\underline z < z^*$ holds, we are still able to identify the sign of $q - q^{CES}$ and determine whether recessions (in partial equilibrium) are welfare-improving. This is formalized in the following proposition.

\begin{prop}[Identification of $\mathrm{sgn}(q - q^{CES})$]\label{prop:sign_identification} 
    Suppose that $q^{CES} = \text{\footnotesize$\frac{1}{\sigma-1}$}$ is identified. Let $\beta$ be the regression coefficient on $\Delta \log X_{i,t}$ in eq. (\ref{eq:regression}).     Then, the bias of $\hat q =  \beta -1 $ is given by 
    \begin{align}
        bias_+ =\hat q - q = \left(\text{\footnotesize$\frac{1}{\sigma-1}$} - q\right) \left( 1 -  \frac{\partial_+\log M}{\partial_+ \log \mathcal I}\right),
    \end{align}
    assuming an expansion in $X$. Furthermore, for $\underline z < z^*$, %
    \begin{align}
        \mathrm{sgn}\left( \hat q- \text{\footnotesize$\frac{1}{\sigma-1}$}\right) = \mathrm{sgn}\left( q - \text{\footnotesize$\frac{1}{\sigma-1}$}\right),
    \end{align}
    Which implies that even without Assumption \ref{ass:entrant_incum_dist}, we can identify the sign of $q-q^{CES}$.
\end{prop}

\paragraph{Firms, Establishments, and Varieties} As a final remark, note that our approach circumvents a fundamental problem in mapping the model to the data: the definition of a variety. In principle, if we assumed that establishments/firms only produced a single variety, we could use the number of firms or establishments as a direct measure of $M$. This mapping is fragile as establishments/firms may produce multiple products (as we highlight in the extension section), or new establishments of firms, such as factories, stores, or offices, may indicate larger quantities of existing rather than new product varieties. In such a case, we would need to observe both the number of firms and the number of products per firm. Our empirical strategy sidesteps this issue since it relies on the observation that $q$ governs the intensity of external aggregate returns to scale. Our approach is vulnerable to the criticism that identification relies on the absence of other sources of spillovers. Conversely, alternative approaches to the estimation of love-for-variety require a precise definition of variety, for example \cite{baqaee2023supplier}. We return in the results section to the comparison of our findings.

\subsection{Implementation -- Instrumental Variable} For each of the WIOD industries, we would like a plausibly exogenous shift in expenditure. We build this as in \cite{ferrari2023inventories} through a shift-share, aggregating destination-specific shifters through destination shares. Formally, an industry $i$ in country $c$ at time $t$ is assigned a change in demand equal to
\begin{align}
    \sum_j \xi^i_{cd}\Delta\log F_{dt},\label{eq:ss}
\end{align}
where $\xi_{cd}^i$ represents the fraction of the value of output of industry $i$ in country $c$ consumed directly or indirectly in destination $d$ in the first sample period and $\Delta\log F_{dt}$ is the growth rate of final consumption expenditure in destination $d$ at time $t$. The exposure share $\xi_{cd}^i$, assumed to be time-invariant, includes direct sales from the industry to consumers in $d$ and output sold to other industries, which eventually sell to consumers in $d$.

The input-output structure of the data allows a full account of these indirect linkages when analyzing sales composition. Formally, define the share of sales of industry $i$ in country $c$ that is consumed by destination $d$ as
\begin{align}
\xi_{cd}^i=\frac{F_{cd}^i+\sum_j\sum_e a_{ce}^{ij}F_{ed}^j+\sum_j\sum_e\sum_k\sum_f a_{ce}^{ij}a_{ef}^{jk}F_{fd}^k+...}{S_c^i},
\label{shareseq}
\end{align}
where $a_{cd}^{ij}$ is dollar amount of output of sector $i$ from country $c$ needed to produce one dollar of output of sector $j$ in destination $d$, defined as $a_{cd}^{ij}=Z_{cd}^{ij}/S_d^j$. The first term in the numerator represents sales from sector $i$ in country $c$ directly consumed by $d$; the second term accounts for the fraction of sales of sector $i$ in $c$ sold to any producer $j$ in country $e$ that uses $c,i$ as input and then sells to destination $d$ for consumption. The same logic applies to higher-order terms. By definition $\sum_d\xi_{cd}^i=1$.

For the shifters, we estimate the common component of all industries selling to a given destination $d$ at time $t$. For each industry, we iteratively exclude flows from the same country and in the same sector
\begin{align}
\Delta \log F_{edt}^j=\eta_{dt}(c,i)+\nu_{edt}^i\quad e\neq c, j\neq i.
\label{fixedeffects2}
\end{align}
This boils down to identifying our effect of interest through changes in foreign demand on other products; we refer to \cite{ferrari2023inventories} for an extended discussion of the identification and briefly discuss the key intuition. The idea behind our instrumental variable approach is that each industry is small relative to each destination country. As a consequence, variations in aggregate consumption in destinations are quasi-randomly assigned to each producing industry through pre-existing network linkages. The exclusions in eq. \eqref{fixedeffects2} state that we only use variation from different products and different countries to reduce concerns of reverse causality. 
Formally, our instrument is given by
\begin{align}
    \eta_{ct}^i = \sum_d \xi^i_{cd} \hat \eta_{dt}(c,i).
    \label{IV-final}
\end{align}
Following Proposition \ref{prop:identification}, we estimate  
\begin{align}
        \Delta\log Y^i_{ct} = \beta \Delta\log \hat{X}_{ct}^i+\epsilon_{ct}^{+i}, 
        \label{empiricalmodel}
\end{align}
Where $\Delta\log Y^i_{ct}$ is the growth rate of output sold by industry $i$ in country $c$ at time $t$, which we compute as described in Appendix \ref{sec:salesVoutput}. $\Delta\log \hat{X}_{ct}^i$ is the instrumented change in either expenditure $\mathcal I$ or production labor $L^p$ using $\eta_{ct}^i$ as instrument.

\subsection{Results}\label{sec:emp_results}
\paragraph{Estimates of $q^{CES}$} We start with a description of our estimated $\sigma$. Based on the WIOD data, we obtain the distribution of profit rates across industries summarized in Table \ref{tab:profit_sigma_qces_joint}. The profit rate ranges between 0\% and 69.9\% with a median of 17.2\%. The corresponding $\sigma$ distribution is between 1.4 and 10, with a median of 5.8, in line with the estimates in the literature \citep[see][for a review]{anderson2004trade}. We are interested in the $q^{CES}$ associated with this distribution of elasticities. We find $q^{CES}$ to be between 0.11 and 2.32, with a median of 0.21. In the SABI data, we find a median profit rate of 18.6\%, as shown in the right panel of Table \ref{tab:profit_sigma_qces_joint}. The associated $q^{CES}$ of 0.23 is our chosen measure of $q^{CES}$ for the quantitative application in Section \ref{sec:quant}, providing a more conservative benchmark than the WIOD median.

\begin{table}[htbp]
\centering
\small
\caption{Distribution of Profit Rate, $\sigma$, and $q^{CES}$}
\label{tab:profit_sigma_qces_joint}

\centering
\begin{tabular}{l c c c| c c c}
\hline\hline
&\multicolumn{3}{c}{\textbf{WIOD}}&\multicolumn{3}{c}{\textbf{Spanish administrative data (SABI)}}\\
\toprule
Percentiles & | \hspace{0.03cm} Profit Rate \hspace{0.03cm} | & $ \sigma $ & | \hspace{0.06cm} $ q^{CES} $ &Profit Rate \hspace{0.03cm} | & $ \sigma $ & | \hspace{0.06cm} $ q^{CES} $ \\
\midrule
1 & 0.0\% & . & . &-12.1\% & . & . \\
5 & 10.0\% & 10.0 & 0.11 &3.6\% & 27.6 & 0.04  \\
10 & 12.2\% & 8.2 & 0.14&6.8\% & 14.6 & 0.07  \\
25 & 13.0\% & 7.7 & 0.15 &11.7\% & 8.5 & 0.13 \\
\rowcolor{lightgray} 50 & 17.2\% & 5.8 & 0.21 & 18.6\% & 5.4 & 0.23  \\
75 & 23.8\% & 4.2 & 0.31 &28.4\% & 3.5 & 0.40  \\
90 & 31.3\% & 3.2 & 0.46 &40.8\% & 2.5 & 0.69 \\
95 & 34.4\% & 2.9 & 0.52 &49.7\% & 2.0 & 0.99  \\
99 & 69.9\% & 1.4 & 2.32& 65.3\% & 1.5 & 1.88  \\
\hline\hline

\end{tabular}

\vspace{6pt}
\begin{minipage}{\textwidth}
\footnotesize
\textbf{Notes:} The empirical distribution pertains to the profit rate; the values for $\sigma$ and $q^{CES}$ are obtained from the profit rates through the theoretical relationships set out in the text. For the SABI panel, we omit the ``Percentiles" column and align rows to the left panel’s percentiles. For SABI, we keep only firms with $\text{employees}\ge 5$, as profit rates of very small firms are outlier-prone. The definition of profits includes capital expenses, consistently with the WIOD definition. Results are very similar when including all firms (median $q^{CES}=0.24$).
\end{minipage}
\end{table}

\paragraph{Estimates of $q$} We estimate love-of-variety $q$ by means of eq. \eqref{empiricalmodel}. Since the procedure holds for changes in both expenditure and labor usage, we report results for both in Table \ref{tab:q_est}. As the positive demand shocks could induce firm entry with a delay, we use both contemporaneous and lagged shocks as instruments.\footnote{In Appendix \ref{appendix: e-rob}, we show that our results are robust to this choice. We also show that the findings hold when weighting observations by industry size.} The first and second columns in each panel show the first- and second-stage estimation results, respectively. In all our estimations, we find a significant first stage and $\beta$s larger than 1. Recall that $q = \beta-1$, which implies that all our estimates of $q$ are strictly positive. Depending on the specification, we find estimates of $q$ between .512 and .634. For comparison, using a different identification and firm-level data, \cite{baqaee2023supplier} estimate a love-of-variety elasticity of $0.3$. We interpret our larger estimate as a sign of even stronger aggregate increasing returns than implied by individual firms' production functions. 

Panel (C) in Table \ref{tab:q_est} reports the results of an empirical model treating eq. (\ref{empiricalmodel}) for labor and income as not independent. We estimate a 3SLS model explicitly accounting for the correlation between the error terms across the system of equations. Additionally, as implied by our model, we impose common coefficients, a restriction that the data does not reject. We choose the estimate of $q$ from panel (C), $0.568$, as our $q$ for the quantitative application in Section \ref{sec:quant}.

\begin{table}[htb]
    \centering
    \caption{Estimation Results}\label{tab:q_est}
    \scalebox{.8}{{\def\sym#1{\ifmmode^{#1}\else\(^{#1}\)\fi} \begin{tabular}{l c c | c c | c} \hline\hline & \multicolumn{2}{c|}{(A)} & \multicolumn{2}{c|}{(B)} & \multicolumn{1}{c}{(C)}\\
                    &\multicolumn{1}{c}{(1)}&\multicolumn{1}{c|}{(2)}&\multicolumn{1}{c}{(3)}&\multicolumn{1}{c|}{(4)}&\multicolumn{1}{c}{(5)}\\
                    &\multicolumn{1}{c}{$\Delta \log \mathcal{I}^i_{c,t} $}&\multicolumn{1}{c|}{$\Delta \log Y^i_{c,t}$}&\multicolumn{1}{c}{$\Delta \log L^{p,i}_{c,t} $}&\multicolumn{1}{c|}{$\Delta \log Y^i_{c,t}$}&\multicolumn{1}{c}{$\Delta \log Y^i_{c,t}$}\\
\hline
                &                     &                     &                     &                     &                     \\
 $ \eta_{c,t}^{+,i} $ &      0.0854\sym{***}&                     &       0.103\sym{***}&                     &                     \\
                    &    (0.0215)         &                     &    (0.0204)         &                     &                     \\
[.5em]
 $ \eta_{c,t-1}^{+,i} $ &      0.0941\sym{***}&                     &      0.0646\sym{***}&                     &                     \\
                    &    (0.0231)         &                     &    (0.0204)         &                     &                     \\
[.5em]
$\Delta \log \mathcal{I}^i_{c,t} $&                     &       1.512\sym{***}&                     &                     &       1.568\sym{***}\\
                    &                     &     (0.191)         &                     &                     &     (0.151)         \\
[.5em]
$\Delta \log L^{p,i}_{c,t} $&                     &                     &                     &       1.634\sym{***}&                    1.568\sym{***} \\
                    &                     &                     &                     &     (0.245)         &          (0.151)           \\
\hline
$ P(q \leq q^{CES}) $&                     &         .055         &                     &        .040         &        .008         \\
Year FE             &         Yes         &         Yes         &         Yes         &         Yes         &         Yes         \\
Industry FE         &         Yes         &         Yes         &         Yes         &         Yes         &         Yes         \\
N                   &       10067         &       10067         &       10083         &       10083         &       10067         \\
F-Stat              &          22         &          63         &          23         &          45         &                 \\
$ R^2 $             &      0.0315         &      0.0864         &      0.0413         &      0.0881         &               \\
F-Stat eq.(1)       &                     &                     &                     &                     &        2843         \\
F-Stat eq.(2)       &                     &                     &                     &                     &         131         \\
$ R^2 $ eq.(1)      &                     &                     &                     &                     &      0.0866         \\
$ R^2 $ eq.(2)      &                     &                     &                     &                     &      0.0877         \\
\hline\hline
\multicolumn{6}{l}{\footnotesize Robust Standard Errors (HC1). \sym{*} \(p<0.10\), \sym{**} \(p<0.05\), \sym{***} \(p<0.01\)}
\end{tabular}
}
}
    
    \vspace{0.2cm}
    
    \begin{minipage}{\textwidth}
    \justifying\noindent
    \footnotesize
     \textbf{Notes:} $ P(q \leq q^{CES}) $ is computed under the null hypothesis $ H_0 : q \leq q^{CES} $. When we estimate the system jointly using 3SLS, $ P(q \leq q^{CES}) $ is computed using the estimated coefficient and standard error of the income regressor. The F-statistic for the jointly estimated model is approximated by the model's Wald test statistic, which follows a $ \chi^2 $ distribution, divided by the model degrees of freedom. We impose a common coefficient in the joint estimation, as we fail to reject, at the 10\% significance level, the hypothesis that the coefficients on labor and income are the same.
    \end{minipage}

\end{table}
Through the lens of our model, the external returns $\hat\beta-1$ are fully attributed to LoV $q$. This interpretation crucially relies on the absence of alternative sources of external returns. As we have established earlier, in the presence of other external increasing returns $1/\delta>1$, the estimate for $q$ has a bias of $(q+1)\frac{1-\delta}{\delta}$. We can then compute how large $1/\delta$ would need to be for the true data-generating process to have $q=1/(\sigma-1)$. Given our estimate of $1/(\sigma-1)$ and the bias formula, we obtain $\bar\delta=.78$, which implies that we would need external increasing returns of the order of 27\% ($1/\delta = 1.27$) to undo our qualitative findings. Empirical work generally finds external returns of smaller magnitude. \cite{noauthor_empirics_2015} review a large number of studies and conclude that the elasticity of aggregate (area-level) TFP with respect to local density is around 0.04-0.07, while \cite{coe_international_1995} estimate knowledge spillover elasticities between 0.06 and 0.1. Interpreting these elasticities entirely as external effects still yields estimates substantially smaller than the 27\% needed to overturn our qualitative result that $q<q^{CES}$.

To check robustness to mismeasurement in deflators, we invoke Remark \ref{rem:ident_pindex}. If deflators correctly incorporate the variety effect, the implied $q$ lies between $0.362$ and $0.568$. If deflators instead neglect the variety effect, then, given any non-zero output share in fixed cost production $\underline{\alpha}$, the implied $q$ lies between $0.595$ and $0.595/\underline{\alpha}$. Note that, if statistical agencies correctly report deflators, our lower-bound remains larger than $q^{CES}$. If instead deflators are mismeasured, our estimate is a lower bound for the true value of love-of-variety. 

In Figure \ref{q_fig}, we provide a visual representation of our results. We plot our chosen estimate of $q$ from panel (C) against the distribution of $q^{CES}$ at the industry level. Our estimate implies that only 2 industries out of 56 in our sample have $q^{CES}>q$, and for the vast majority of industries, $q$ is substantially larger than $q^{CES}$. Furthermore, the median $q^{CES}$, which we use as our benchmark, lies below the lower bound of the two-sided 95\% confidence interval for our estimated $q$. Therefore, when conducting a one-sided $t$-test, we can reject the hypothesis that $q \leq q^{CES}$ at the 5\% significance level as per Table \ref{tab:q_est}.
This result suggests that recessions have a significant potential to reduce long-run welfare, because agents are sensitive to the loss of varieties. To quantify the long-run output and welfare effects of recessions, we use our estimates in a quantitative version of the model that allows us to account for transitional effects.
\begin{figure}[htb]
    \centering
    \includegraphics[width=0.8\linewidth]{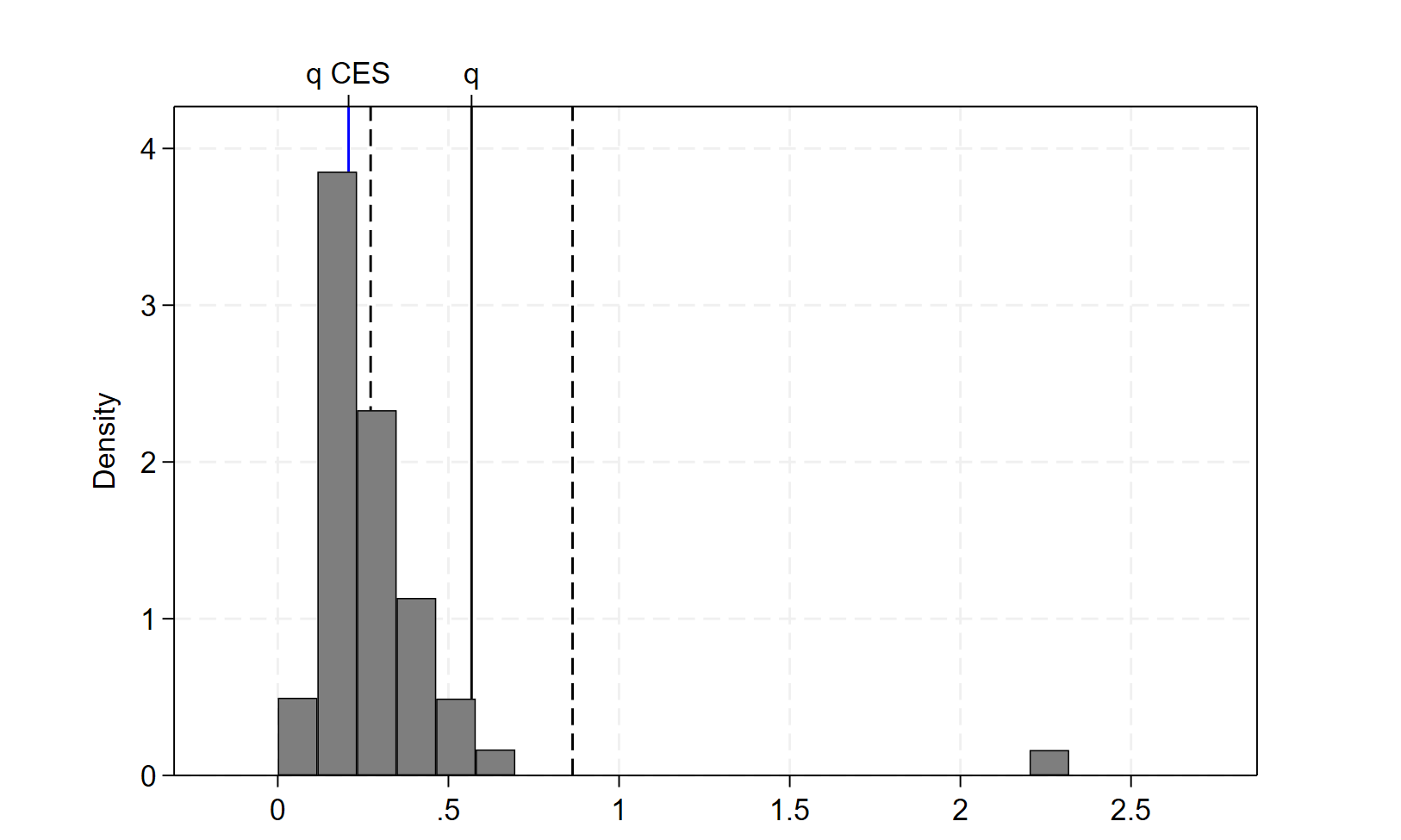}
    \caption{Estimated $q$ and $q^{CES}$. The figure plots the distribution of estimates of $1/(\sigma -1)$ by industry, highlighting the median, which is our preferred point estimate for $q^{CES}$. Additionally, it shows our estimate of a global $q$, with $95\%$ confidence bands.}
    \label{q_fig}
\end{figure}


\section{Quantitative Model}\label{sec:quant}
We conclude by considering a quantitative version of the model of Section \ref{sec:ge} to study fixed-cost cycles and the optimal policy response. We simulate firm exit and entry in a mean-reverting fixed-cost cycle and characterize output losses in terms of log deviations from the steady state along the transition.

\subsection{Set-Up and Calibration}\label{subsec:calib}
We generalize the analytical model by allowing firms to forecast the full path of fixed costs after the initial shock, rather than treating each period myopically. Firms discount the path of fixed cost with factor $\beta\in (0,1)$. They forecast
\begin{align}
    \pi(z, \mu_t) = \frac{\tilde \pi(z, \mu_t)}{1-\beta} - \sum_{0 \leq \tau < \infty} \beta^\tau \bar f^c_{t+\tau},  
\end{align}%
where $\tilde\pi$ are profits gross of fixed cost payments. Fixed cost after potential government subsidies are given by $\bar f_t^c$. Firms estimate the NPV of expected future profits by anticipating the transition of fixed costs after subsidies. The net present value of fixed cost, $NPV(\bar{\bm{f}}^c)$, determines the cutoff and enters the eqs. \eqref{eq:fe-pe}, \eqref{eq:zcp-pe} in place of $f^c$. We also include random exit shocks in the analysis, i.e., at the beginning of each period, a fraction $\delta \in (0,1)$ of firms exit, which yields a unique steady-state of the economy in the long run. The factor $(1-\delta)$ is subsumed into the firm discount factor. We set $\delta=3.36\%$, which is the average exit rate of all active Spanish firms, as covered by SABI, in the period preceding the Global Financial Crisis (GFC). As our sample begins in 2000, the average is calculated from 2000 to 2006. We assume that the entrant productivity distribution, $\mu^E$, is Pareto, and parameterized by location $z_{min} = 1$. For the shape parameter, we follow \cite{gabaix2011rank} to estimate the shape parameter of the pre-GFC revenue distribution in SABI. Denoting $h$ the shape parameter of the firm size distribution, in our model, the productivity distribution has a shape parameter of $h(\sigma-1)$. The elasticity of substitution, $\sigma$, and LoV, $q$, were estimated in Section \ref{sec:empirics}. The log-deviations of model quantities are invariant to the scale of $f^e > 0$, which therefore does not need to be calibrated. For the same reason, we normalize $\bar L = f_{ss}^c = 1$, where $f_{ss}^c$ refers to the steady state fixed cost.

The social planner solves the full dynamic Ramsey problem, accounting for the future effects of policy interventions in the current period. The planner maximizes welfare with discount factor $\beta^*$ and a social welfare function $U(\bm{Y}) = \sum_t (\beta^{*})^t \log (Y_t)$. The planner's discount factor is linked to the firms' discount factor through $\beta^* = \beta/(1-\delta)$. Namely, $\beta^*$ is the social discount factor, and firms discount faster due to the probability of exit.\footnote{We calibrate $\beta$ to the firm discount rate estimated in \citet{gormsen_corporate_2025}, which lies toward the lower end of empirical estimates. Our quantitative results are similar under higher values of $\beta$.} We allow the planner to choose two policy instruments, $\theta_{ss}$ and $\theta_{cyc}$, which modify the steady state and cycle fixed cost, respectively. Suppose that the economy is in its steady state in period $0$, and a crisis occurs in period $1$, such that $f_t^c > f_{ss}^c,\,\forall t\geq 1$. Given the policy levers $\theta_{ss}$ and $\theta_{cyc}$, the fixed cost paid by firms is %
\begin{align}
    \bar{f_t^c} =
    (1-\theta_{ss}) [ f_{ss}^c + (1-\theta_{cyc})  (f_t^c-f_{ss}^c) ].
\end{align} %
Hence, $\theta_{ss}$ is a subsidy/tax on the overall fixed cost ($f_{ss}^c$ in steady state), while $\theta_{cyc}$ is an additional subsidy/tax on the cyclical deviation $f_t^c-f_{ss}^c$, tracking the path of the exogenous fixed cost. An inactive planner is one for whom $\theta_{ss} =  \theta_{cyc}=0$.

The magnitude, $\varepsilon$, and persistence, $\alpha$, of the fixed cost crisis are calibrated to the GFC. We assume exponential decay of the fixed cost shock: \begin{align}
    f_{t}^{c} = f_{ss}^c + \varepsilon \cdot \alpha^{t}, \quad \alpha\in (0,1).
\end{align} We choose $\alpha$ such that the shock has a half-life of one year, and $\varepsilon$ is calibrated such that $20.44\%$ of firms exit on impact if the cycle policy is passive. This figure corresponds to the percentage deviation from trend in the number of active firms in Spain during the crisis period, as shown in Figure \ref{fig:loss_of_m_es}.

\begin{table}[htb]
    \centering
    \caption{Parametrization of the quantitative model. }\label{tab:quant_calibration}
    \begin{tabular}{ccll}
        \hline\hline
        & \textbf{Value} & \textbf{Description} & \textbf{Target/Source}\\ \hline
        $\beta$     & 0.864      & Firm discount rate (p.a.) & \cite{gormsen_corporate_2025}\\
         $\delta$    & 0.0336    & Death rate (p.a.) & Average pre-GFC exit rate in SABI \\
         $\beta^\ast$   & 0.894      & Social planner discount factor (p.a.) & $\beta/(1-\delta)$ \\
        $h$           & 1.2       & Tail parameter & Revenue distribution of firms in SABI \\
        $\sigma$    & 5.4      & Elasticity of substitution & Section \ref{sec:empirics} \\
        $q$         & 0.568     & Love-of-variety & Section \ref{sec:empirics} \\
        $L^p$       & 1.0       & Labor supply & Normalization \\
        $f_{ss}^c$     & 1.0      & Steady state fixed cost& Normalization \\
        $\alpha$    & 0.841     & Shock decay & 1-year half-life\\
        $\epsilon$  & see text & Shock magnitude & $\Delta$ Num. firms from trend in SABI\\ \hline\hline
    \end{tabular}
\end{table}

\subsection{Results}

We discuss the dynamics of output, $Y$, and its components: the mass of variety, $M$, average productivity, $\bar{z}\equiv \left(\int z^{\sigma-1}\mu(z) \d z\right)^{\frac{1}{\sigma-1}}$, and production labor, $L^p$, over the business cycle. 
We start by considering the laissez-faire case when the economy is hit by a recession that induces the targeted variety loss of $20.44\%$ on impact. Next, we allow the planner to implement a policy response along this cycle by subsidizing firms. These scenarios, however, feature steady-state distortions that the planner would like to correct because $q>1/(\sigma-1)$. Therefore, we also consider an alternative setting in which the planner implements the steady-state optimal policy, and then the economy is hit by the recession. Within this setting, we study the same two cycle regimes as above: an idle and active planner along the cycle. Importantly, note that having or not having the steady-state subsidy in place changes the steady-state number of firms. As a consequence, the same shock to the economy would induce different losses in varieties on impact. To ensure the comparability of these experiments, we recalibrate the magnitude of the shock so that in each specific scenario, there is an on-impact $20.44\%$ loss of varieties.


We report the welfare effects of the recession for each of these different scenarios in Table \ref{tab:policy}. Importantly, these should not be compared to the estimates of the welfare costs of business cycles, since all statements we make are specific to a single recession and its aftermath.

\subsubsection{Laissez-Faire}

\paragraph{Inactive Cycle Policy} We start by studying the laissez-faire scenario of Figure \ref{fig:quant_grid-lsf}. Panel B shows that, on impact, the targeted $20.44\%$ of varieties are lost. At the same time, production labor falls as higher fixed costs absorb labor, while average productivity rises through firm selection. As shown in Panel C, these adjustments jointly lead to a drop in output of $9.28\%$ on impact, compared with a $2.75\%$ contraction in the CES benchmark. While output in the CES economy experiences a small and stable long-term gain, the economy featuring our estimated $q$ experiences a persistent fall in output relative to its pre-crisis steady-state level. 

In consumption equivalent variation (CEV), the recession cost is $5.28\%$ of laissez-faire steady-state consumption, as reported in Table \ref{tab:policy}.\footnote{A value of $-x\%$ in Panel B of Table \ref{tab:policy} means that the planner would be willing to sacrifice $x\%$ of steady-state consumption to avoid the crisis.} The decomposition of the total recession cost shows that most of the welfare loss, $7.82\%$ CEV, can be attributed to variety losses, mirroring the large and persistent drop in $M$ in the figure. Accordingly, the welfare impact is dominated by the variety channel, despite a sizable positive contribution from selection-induced productivity gains ($+2.94\%$ CEV). The contribution of production labor is negative ($-0.4\%$) because the decline in $M$ is not strong enough to prevent $f^c M$ from rising. With fixed costs denominated in labor, this increase in $f^c M$ reduces $L^p$ one-for-one.

\begin{figure}[htbp]
    \centering
    \begin{tabular}{@{\hskip 0pt}c@{\hskip 0pt}c@{\hskip 0pt}c@{\hskip 0pt}}
        \subfloat{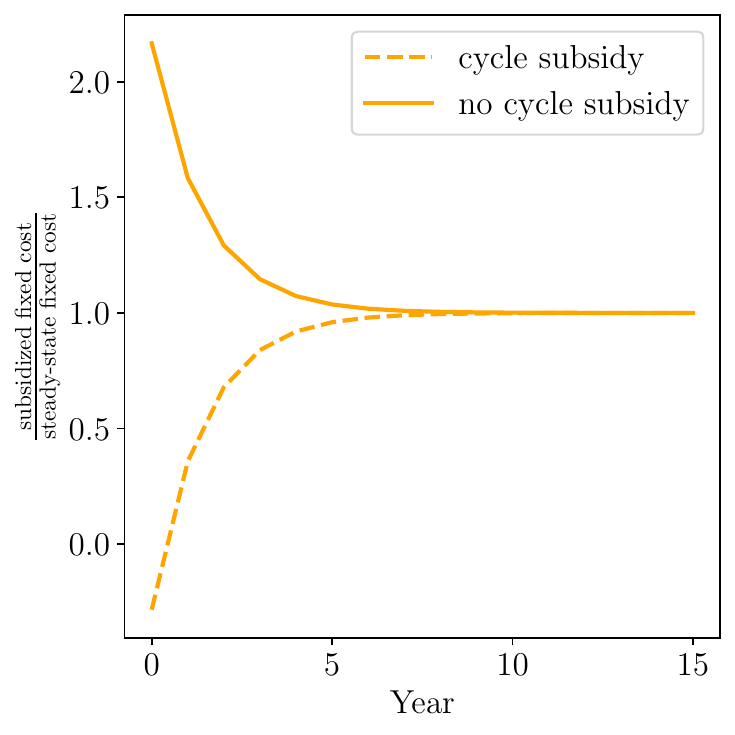}{A}
        \subfloat{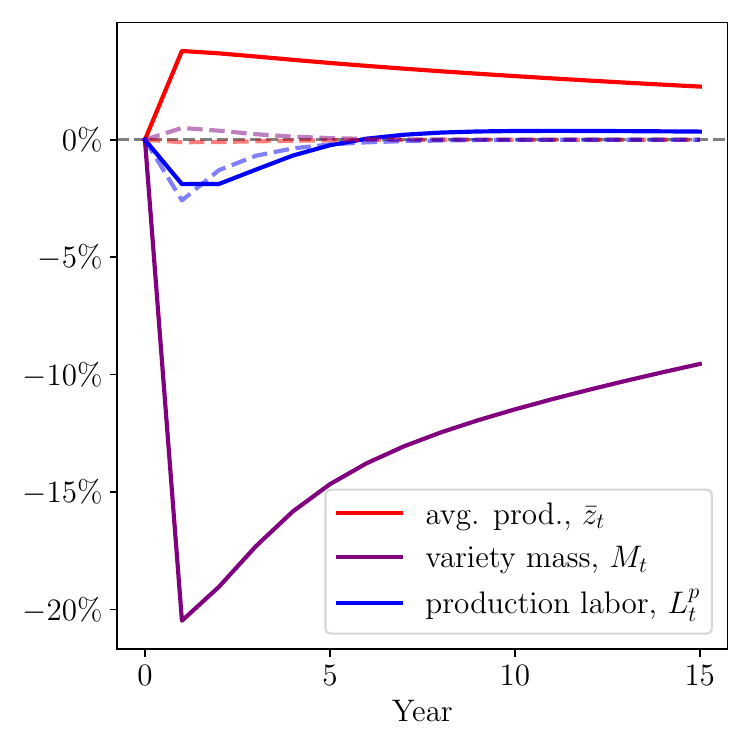}{B}
        \subfloat{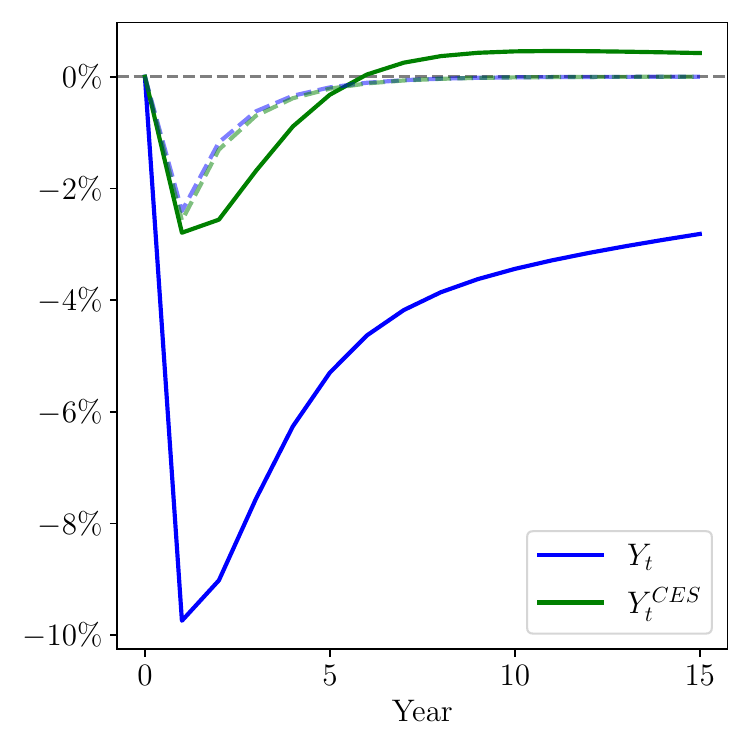}{C}
    \end{tabular}
    \caption{Fixed-cost shock path and impulse responses in a laissez-faire economy. Solid lines show the responses without an active cycle policy; dashed lines show the responses with an active cycle policy. Panel A displays the fixed-cost shock and its reversion calibrated as per \ref{subsec:calib}. Panel C and B report the impulse responses--log deviations from the pre-shock steady-state--of output and its components $Y_t= (M_t^q) (L_t^p) (\bar{z}_t)$: variety mass $M_t$, production labor $L^p_t$, and average productivity $\bar{z}_t$, respectively.}
    \label{fig:quant_grid-lsf}
\end{figure}

\paragraph{Active Cycle Policy} Allowing the planner to respond to the crisis by choosing $\theta_{cyc}$ decreases the cost of the crisis from 5.28\% to $0.63\%$ of steady-state consumption. The planner implements a large subsidy to limit firm exit. The subsidy diverts labor from production to fixed costs payments so that no firms exit during the recession, as shown by the dashed lines in Panel B of Figure \ref{fig:quant_grid-lsf} and in Table \ref{tab:policy}. Importantly, this policy intervention is effective because the economy starts with a sub-optimally low number of active firms. As a consequence, the cost of subsidizing their fixed cost is relatively limited compared to the benefit of reduced exit through the variety gains. In fact, the planner uses the recession and the possibility of engaging in policy intervention to partially substitute the missing steady-state instrument. As shown in Panel A of Figure \ref{fig:quant_grid-lsf}, the planner temporarily sets the subsidy to be such that the entry cost is negative at the beginning of the crisis. This induces an increase in the number of active firms that slowly fades as the economy converges back to the steady state. Thanks to this active policy, the effect on output is limited (-2\% instead of -10\%), and there is no difference between the economy and the counterfactual CES case where there are no variety effects. Since in this scenario the planner uses the intervention during the recession to imperfectly replace the missing steady-state intervention, we next allow it to operate these instruments separately.


\subsubsection{Optimal Steady-State Subsidy}
Next, we consider an economy where a social planner implements the optimal steady-state intervention characterized in Proposition \ref{prop:SP_solution}.\ref{prop:SP_subsidy}. Since $q>1/(\sigma-1)$, the planner optimally subsidizes the fixed cost in steady state, accordingly. Quantitatively, we find that $\theta_{ss}=0.88$.\footnote{Note that in Proposition \ref{prop:SP_solution}, we solved the problem of a myopic planner, whereas here we solve a forward-looking Ramsey problem with discount rate $\beta$ and exit rate $\delta$. While the qualitative implication is the same, the numerical solution is different because of the different planning problem.}
Crucially, since $q$ is substantially larger than $1/(\sigma-1)$, the economy is very far from first-best. Hence, the steady-state subsidy is very powerful and increases welfare by $55.30\%$ in consumption-equivalent terms, relative to the laissez-faire steady state.

In this economy, the recession is evaluated relative to a different steady state than under laissez-faire. Therefore, two clarifications are in order. First, any given CEV in Panel B of Table \ref{tab:policy} is computed with respect to a $55\%$ larger steady-state consumption if $\theta_{ss}$ is active. Second, when the steady-state policy is active, the marginal benefit and cost of adding varieties are balanced. Therefore, by an envelope-type argument, the marginal value of preserving varieties using $\theta_{cyc}$ during the crisis is much lower, reducing the effectiveness of the cycle intervention.

 \paragraph{Inactive Cycle Policy} The recession dynamics in this economy are shown in Figure \ref{fig:quant_grid-ssp}. The initial output loss of $12.04\%$ is substantially larger than in the laissez-faire economy, as shown in Panel C. This is driven by the higher equilibrium number of firms, which requires more labor diversion as the fixed costs increase. After the initial decline, $L^p$ recovers as fixed costs return to their original level.  The steady-state subsidy implemented by the planner fundamentally changes the economy by making entry and continued existence cheaper for firms, while increasing the price of production labor. It exacerbates the crisis in the very short run, but causes the recovery to be faster. The long-run consequence of the recession is a small increase in output, still associated with negative welfare effects due to volatility and discounting. In CEV terms, the welfare costs of the cycle are $3.68\%$ of steady-state consumption. As shown in Table \ref{tab:policy}, this loss is mostly driven by the loss of varieties, which alone reduces welfare by $6.75\%$, and is only partially offset by gains in production labor and in average productivity of $0.31\%$ and $2.76\%$, respectively.

\begin{figure}[htbp]
    \centering
    \begin{tabular}{@{\hskip 0pt}c@{\hskip 0pt}c@{\hskip 0pt}c@{\hskip 0pt}}
        \subfloat{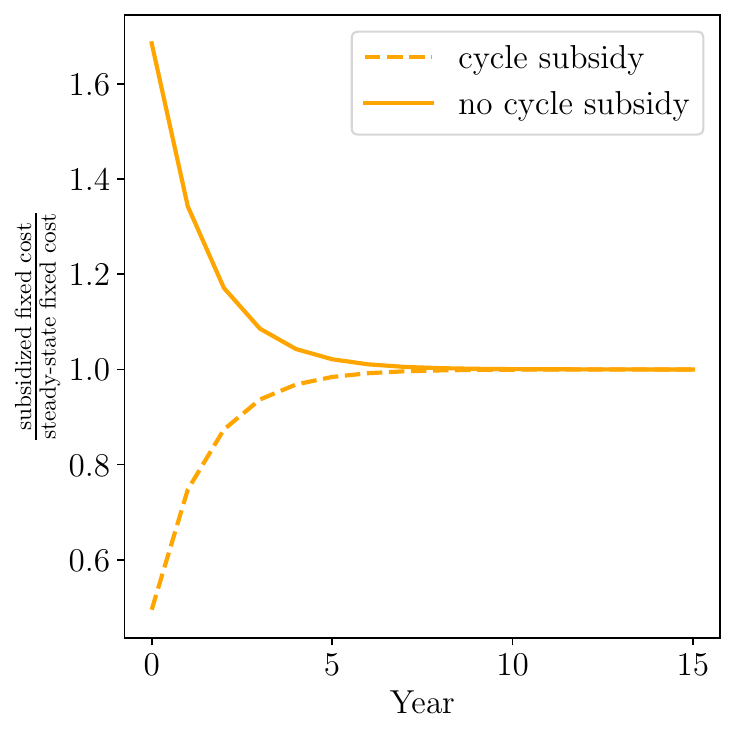}{A}
        \subfloat{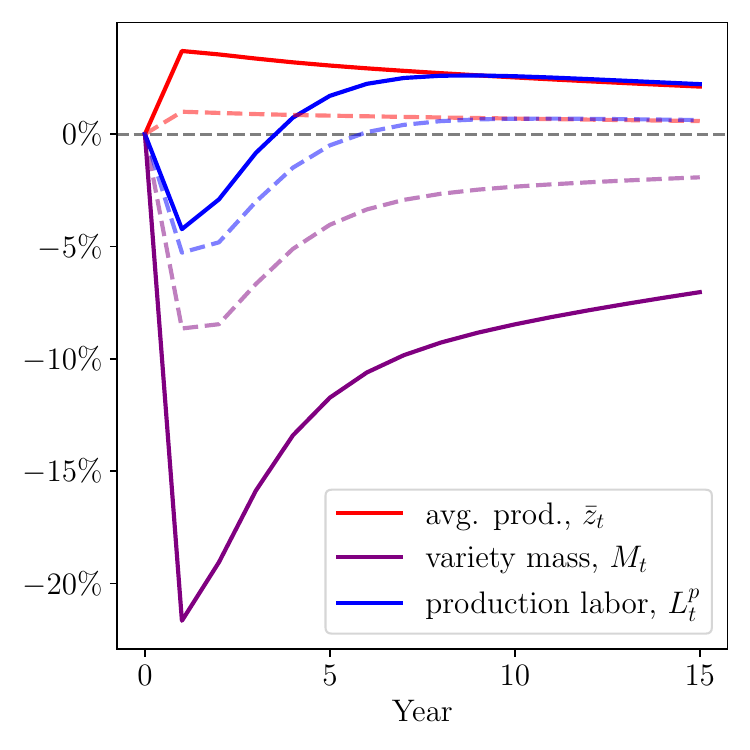}{B}
        \subfloat{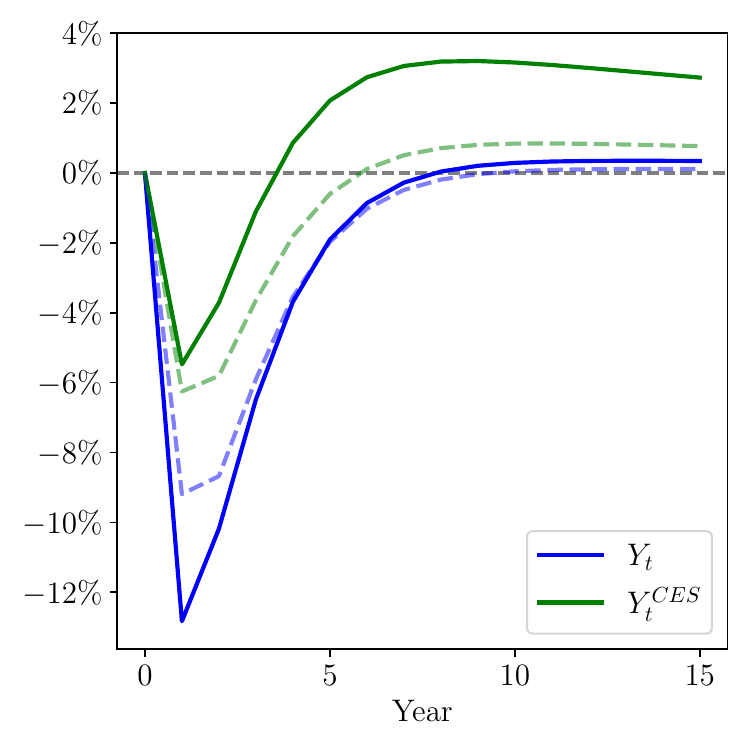}{C}
    \end{tabular}
    \caption{Fixed-cost shock path and impulse responses in an economy with active steady-state policy. See the notes of Figure \ref{fig:quant_grid-lsf} for additional details.
    }
    \label{fig:quant_grid-ssp}
\end{figure}

\paragraph{Active Cycle Policy} Finally, we examine the transition dynamics when the planner has access to both the steady-state ($\theta_{ss}$) and cycle instrument ($\theta_{cyc}$). 
In this setting, the social planner minimizes variety losses by subsidizing the fixed cost, as shown in Panels A and B of Figure \ref{fig:quant_grid-ssp}. Unlike when the steady-state policy is unavailable, the planner does not induce net entry by temporarily setting the fixed cost negative; instead, it absorbs part of the fixed-cost increase to deter exit.
This causes the crisis to unfold almost exclusively in the factor markets: since the policy prevents firm exit, the total overhead $M f^c$ shouldered by the economy becomes very large during the crisis. As a consequence, labor allocated to production collapses, and output gain from the active policy is relatively small and temporary, see Panel C. 

The crisis is deep but short-lived, and this pattern is largely shaped by the steady-state policy rather than by the cyclical intervention. The steady-state policy places the economy at a point where the variety margin is already internalized, so the marginal value of an additional variety is zero in steady state. When the shock hits, preserving varieties through cyclical intervention yields limited additional welfare benefits, because any reduction in variety losses is quickly offset by weaker selection and a larger diversion of labor away from production. Quantitatively, giving the planner access to the cycle instrument results in a small reduction in the cost of the crisis, from $3.68\%$ to $3.48\%$ CEV, as shown in Table \ref{tab:policy}. The decomposition of the total recession cost shows that the improvement from a smaller variety-loss component, from $-6.75\%$ to $-2.39\%$ CEV, is largely offset by reductions in the production-labor and average-productivity margins.

The equivalent of Figures \ref{fig:quant_grid-lsf} and \ref{fig:quant_grid-ssp}, and Table \ref{tab:policy} with $q=0.3$, as estimated by \cite{baqaee2023supplier}, can be found in Appendix \ref{sec:appendix_quant}. The results are qualitatively similar, with the caveat that, as love-of-variety is smaller, the economy is less inefficient and recessions are less costly.
 
\begin{table}[htb]
    \centering
    \caption{Steady-state welfare gains and recession welfare costs (CEV, \%)}
    \label{tab:policy}
    \scalebox{.9}{
    \begin{tabular}{l c c c}
        \hline\hline\addlinespace
        \multicolumn{4}{l}{\textbf{Panel A: Welfare Gains from Steady-State Policy}}\\
        \addlinespace
        Steady-State Welfare & $+55.30\%$ &  &  \\
        \midrule
        \addlinespace
        \multicolumn{4}{l}{\textbf{Panel B: Welfare Costs of Recessions}}\\
        \addlinespace
        Component & CEV relative to Steady State & $\theta_{ss}$ & $\theta_{cyc}$\\
        \midrule
        Total Output & $-5.28\%$ & \multirow{4}{*}{\centering inactive} & \multirow{4}{*}{\centering inactive} \\
        \hspace{6pt}Varieties & \hspace{6pt} $-7.82\%$ &  &  \\
        \hspace{6pt}Production Labor & \hspace{6pt} $-0.40\%$ &  &  \\
        \hspace{6pt}Avg. Productivity & \hspace{6pt} $+2.94\%$ &  &  \\
        \addlinespace
        Total Output & $-0.63\%$ & \multirow{4}{*}{\centering inactive} & \multirow{4}{*}{\centering active} \\
        \hspace{6pt}Varieties & \hspace{6pt} $+0.07\%$ &  &  \\
        \hspace{6pt}Production Labor & \hspace{6pt} $-0.67\%$ &  &  \\
        \hspace{6pt}Avg. Productivity & \hspace{6pt} $-0.03\%$ &  &  \\
        \addlinespace
        Total Output & $-3.68\%$ & \multirow{4}{*}{\centering active} & \multirow{4}{*}{\centering inactive} \\
        \hspace{6pt}Varieties & \hspace{6pt} $-6.75\%$ &  &  \\
        \hspace{6pt}Production Labor & \hspace{6pt} $+0.31\%$ &  &  \\
        \hspace{6pt}Avg. Productivity & \hspace{6pt} $+2.76\%$ &  &  \\
        \addlinespace
        Total Output & $-3.48\%$ & \multirow{4}{*}{\centering active} & \multirow{4}{*}{\centering active} \\
        \hspace{6pt}Varieties & \hspace{6pt} $-2.39\%$ &  &  \\
        \hspace{6pt}Production Labor & \hspace{6pt} $-1.84\%$ &  &  \\
        \hspace{6pt}Avg. Productivity & \hspace{6pt} $+0.75\%$ &  &  \\
        \addlinespace
        \hline\hline
    \end{tabular}
    }\vspace{6pt}
\begin{minipage}{\textwidth}\justifying\noindent\footnotesize
\textbf{Notes: }Panel A reports the steady-state welfare gain in the economy with the steady-state planner relative to the laissez-faire benchmark. Panel B reports the \emph{welfare cost of the recession} relative to the corresponding steady state. The decomposition of the total welfare cost shows the contribution of the unique components of output, i.e., variety mass, production labor, and average productivity, to the total CEV, using the output equation $Y = (M)^q (L^p)(\bar{z})$ as in eq. \eqref{eq:output_FD_avgprod}. The columns $\theta_{ss}$ and $\theta_{cyc}$ indicate whether the corresponding policy instrument is active in the reported scenario. The results are reported using the $q$ and $\sigma$ estimates from Section \ref{sec:empirics}. Numbers are shown as CEV percentages.
\end{minipage}
\end{table}

\section{Conclusions}
We study the cleansing effect of recessions in an economy where agents value the presence of multiple differentiated varieties. We show that recessions driven by rising fixed costs of production generate \textit{cleansing} in the sense that exiting firms are replaced by more productive entrants. However, we show that this need not translate into long-run gains in output and welfare. Whether the household is better off in the long run fundamentally depends on the extent of \textit{love-of-variety} in the downstream aggregation technology for industry output. We show that when \textit{love-of-variety} is stronger than in the benchmark CES economy, a social planner finds it optimal to increase fixed-cost subsidies during recessions to reduce the extent of firm exit. Our model-consistent estimates suggest that recessions generate long-run welfare losses over and above their transitional costs, contrary to what posited by the liquidationist view of business cycles.

\center{\addcontentsline{toc}{section}{ References} }
\bibliographystyle{aer}
{\footnotesize\bibliography{draft/references}}

@article{bilbiie2012endogenous,
  title={Endogenous entry, product variety, and business cycles},
  author={Bilbiie, Florin O and Ghironi, Fabio and Melitz, Marc J},
  journal={Journal of Political Economy},
  volume={120},
  number={2},
  pages={304--345},
  year={2012},
  publisher={University of Chicago Press Chicago, IL}
}

@article{caballero2005cost,
  title={The cost of recessions revisited: A reverse-liquidationist view},
  author={Caballero, Ricardo J and Hammour, Mohamad L},
  journal={The Review of Economic Studies},
  volume={72},
  number={2},
  pages={313--341},
  year={2005},
  publisher={Wiley-Blackwell}
}

@article{acabbi2022human,
  title={Human capital ladders, cyclical sorting, and hysteresis},
  author={Acabbi, Edoardo Maria and Alati, Andrea and Mazzone, Luca},
  journal={Cyclical Sorting, and Hysteresis (March 28, 2022)},
  year={2022}
}

@techreport{MatsuyamaUshchev2023_LoveForVariety,
  author       = {Matsuyama, Kiminori and Ushchev, Philip},
  title        = {Love-for-Variety},
  institution  = {CEPR Discussion Paper},
  number       = {DP18184},
  address      = {London},
  year         = {2023},
}

@article{coe_international_1995,
	title = {International {R}\&{D} spillovers},
	volume = {39},
	issn = {0014-2921},
	url = {https://www.sciencedirect.com/science/article/pii/001429219400100E},
	doi = {10.1016/0014-2921(94)00100-E},
	number = {5},
	urldate = {2026-01-02},
	journal = {European Economic Review},
	author = {Coe, David T. and Helpman, Elhanan},
	month = may,
	year = {1995},
	pages = {859--887},
}

@incollection{noauthor_empirics_2015,
	title = {The {Empirics} of {Agglomeration} {Economies}},
	volume = {5},
	url = {https://www.sciencedirect.com/science/chapter/handbook/pii/B9780444595171000052},
	language = {en-US},
	urldate = {2026-01-02},
	author = {Combes, Pierre-Philippe and Gobillon, Laurent},
	booktitle = {Handbook of {Regional} and {Urban} {Economics}},
	publisher = {Elsevier},
	month = jan,
	year = {2015},
	doi = {10.1016/B978-0-444-59517-1.00005-2},
	pages = {247--348},
}

@article{ackerberg_identification_2015,
	title = {Identification {Properties} of {Recent} {Production} {Function} {Estimators}},
	volume = {83},
	issn = {1468-0262},
	doi = {10.3982/ECTA13408},
	language = {en},
	number = {6},
	urldate = {2026-02-09},
	journal = {Econometrica},
	author = {Ackerberg, Daniel A. and Caves, Kevin and Frazer, Garth},
	year = {2015},
	pages = {2411--2451},
}

@techreport{bilbiie2020aggregate,
  title={Aggregate-demand amplification of supply disruptions: The entry-exit multiplier},
  author={Bilbiie, Florin O and Melitz, Marc J},
  year={2020},
  institution={National Bureau of Economic Research}
}

@techreport{poilly,
  title={Love of Variety and Uncertainty Shocks
},
  author={Gaudio, Francesco Saverio and Poilly, C\'eline},
  year={2026}
}

@article{gormsen_corporate_2025,
	title = {Corporate {Discount} {Rates}},
	volume = {115},
	issn = {0002-8282},
	url = {https://www.aeaweb.org/articles?id=10.1257/aer.20231246},
	doi = {10.1257/aer.20231246},
	language = {en},
	number = {6},
	urldate = {2025-09-13},
	journal = {American Economic Review},
	author = {Gormsen, Niels Joachim and Huber, Kilian},
	month = jun,
	year = {2025},
	pages = {2001--2049}
}

@book{bendetti2024hetereogeneous,
  title={Hetereogeneous firms, growth and the long shadows of business cycles},
  author={Bendetti-Fasil, Cristiana and Impullitti, Giammario and Licandro, Omar and Sedl{\'a}cek, Petr and Spencer, Adam},
  year={2024},
  publisher={Centre for Finance, Credit and Macroeconomics, School of Economics~…}
}

@article{zhelobodko2012monopolistic,
  title={Monopolistic competition: Beyond the constant elasticity of substitution},
  author={Zhelobodko, Evgeny and Kokovin, Sergey and Parenti, Mathieu and Thisse, Jacques-Fran{\c{c}}ois},
  journal={Econometrica},
  volume={80},
  number={6},
  pages={2765--2784},
  year={2012},
  publisher={Wiley Online Library}
}

@book{matsuyama2020does,
  title={When Does Procompetitive Entry Imply Excessive Entry?},
  author={Matsuyama, Kiminori and Ushchev, Philip},
  year={2020},
  publisher={Centre for Economic Policy Research}
}

@article{gabaix2011rank,
  title={Rank- 1/2: a simple way to improve the OLS estimation of tail exponents},
  author={Gabaix, Xavier and Ibragimov, Rustam},
  journal={Journal of Business \& Economic Statistics},
  volume={29},
  number={1},
  pages={24--39},
  year={2011},
  publisher={Taylor \& Francis}
}

@article{moreira2016firm,
  title={Firm dynamics, persistent effects of entry conditions, and business cycles},
  author={Moreira, Sara},
  journal={Persistent Effects of Entry Conditions, and Business Cycles (October 1, 2016)},
  year={2016}
}

@article{collard2025neoclassical,
  title={The neoclassical model and the welfare costs of selection},
  author={Collard, Fabrice and Licandro, Omar},
  journal={Review of Economic Dynamics},
  volume={57},
  pages={101284},
  year={2025},
  publisher={Elsevier}
}

@article{argente_innovation_2018,
	title = {Innovation and product reallocation in the great recession},
	volume = {93},
	issn = {03043932},
	language = {en},
	urldate = {2024-09-27},
	journal = {Journal of Monetary Economics},
	author = {Argente, David and Lee, Munseob and Moreira, Sara},
	month = jan,
	year = {2018},
	pages = {1--20},
}

@book{barro_economic_2004,
	address = {Cambridge, Mass},
	edition = {2nd ed},
	title = {Economic growth},
	isbn = {978-0-262-02553-9 978-0-262-26779-3 978-1-4237-2517-6},
	language = {en},
	publisher = {MIT Press},
	author = {Barro, Robert J. and Sala-i-Martin, Xavier},
	year = {2004},
}

@article{dixit1975monopolistic,
  title={Monopolistic competition and optimum product diversity, University of Warwick},
  author={Dixit, Avinash K and Stiglitz, Joseph E},
  journal={Economic Research Paper},
  volume={64},
  year={1975}
}

@article{mankiw1986free,
  title={Free entry and social inefficiency},
  author={Mankiw, N Gregory and Whinston, Michael D},
  journal={The RAND Journal of Economics},
  pages={48--58},
  year={1986},
  publisher={JSTOR}
}

@article{broda2010product,
  title={Product creation and destruction: Evidence and price implications},
  author={Broda, Christian and Weinstein, David E},
  journal={American Economic Review},
  volume={100},
  number={3},
  pages={691--723},
  year={2010},
  publisher={American Economic Association}
}

@article{tian2018firm,
  title={Firm-level entry and exit dynamics over the business cycles},
  author={Tian, Can},
  journal={European Economic Review},
  volume={102},
  pages={298--326},
  year={2018},
  publisher={Elsevier}
}

@article{spence1976product,
  title={Product selection, fixed costs, and monopolistic competition},
  author={Spence, Michael},
  journal={The Review of economic studies},
  volume={43},
  number={2},
  pages={217--235},
  year={1976},
  publisher={Wiley-Blackwell}
}

@article{parenti2017toward,
  title={Toward a theory of monopolistic competition},
  author={Parenti, Mathieu and Ushchev, Philip and Thisse, Jacques-Fran{\c{c}}ois},
  journal={Journal of Economic theory},
  volume={167},
  pages={86--115},
  year={2017},
  publisher={Elsevier}
}

@article{ethier1982national,
  title={National and international returns to scale in the modern theory of international trade},
  author={Ethier, Wilfred J},
  journal={The American Economic Review},
  volume={72},
  number={3},
  pages={389--405},
  year={1982},
  publisher={JSTOR}
}

@misc{chatterjee1993entry,
  title={Entry and exit, product variety and the business cycle},
  author={Chatterjee, Satyajit and Cooper, Russell},
  year={1993},
  publisher={National Bureau of Economic Research Cambridge, Mass., USA}
}

@techreport{baxter1991productive,
  title={Productive externalities and business cycles},
  author={Baxter, Marianne and King, Robert G},
  year={1991},
  institution={Federal Reserve Bank of Minneapolis}
}

@article{argente2024life,
  title={The life cycle of products: Evidence and implications},
  author={Argente, David and Lee, Munseob and Moreira, Sara},
  journal={Journal of Political Economy},
  volume={132},
  number={2},
  pages={337--390},
  year={2024},
  publisher={The University of Chicago Press Chicago, IL}
}

@article{davis1992gross,
  title={Gross job creation, gross job destruction, and employment reallocation},
  author={Davis, Steven J and Haltiwanger, John},
  journal={The Quarterly Journal of Economics},
  volume={107},
  number={3},
  pages={819--863},
  year={1992},
  publisher={MIT Press}
}

@article{dunne1989firm,
  title={Firm entry and postentry performance in the US chemical industries},
  author={Dunne, Timothy and Roberts, Mark J and Samuelson, Larry},
  journal={The Journal of Law and Economics},
  volume={32},
  number={2, Part 2},
  pages={S233--S271},
  year={1989},
  publisher={The University of Chicago Press}
}

@article{ferrari2022firm,
  title={Firm heterogeneity, market power and macroeconomic fragility},
  author={Ferrari, Alessandro and Queir{\'o}s, Francisco},
  journal={arXiv preprint arXiv:2205.03908},
  year={2024}
}

@techreport{baqaee2023supplier,
  title={Supplier churn and growth: a micro-to-macro analysis},
  author={Baqaee, David and Burstein, Ariel and Duprez, C{\'e}dric and Farhi, Emmanuel},
  year={2023},
  institution={National Bureau of Economic Research}
}

@article{barlevy2002sullying,
  title={The sullying effect of recessions},
  author={Barlevy, Gadi},
  journal={The Review of Economic Studies},
  volume={69},
  number={1},
  pages={65--96},
  year={2002},
  publisher={Wiley-Blackwell}
}

@article{ouyang2009scarring,
  title={The scarring effect of recessions},
  author={Ouyang, Min},
  journal={Journal of Monetary Economics},
  volume={56},
  number={2},
  pages={184--199},
  year={2009},
  publisher={Elsevier}
}

@article{caballero1996timing,
  title={On the timing and efficiency of creative destruction},
  author={Caballero, Ricardo J and Hammour, Mohamad L},
  journal={The Quarterly Journal of Economics},
  volume={111},
  number={3},
  pages={805--852},
  year={1996},
  publisher={MIT Press}
}

@article{caballero1994cleansing,
  title={The Cleansing Effect of Recessions},
  author={Caballero, Ricardo J and Hammour, Mohamad L},
  journal={American Economic Review},
  volume={84},
  number={5},
  pages={1350--1368},
  year={1994},
  publisher={American Economic Association}
}

@book{schumpeter1939business,
  title={Business cycles: A theoretical, historical and statistical analysis of the capitalist process},
  author={Schumpeter, Joseph Alois},
  year={1939},
  publisher={McGraw Hill Book Company}
}

@book{brakman2001monopolistic,
  title={The monopolistic competition revolution in retrospect},
  author={Brakman, Steven and Heijdra, Ben J},
  year={2001},
  publisher={Cambridge University Press}
}

@article{dixit1977monopolistic,
  title={Monopolistic competition and optimum product diversity},
  author={Dixit, Avinash K and Stiglitz, Joseph E},
  journal={The American economic review},
  volume={67},
  number={3},
  pages={297--308},
  year={1977},
  publisher={JSTOR}
}

@article{anderson2004trade,
  title={Trade costs},
  author={Anderson, James E and Van Wincoop, Eric},
  journal={Journal of Economic literature},
  volume={42},
  number={3},
  pages={691--751},
  year={2004},
  publisher={American Economic Association}
}

@article{kozeniauskas2022cleansing,
  title={On the cleansing effect of recessions and government policy: Evidence from Covid-19},
  author={Kozeniauskas, Nicholas and Moreira, Pedro and Santos, Cezar},
  journal={European Economic Review},
  volume={144},
  pages={104097},
  year={2022},
  publisher={Elsevier}
}

@article{ardelean_how_2006,
	title = {How {Strong} is  the {Love} of {Variety}?},
	url = {https://docs.lib.purdue.edu/ciberwp/46},
	journal = {Purdue CIBER Working Papers},
	author = {Ardelean, Adina},
	month = jan,
	year = {2006},
}

@article{hamano_monetary_2022,
	title = {Monetary policy, firm heterogeneity, and product variety},
	volume = {144},
	issn = {0014-2921},
	url = {https://www.sciencedirect.com/science/article/pii/S001429212200037X},
	doi = {10.1016/j.euroecorev.2022.104089},
	urldate = {2024-09-09},
	journal = {European Economic Review},
	author = {Hamano, Masashige and Zanetti, Francesco},
	month = may,
	year = {2022},
	pages = {104089},
}

@article{hamano_endogenous_2017,
	title = {Endogenous product turnover and macroeconomic dynamics},
	volume = {26},
	issn = {10942025},
	url = {https://linkinghub.elsevier.com/retrieve/pii/S1094202517300558},
	doi = {10.1016/j.red.2017.06.003},
	language = {en},
	urldate = {2024-09-09},
	journal = {Review of Economic Dynamics},
	author = {Hamano, Masashige and Zanetti, Francesco},
	month = oct,
	year = {2017},
	pages = {263--279},
}

@article{hopenhayn1992entry,
  title={Entry, exit, and firm dynamics in long run equilibrium},
  author={Hopenhayn, Hugo A},
  journal={Econometrica: Journal of the Econometric Society},
  pages={1127--1150},
  year={1992},
  publisher={JSTOR}
}

@article{melitz_impact_2003,
	title = {The {Impact} of {Trade} on {Intra}-{Industry} {Reallocations} and {Aggregate} {Industry} {Productivity}},
	volume = {71},
	issn = {1468-0262},
	url = {https://onlinelibrary.wiley.com/doi/abs/10.1111/1468-0262.00467},
	doi = {10.1111/1468-0262.00467},
	language = {en},
	number = {6},
	urldate = {2024-05-09},
	journal = {Econometrica},
	author = {Melitz, Marc J.},
	year = {2003},
	note = {\_eprint: https://onlinelibrary.wiley.com/doi/pdf/10.1111/1468-0262.00467},
	pages = {1695--1725},
}

@techreport{stiglitz1993endogenous,
  title={Endogenous growth and cycles},
  author={Stiglitz, Joseph E},
  year={1993},
  institution={National Bureau of Economic Research Cambridge, Mass., USA}
}

@book{schumpeter1934,
  title={Depressions: Can We Learn From Past Experience?},
  author={Schumpeter, Joseph Alois},
  year={1934},
  publisher={McGraw Hill Book Company}
}

@article{carvalho2019large,
  title={Large firm dynamics and the business cycle},
  author={Carvalho, Vasco M and Grassi, Basile},
  journal={American Economic Review},
  volume={109},
  number={4},
  pages={1375--1425},
  year={2019},
  publisher={American Economic Association 2014 Broadway, Suite 305, Nashville, TN 37203}
}

@article{clementi2016entry,
  title={Entry, exit, firm dynamics, and aggregate fluctuations},
  author={Clementi, Gian Luca and Palazzo, Berardino},
  journal={American Economic Journal: Macroeconomics},
  volume={8},
  number={3},
  pages={1--41},
  year={2016},
  publisher={American Economic Association 2014 Broadway, Suite 305, Nashville, TN 37203-2425}
}

@article{gouin2022productivity,
  title={Productivity dispersion, between-firm competition, and the labor share},
  author={Gouin-Bonenfant, {\'E}milien},
  journal={Econometrica},
  volume={90},
  number={6},
  pages={2755--2793},
  year={2022},
  publisher={Wiley Online Library}
}

@article{matsuyama2023non,
  title={Non-CES aggregators: a guided tour},
  author={Matsuyama, Kiminori},
  journal={Annual Review of Economics},
  volume={15},
  pages={235--265},
  year={2023},
  publisher={Annual Reviews}
}

@article{benassy1996taste,
  title={Taste for variety and optimum production patterns in monopolistic competition},
  author={Benassy, Jean-Pascal},
  journal={Economics Letters},
  volume={52},
  number={1},
  pages={41--47},
  year={1996},
  publisher={Elsevier}
}

@article{jovanovic1982selection,
  title={Selection and the Evolution of Industry},
  author={Jovanovic, Boyan},
  journal={Econometrica: Journal of the econometric society},
  pages={649--670},
  year={1982},
  publisher={JSTOR}
}

@article{matsuyama2025homothetic,
  title={Homothetic Non-CES Demand Systems with Applications to Monopolistic Competition},
  author={Matsuyama, Kiminori},
  journal={Annual Review of Economics},
  volume={17},
  year={2025}
}

@article{dhingra2019monopolistic,
  title={Monopolistic competition and optimum product diversity under firm heterogeneity},
  author={Dhingra, Swati and Morrow, John},
  journal={Journal of Political Economy},
  volume={127},
  number={1},
  pages={196--232},
  year={2019},
  publisher={University of Chicago Press Chicago, IL}
}

@article{kehrig2015cyclical,
  title={The cyclical nature of the productivity distribution},
  author={Kehrig, Matthias},
  journal={Earlier version: US Census Bureau Center for Economic Studies Paper No. CES-WP-11-15},
  year={2015}
}

@article{bilbiie2019monopoly,
  title={Monopoly power and endogenous product variety: Distortions and remedies},
  author={Bilbiie, Florin O and Ghironi, Fabio and Melitz, Marc J},
  journal={American Economic Journal: Macroeconomics},
  volume={11},
  number={4},
  pages={140--174},
  year={2019},
  publisher={American Economic Association 2014 Broadway, Suite 305, Nashville, TN 37203-2425}
}

@article{lee2015entry,
  title={Entry and exit of manufacturing plants over the business cycle},
  author={Lee, Yoonsoo and Mukoyama, Toshihiko},
  journal={European Economic Review},
  volume={77},
  pages={20--27},
  year={2015},
  publisher={Elsevier}
}

@book{matsuyama2023love,
  title={Love-for-variety},
  author={Matsuyama, Kiminori and Ushchev, Philip},
  year={2023},
  publisher={Centre for Economic Policy Research}
}

@article{ferrari2023inventories,
  title={Inventories, demand shocks propagation and amplification in supply chains},
  author={Ferrari, Alessandro},
  journal={ArXiv. org},
  number={2205.03862},
  year={2024},
  publisher={Cornell University}
}

@article{timmer2015illustrated,
  title={An illustrated user guide to the world input--output database: the case of global automotive production},
  author={Timmer, Marcel P and Dietzenbacher, Erik and Los, Bart and Stehrer, Robert and De Vries, Gaaitzen J},
  journal={Review of International Economics},
  volume={23},
  number={3},
  pages={575--605},
  year={2015},
  publisher={Wiley Online Library}
}

@techreport{chiavari2024customer,
  title={Customer Accumulation, Returns to Scale, and Secular Trends},
  author={Chiavari, Andrea},
  year={2024},
  institution={Working paper}
}

@techreport{MatsuyamaUshchev2022_SelectionSorting,
  author       = {Matsuyama, Kiminori and Ushchev, Philip},
  title        = {Selection and Sorting of Heterogeneous Firms through Competitive Pressures},
  institution  = {CEPR Discussion Paper},
  number       = {DP17092},
  address      = {London},
  year         = {2022},
}

@article{mayer2021product,
  title={Product mix and firm productivity responses to trade competition},
  author={Mayer, Thierry and Melitz, Marc J and Ottaviano, Gianmarco IP},
  journal={Review of Economics and Statistics},
  volume={103},
  number={5},
  pages={874--891},
  year={2021},
  publisher={MIT Press One Rogers Street, Cambridge, MA 02142-1209, USA journals-info~…}
}
\begin{appendices}
\justifying
\newpage
\begin{center}
    \huge{Online Appendix}\\
    \Large{Not-so-Cleansing Recessions}
\end{center}
\justifying
\setcounter{prop}{0}
\setcounter{lemma}{0}
\setcounter{figure}{0}
\setcounter{table}{0}
\setcounter{remark}{0}
\setcounter{cor}{0}

\renewcommand{\theprop}{A.\arabic{prop}}
\renewcommand{\thefigure}{A.\arabic{figure}}
\renewcommand{\thetable}{A.\arabic{table}}
\renewcommand{\theappprop}{A.\arabic{appprop}}
\renewcommand{\theapplemma}{A.\arabic{applemma}}
\renewcommand{\theappcor}{A.\arabic{appcor}}
\renewcommand{\theappremark}{A.\arabic{appremark}}

\section{Extensions and Additional Results}\label{sec:extensions}

\subsection{Homothetic Aggregators}\label{ext:hom-aggr}
In this section, we show that within the class of homothetic aggregators with variety externality, the aggregator in eq. (\ref{eq:aggr}) is particularly suited for analytical characterizations and empirical work. 

\begin{appprop}[Homothetic Aggregators]\label{prop:HSA}
    Let $Y(\bm{y})$ be a homothetic production function from the HSA, HIIA, or HDIA families by \cite{matsuyama2025homothetic}, acting on a bundle of inputs $\bm{y}$. Let $M = |\bm{y}|$ be the number of varieties in the input bundle. Suppose that \begin{enumerate}
        \item One can separate $Y$ into a variety externality and an aggregator with constant, zero LoV: $$
            Y(\bm{y}) = h(M) g(\bm{y}) \quad \text{with} \quad \frac{d\ln g}{d \ln M}=0
        $$
        \item The variety externality has constant LoV: $$
        \frac{d\ln h}{d \ln M}= q = const.
        $$
    \end{enumerate}
    Then $Y$ is the aggregator in eq. (\ref{eq:aggr}), with $g(\bm y)=Y^{CES}(\bm y)\,M^{-\,\frac{1}{\sigma-1}} \;\text{and}\;\; h(M)=M^{q}$.
\end{appprop}

These desirable properties allow us to separate the role of love-of-variety from that of the price elasticity of demand sharply. They also allow us to recover empirically an exact map of our model, rather than a local effect in Section \ref{sec:empirics}.

\subsection{Entry Game Equivalence}\label{ext:entry_equiv}
In the main text, we have assumed that the entry process is an instantaneous equilibrium in which all market participants anticipate the firm distribution after entry, $m$. In an alternative equilibrium characterization without anticipation of $m$, firms iteratively enter the market, only considering the ex-ante distribution of firms in each iteration. Given a distribution of incumbents, initially, a small mass of firms enters the market, and the low-productivity subset folds immediately. This play repeats until no additional firm wishes to enter and no active firm wishes to quit. Proposition \ref{remark:equiv} states that the two entry processes are equivalent. 
\begin{appprop}[Entry game equivalence]\label{remark:equiv}
The instantaneous equilibrium has a unique solution $(E, \underline z, m)$, which coincides with the limit point of the iterative entry game.
\end{appprop}

\subsection{Elastic Factor Supply}\label{ext:elastic-labor}
{ In the main text, we have considered economies with fixed factor supply. Here, we extend a positive result to economies with elastic factor supply.

\begin{appprop}[Elastic Factor Supply]\label{prop:elast_labor}
    Consider an economy with elastic primary factor supply (EF). Then, relative to an economy with fixed factor supply (FF), long-run output effects of a recession are dampened: $|\Delta\log Y^{EF}|<|\Delta\log Y^{FF}|$.
\end{appprop}
}

\subsection{Are larger crises more cleansing?} \label{sec:depth}
Until now, we have considered a fixed business cycle characterized by $(f_l^c, f_h^c)$. A natural question is whether larger increases in the fixed cost induce larger cleansing effects and variety losses. To answer this, we start this analysis by noting that the effect of a marginally deeper recession can be summarized by the following elasticity:
\begin{align}
       \frac{\partial \log Y_3}{\partial \log f^c_h}
= \frac{\partial \log M_3}{\partial \log f^c_h}\times \bigg[    \frac{\partial \log Y^{CES}_3}{\partial \log  M_3 } + (q-q^{CES})
 \bigg],\label{eq:elast}
\end{align}
where $Y_\tau^{CES}:=L_\tau^p (\int z^{\sigma - 1} m_\tau(z) \d z)^{1/(\sigma-1)}$ is the phase $\tau$ output in an economy with $q=q^{CES}$. We emphasize that eq. (\ref{eq:elast}) is a local chain-rule identity for the \emph{full} equilibrium response, \emph{not} a partial derivative that holds the composition $m_\tau(\cdot)$ fixed. The sign of this elasticity determines whether a marginally bigger recession increases or decreases long-run output. This effect can be decomposed into two elements: (i) the effect of larger fixed costs increases on a CES economy and (ii) the additional variety effect. Both of these terms depend fundamentally on whether a larger fixed cost change increases or decreases the long-run number of available varieties. We study each element of eq. (\ref{eq:elast}) in turn.

\paragraph{Effects of $f_h^c$ on the mass of active firms $M_3$.} 
We have argued above that the long-run number of firms declines after a crisis. This effect is driven by selection, as the average entrant is strictly more productive than the average exiter. This holds true for all recession intensities $f^h_c$ and independently of LoV, $q$. 

However, larger crises do not always imply marginally fewer firms in the long run: the sign of $\frac{\partial \log M_3}{\partial \log f^c_h}$ is ambiguous. The magnitude of the selection force that induces a decline in $M_3$ depends on the intensity of the recession $f^h_c$. To understand this, consider a very large recession. In such a case, the marginal exiting firm is very productive, in fact, more productive than the expected entrant. As a consequence, even larger recessions induce less, not more selection. The opposite holds for very small recessions: a larger increase in $f_h^c$ induces stronger selection effects. As selection effects drive the long-run number of firms $M$, we have that stronger recessions have ambiguous effects on $M_3$. We formalize this insight in Lemma \ref{lem:depth}.

\begin{applemma}[Recessions Depth and Cleansing Effects]\label{lem:depth}
    The long run number of active firms in the economy $M_3(f_h^c) \leq M_1$ attains a unique minimum $M_3^{min}$ at some crisis level, $f_h^{c, *} \in (f_l^c, \infty)$. For each crisis driven by $f_h^c\in (f_l^c, f_h^{c, *})$ there exists some crisis $f_h^{c\prime}\in (f_h^{c, *}, \infty)$, such that $M_3(f_h^c) = M_3(f_h^{c\prime})$.
\end{applemma}
The mechanism underlying this result is fully driven by the presence of incumbents. The marginal effect of the magnitude of the crisis on the long-run number of firms is given by 
\begin{align}\label{eq:m3-fc}
    \frac{\partial \log M_3}{\partial \log f_h^c}\Bigg\vert_{f^c = f_h^c} &\propto\frac{\partial}{\partial f_h^c}\left[ \underbrace{M_1 p_E(\underline z_2)}_{\text{Surviving Firms in Phase 2}}+ \underbrace{E_3  p_E(\underline z_1)}_{\text{Entrants in Phase 3}}\right] \\
    &=\left[  \underbrace{M_1\frac{\partial p_E(\underline z_2) }{\partial \underline z_2}}_{<0}  +  \underbrace{p_E(\underline z_1) \frac{\partial E_3}{\partial \underline z_2}}_{>0} \right] \cdot \underbrace{\frac{\partial \underline z_2}{\partial f_h^c}}_{> 0}, \notag
\end{align}
where $p_{E}( \underline z) = \int_{\underline z}^{\infty} \mu^E(z) dz$. As $f_h^c$ increases, the crisis cutoff $\underline z_2$ shifts up. At low levels of $f_h^c$, the first term in the bracket of eq. (\ref{eq:m3-fc}) dominates: marginal exiters are less productive than entrants, and there is a marginal decline in the number of firms. At high levels of $f_h^c$, the marginal exiter is more productive than the average entrant, and, therefore, any additional increase in the fixed cost is associated with marginally more firms in the long run.

\paragraph{Effect of $M_3$ on long-run output, $Y_3$.} The elasticity of long-run output $Y_3$ with respect to $M_3$ decomposes into (i) the elasticity of $Y_3^{CES}$---output if the economy had LoV of level $q^{CES}$---and (ii) the elasticity of the variety effect in the $q$-economy relative to a CES economy. The latter is constant and its sign depends on whether $q\lessgtr 1/(\sigma-1)$. The elasticity of $Y^{CES}$ with respect to $M_3$ is always negative. 
A higher number of firms $M_3$ affects $Y^{CES}$ in two different ways: (i) more firms directly imply more fixed cost payments and, therefore, less labor used in production, and (ii) a long-run equilibrium with more firms has undergone weaker selection along the transition. Both of these effects reduce $Y_3^{CES}$. However, note that $\frac{\partial \log Y_3^{CES}}{\partial\log M_3}$ is globally negative but not constant across different crisis intensities and independent of $q$.
We can now return to the total effect of having more firms in the post-cycle steady state. On the one hand, additional firms generate one-to-one gains from varieties, provided that $q>q^{CES}$. On the other hand, additional firms generate more-than-linear reductions in the labor-saving and selection gains from the crisis, and thus in $Y_3^{CES}$. We conclude that reducing firm and hence product variety, $M_3$, always has a negative effect on long-run output if $q$ is large and always has a positive effect if $q$ is small.

In between these two extremes, reducing firm and hence product variety $M_3$ has a negative effect on long-run output for small values of $M_3$ (where the variety effect is stronger than cleansing and selection) and positive for large values of $M_3$. Next, we discuss this region of ambiguity and its implications for the effect of crises of varying magnitude, $f_h^c$, on total output.

\paragraph{Effect of $f_h^c$ on $Y_3$.}
Taking stock, we know that as the crisis intensifies, the mass of varieties available in phase 3 first decreases then increases, and that this mass of varieties can have unambiguously positive, negative, or indeterminate effects on total output. We assemble these insights about the effect of $f_h^c$ on $Y_3$ in Proposition \ref{prop:compreces}:
\begin{appprop}[Interaction of cycle depth and LoV]\label{prop:compreces}
Index economies with otherwise equal parameters by their love-of-variety, $q$ (`$q$-economies'). Then, there exists a unique, nonempty interval $(q_\circ, q^\circ)$ with $q_\circ > q^{CES}$ such that: \begin{enumerate}
            \item $q\geq q^\circ \implies Y_3/Y_1 \leq 1$ for all $f_h^c$.
            \item $q\leq q_\circ \implies Y_3/Y_1 \geq 1$ for all $f_h^c$.
            \item for all $q\in (q_\circ, q^\circ)$, larger crises can be welfare improving or welfare reducing depending on their intensity $f_h^c$.
        \end{enumerate}
        
\end{appprop}

Proposition \ref{prop:compreces} provides three results. First, there exists a level of LoV $q_\circ$ such that any economy with $q\geq q^\circ$ faces output and welfare drops for any crisis intensity $f_h^c$. These economies value the presence of varieties so much that they experience long-run welfare losses, even for the smallest crises. Symmetrically, there is a subset of $q$-economies that do not value varieties as much and for which all recessions are long-run welfare-improving, independently of their intensity. Finally, the set of economies indexed by $q\in(q_\circ, q^\circ)$ is such that small and extremely large recessions can increase long-run output, while medium-sized recessions reduce it. This follows from the decreasing returns to labor-saving and selection effects as the number of varieties shrinks. We provide a graphical representation in Figure \ref{fig:decomp}.

Panel (A) shows how CES output, $Y^{CES}$, and the total mass of firms, $M_3$, respond to different increases in the fixed costs, $f_h^c$. Both quantities are independent of $q$. Relative to before the crisis, the mass $M_3$ drops, reaching a minimum at depth $f_h^{c,*}$, where all the incumbents with productivity less than the average successful entrant have been replaced. CES output always increases relative to before the cycle and reaches its highest level when no further returns from labor cleansing and selection effects can be extracted, i.e., at $f_h^{c,*}$. As the fixed cost rises toward eliminating every firm during phase 2, all quantities return to their pre-crisis levels.

Next, panel (B) shows the evolution of the love-of-variety effect for different levels of $q$. When $q>q^{CES}$, this effect tracks the path of the number of firms in panel (A). When $q<q^{CES}$, the LOV externality is such that the economy is better off when fewer firms are active, so the effect mirrors that of $M_3/M_1$. 

Finally, panel (C) shows the behavior of output for different $q-$economies. When $q$ is large, recessions are unambiguously welfare reducing and, vice versa, when $q$ is small they are always cleansing. In the intermediate range, some recessions can induce welfare gains while others generate welfare losses depending on their depth. 

This set of intermediate $q-$economies is characterized by an inference problem. In fact, past data on recessions cannot inform on whether future recessions will have cleansing effect or not. 

\begin{figure}[!!htbp]
    \centering
    \includegraphics[trim=0cm 0cm 0cm 0.3cm, clip, width=1\linewidth]{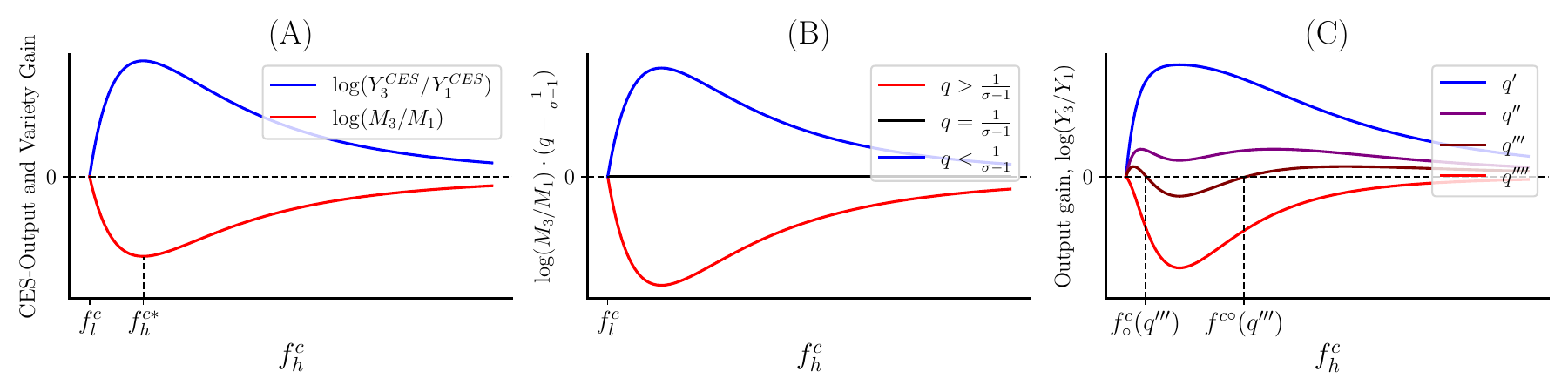}
    \caption{Decomposition of the log-output ratio for varying levels of LoV, $q$ and crisis, $f_h^c$. All curves start in $(f_l^c, 0)$ and are simulated from the model. Panel (A) shows the how $\ln(Y_3^{CES}/Y_1^{CES})$ and $\ln(M_3/M_1)$ vary with $f_h^c$. Panel (B) shows how LOV contributes to the output ratio through the factor $\ln((M_3/M_1)^{q- q^{CES}})$. Panel (C) sketches the curve of $\ln(Y_3/Y_1)$ for values of $q$, $q' < q'' < q'''< q''''$, whereby only $q''' \in ( q_\circ, q^\circ)$ holds (cf. Proposition \ref{prop:compreces}). Correspondingly, at $q'''$, the effect of crises is ambiguous, and there are two crises $f_\circ^c(q''') < {f^c}^\circ(q''')$ at which the output ratio returns equals 1. A proof of the `wiggly' shape when LoV is $q'''$ is provided in the proof of Proposition \ref{prop:compreces}.}
    \label{fig:decomp}
\end{figure}


\subsection{Heterogeneous Externality Draws}\label{app:matsuyama}
In the main text, we study a setting in which firms draw idiosyncratic productivities upon entry but are symmetric in terms of their contribution to the variety externality: each firm contributes a single variety. Consider an extension of our baseline economy where firms draw their contribution $\xi$ to the aggregate externality term. Suppose that, upon entry, entrants draw once and for all from a joint distribution $(\xi, z)\sim \mu^E(\xi, z)$. Denote the marginal entry distributions $\mu_\xi^E$ and $\mu_z^E$, and the marginal distributions of active firms by $m_\xi$ and $m_z$, for $\xi$ and $z$, respectively. W.l.o.g., we assume that the expected externality contribution drawn by an entrant is unity: $\E_{\mu_\xi^E}[\xi]=1$. The aggregate externality is then $\mathcal{M} \equiv \int\int \xi m(z, \xi) \d z\d\xi = \E_{m_\xi} [\xi] \times M$, with $ M$ being the count of how many firms there are in the market: $M=\int \int m(\xi, z) \d z \d \xi$. Importantly, as varieties have heterogeneous effects on the aggregator, we adopt $\mathcal M$ as the measure of varieties. Namely, LoV is the elasticity to $\mathcal M$. The only HSA aggregator which satisfies the separability conditions of Proposition \ref{prop:HSA} is given by \begin{align}
    Y = \mathcal{M}^q \times \left[\int y(z)^{(\sigma-1)/\sigma} \mu_z(z) \d z\right]^{\frac{\sigma}{\sigma -1}}.
\end{align} %
Note that the first factor is isoelastic in $\mathcal{M}$ with elasticity $q$, while the second factor has zero elasticity with respect to the number of varieties, since $\mu_z$ is a marginal probability density. 

Intuitively, heterogeneous externalities interact with the cleansing effects of cycles if and only if they are correlated with the selection mechanism based on productivity. We confirm this intuition in the following remark.

\begin{appremark}[Independence]\label{rem:indep}
    If $\xi$ and $z$ are independent in entrant and incumbent distribution, then the model is equivalent to our baseline model in firm-level choices and aggregates.
\end{appremark}
Away from this special case, we characterize the economy with a general joint distribution through the cycle, where we again refer to pre-cycle, recession, and post-cycle values with indexes 1, 2, and 3, respectively. The average externality contribution in the economy along the cycle is given by \begin{align}
    \E_{m_1}[\xi] &= \E_{\mu^{E}}[\xi \g z \geq \ulz_1], \\
    \E_{m_2}[\xi] &= \E_{\mu^{E}}[\xi \g z \geq \ulz_2], \\
    \E_{m_3}[\xi] &= \E_{\mu^{E}}[\xi \g z \geq \ulz_2]\times \frac{M_2}{M_3} + \E_{\mu^{E}}[\xi \g z \geq \ulz_1] \times \frac{E \mu_z(z\geq \ulz_1)}{M_3}.
\end{align}%
The last equation states that the long-run average contribution depends on the average among firms that survived throughout the recession (first term) and the average contribution among new entrants post-recession (second term). Like in the baseline model, we have $M_3 < M_1$. However, this is no longer the externality-relevant quantity in the current model. Instead, we note that 
 \begin{align}\label{eq:extension_hetvar_deltacurlyM}
    \mathcal{M}_3 > \mathcal{M}_1 \quad  \iff \quad \Delta \log \E_{m}[\xi] > - \Delta \log M.
\end{align} 
Namely, the variety effect is positive if the loss in the number of firms is more than compensated for by an increase in the average externality per firm. Since entry and exit choices at the firm level depend on productivity through their expected profits, the correlation between $\xi $ and $z$ determines whether recessions are less or more costly relative to the independent case of Remark \ref{rem:indep}.

Clearly, the welfare effect of the externality depends now on whether a business cycle selects out firms with relatively high or low contributions to the aggregate externality. We formalize this additional effect in the following proposition, which relates this extension to our baseline model and Proposition \ref{proposition:PE-cycles}.

\begin{appprop}[Business Cycles with Heterogeneous Externalities]\label{prop:het_ext}
  Consider the extension with heterogeneous contributions to the externality. Say that $\xi$ and $z$ are positively (negatively) correlated if $\E_{\mu^E}[ \xi \g z \geq a]$ is strictly increasing (decreasing) in $a$. Let $Y^{hom}$ be output in our baseline model with homogeneous externalities. Then: \begin{enumerate}
    \item The cycle induces negative selection in $\xi$ and $\Delta Y < \Delta Y^{hom}$ if $\xi$ and $z$ are negatively correlated.
    \item The cycle induces positive selection in $\xi$ and $\Delta Y > \Delta Y^{hom}$ if $\xi$ and $z$ are positively correlated.
    \item The cycle induces no selection in $\xi$ and the heterogeneity is irrelevant to output $\Delta Y = \Delta Y^{hom}$, if $\xi$ and $z$ are uncorrelated.
  \end{enumerate}
\end{appprop}

\subsection{Forward-Looking Exit}\label{ext:fwd-look-exit} A straightforward extension of the model is one in which firms anticipate the mean-reverting path of an adverse fixed cost shock, allowing them to see and discount $f_{t+1}^c,f_{t+2}^c, ...$. However, they do not foresee the path of future entry, forecasting future profits using $\tilde \pi(z, m_t)$. The decision to exit then hinges on \begin{align}
    NPV(\{f_{t+k}^c\}_k) \gtreqqless \tilde\pi(z, m_t)/(1-\beta).
\end{align} Accordingly, the cutoff jumps up less during a recession, while the dynamics of the model are unaffected.

\subsection{Aggregate Dynamics}\label{app:dynamics}

In this section, we briefly revisit the dynamics of an economy in which the fixed cost increases and slowly reverts to its original value, in order to characterize the transition behavior. For exposition in this section, we assume that $\mu^E$ is Pareto, i.e., $\mu^E(z)= \beta z_{min}^\beta z^{-(\beta+1)}$ with $\beta > \sigma -1$, and that the initial fixed cost are $f_{-1}^c = 1$. Fixed cost increase to $\phi>1$ in $t=0$ and revert back according to %
\begin{align}
    f_t^c = \phi^{ \max \{0, 1 - t\cdot \frac{1}{T^*}\} }, \quad (t = 0, 1, 2, ...)
\end{align}
where $T^*$ is the number of fixed cost changes until complete reversion. Hence, $\phi$ parametrizes intensity and $T^*$ parametrizes smoothness of the crisis. Since the equilibrium is attained instantly with any parameter change, and since each equilibrium would be a steady state in the absence of further parameter changes, we are effectively comparing a sequence of steady state equilibria. 
If one assigns unit length to the crisis, then $\mathrm{d}t = \frac{1}{T^*}$ is the length of a single time-delta. A subscript-$t$ indicates the value of a variable after $t$ such small intervals. Since the number of time-deltas passed through in a fixed cost cycle is $T^*$, $T^* \rightarrow \infty$ characterizes a continuous time limit. Using the equilibrium definition repeatedly, we derive closed-form solutions of the equilibrium at the end of the crisis (after $T^*$ time-deltas). Solutions are given in the propositions below. We assume that the pre-crisis equilibrium with entry and cutoff---which we call $E_{-1}$ and $\ulz_{-1}$, respectively---is derived from entry into a previously empty economy (i.e., no prior incumbents).

\begin{appprop}[Pareto Economy with Sequential Fixed Cost Reversion] \label{prop:smooth_reversion_pareto_sequential_fc}
    After a fixed cost increase, during the recovery, $t \in \{1, ..., T^*\}$, the economy features a constant flow of entrants
    \begin{align}
        E_t = \frac{\mathcal{I}}{\beta} \frac{\sigma-1}{\sigma} \frac{1}{f^e} \left[1 - \phi^{\frac{\sigma -1 -\beta}{\beta} \frac{1}{T^*}} \right] \quad (t = 1, 2, ..., T^*).
    \end{align} 
    The sequence of cutoff productivities, including the cutoff on impact in $t=0$, is smoothly declining and given by 
    \begin{align}
        \ulz_t = z_{min} \left\{ \frac{\phi^{ (T^* - t)/T^* } }{f^e} \frac{\sigma -1}{\beta - (\sigma - 1)} \right\}^{1/\beta} \quad (t = 0, 1, ..., T^*).
    \end{align} The measure of active firms after the crisis is given by %
    \begin{align}
        m_{T^*}(z) &= \frac{\mathcal{I}}{f^e} \frac{\sigma -1}{\sigma} \beta z_{min}^\beta z^{-(\beta+1)} \indicalt{ z_{min} \left\{ \frac{ 1 }{f^e} \frac{\sigma -1}{\beta - (\sigma - 1)} \right\}^{1/\beta} } \notag \\ 
        & +\frac{\mathcal{I}}{f^e} \frac{\sigma -1}{\sigma} \beta z_{min}^\beta z^{-(\beta+1)}  \sum_{t=1}^{T^*} \indicalt{ z_{min} \left\{ \frac{\phi^{ (T^* - t)/T^* } }{f^e} \frac{\sigma -1}{\beta - (\sigma - 1)} \right\}^{1/\beta} }  \left[1 - \phi^{\frac{\sigma -1 -\beta}{\beta} \frac{1}{T^*}} \right].
    \end{align}
    The long-run aggregate productivity, $\mathcal{Z}_{T^*}$, equals its pre-crisis value: %
    \begin{align}
        \mathcal{Z}_{T^*} = \mathcal{Z}_{-1},
    \end{align}
    and the long-run mass of active firms is \begin{align}
        M_{T^*} = E_{-1}p_E(\ulz_0) + \sum_{\tau = 1}^{T^*} p_E(\ulz_\tau) E_{\tau} &= 
        \frac{\mathcal{I}}{\sigma} \frac{\beta - (\sigma -1)}{\beta} \left\{ \phi^{-1}
        + (1-1/\phi)  \left( \frac{1-\phi^{- \frac{\beta- (\sigma -1 )}{\beta} \frac{1}{T^*} }}{ 1- \phi^{-\frac{1}{T^*}}} \right) \right\} \\
        < \frac{\mathcal{I}}{\sigma} \frac{\beta - (\sigma -1)}{\beta} = M_{-1},
    \end{align}
    and all quantities approach their pre-crisis values as $\phi \downarrow 1$.
\end{appprop}

The characterization of the transition provides one important economic insight: the slow reversion of the fixed cost induces selection \textit{along the way}. Since firms are allowed to enter before the fixed cost has reverted completely to $f^c_{-1}$, some of the post-crisis, entrants in period 2 face a more stringent cutoff than in period 3 and so on. This effect strengthens the cleansing effects and induces more long-run selection. The continuous time limit yields more succinct expressions in the following Proposition.

\begin{appprop}[Continuous Time Reversion] \label{prop:smooth_reversion_pareto_smooth_fc}
    In the limit as $T^* \rightarrow \infty$, the mass of active firms converges from above to %
    \begin{align}
        M_{T^*} \downarrow M_\infty = \frac{\mathcal{I}}{\sigma} \frac{\beta - (\sigma - 1)}{\beta} \left( (1 - \phi^{-1}) \frac{\beta - (\sigma -1)}{\beta} + \phi^{-1} \right).
    \end{align}%
    Moreover, the firm distribution converges to %
    \begin{align}
        m_{T^*}(z) \rightarrow m_\infty(z) = \underbrace{ \mu^E(z) \overbrace{\frac{\sigma -1}{\beta} \frac{\mathcal{I}}{\sigma} \frac{1}{f^e}}^{ = E_{-1}}
        }_{=m_{-1}(z)} \times \begin{cases}
             (\beta - (\sigma -1))\left[ \ln z - \ln \ulz_{-1} \right]  &\quad z \in [ \ulz_{-1}, \ulz_0 ), \\
              \left[ 1 + \frac{\beta - (\sigma -1)}{\beta} \ln \phi \right] &\quad z \in [\ulz_0 ,\infty).
        \end{cases}
    \end{align}
    The right tail of $m_{T^*}(z)$, $z\geq \underline z_0$ converges to $m_\infty(z)$ from below. The limit density $m_\infty(z)$ is continuous, except at $z = \ulz_0$, where it jumps up.
\end{appprop}

The continuous version discussed in Proposition \ref{prop:smooth_reversion_pareto_smooth_fc} reveals that cleansing effects in terms of an increased average productivity are more pronounced if fixed costs revert smoothly. Nonetheless, our main results still apply: If $q > q^{CES}$, then the long-run welfare effects of crises are negative. We illustrate graphically the result in Proposition \ref{prop:smooth_reversion_pareto_smooth_fc} in Figure \ref{fig:tent}. As highlighted, the limiting distribution is different when the reversion is smooth. The distribution has zero mass at its lower truncation point, $\ulz_{-1}$. Such a shape is more in line with productivity distributions inferred from data. Second, the right tail thickens relative to $m_{-1}$ by a factor that is increasing in the log of crisis intensity, yet the tail index of the distribution is unaffected.

\begin{figure}[htbp]
    \centering
    \includegraphics[width=0.55\linewidth]{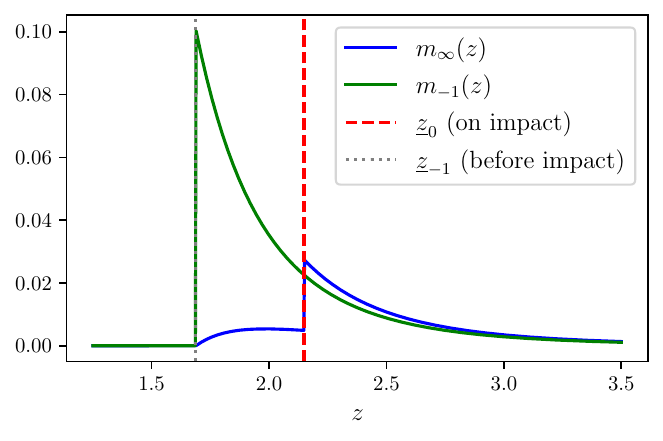}
    \caption{Limit firm productivity distribution after smooth fixed cost reversion. The distribution $m_\infty(z)$ has zero mass at its lower truncation point and a thickened right tail, as described in Proposition \ref{prop:smooth_reversion_pareto_smooth_fc}.}
    \label{fig:tent}
\end{figure}

\paragraph{General Equilibrium} Finally, we extend this characterization to the general equilibrium framework of Section \ref{sec:ge}. Because the smooth transition increases the strength of cleansing effects, the long-run number of firms is smaller. As a consequence, even less labor is used for fixed cost payments than in the case in which fixed costs revert all at once.

\begin{appprop}[Aggregate Dynamics in General Equilibrium]
    
\label{remark:GE_exponential_decay_entry}
    In general equilibrium, entry subsides over multiple periods, and the steady state is attained in the limit as $E_t \rightarrow 0$. $E_t$ exponentially decays, and the rate of decay is faster if successful entrants are more productive than firms at the cutoff, $\ulz$, and if the cutoff is low in the sense that $p_E(\ulz)$ is close to 1. 
\end{appprop}

\subsection{Stochastic Idiosyncratic Productivity} \label{app:stoch_prod}
Our baseline model posits that firms draw their productivity once and for all. Here, we study a setting in which they experience fluctuations in idiosyncratic productivity even after entry. Suppose firms are subject to individual, transitory productivity shocks as in \cite{hopenhayn1992entry}. Productivity shocks are i.i.d. across firms and periods. Firms decide at the beginning of the period, before learning their idiosyncratic shock, whether or not to produce. Effective productivity $\tilde z$ consists of a fixed component, $z$, drawn from $\mu^E$ on entry, and a shock $\varepsilon$, such that $\tilde z = z + \varepsilon$ and $\varepsilon \sim \psi$ with $\mathbb{E}[\varepsilon]=0$.\footnote{For ease of notation, we assume that $z + \varepsilon \geq 0$ always holds. Similar results can be achieved if one parametrizes the distribution of $\varepsilon$ by $z$ to ensure $\varepsilon + z \geq 0$.} We directly discuss the general equilibrium version of the model. Firm profits now read 
\begin{align}
    \pi(\tilde z, m) = \frac{R}{\sigma} \frac{(z + \varepsilon)^{\sigma - 1}}{\tilde{Z}(m, \psi)} - f^c,
\end{align}
where 
\begin{align}
    \tilde{Z}(m, \psi) := \int \int (z + \varepsilon)^{\sigma - 1} m(z)\psi(\varepsilon)\, \mathrm{d} z \, \mathrm{d} \varepsilon
\end{align}
is market intensity, which nests our baseline model for $\psi = \delta_0$ (the Dirac distribution centered about 0). Given the timing assumption, the zero profit condition pins down the cutoff productivity. The equilibrium conditions describing firm entry and exit with idiosyncratic risk are given by \begin{align}
        & m_t(z) = m_{t-1}(z)\mathbb{I}_{\{z\geq \underline z_t\}} + E_{\varepsilon,t}\mu^E(z)\mathbb{I}_{\{z\geq \underline z_t\}},\\
       & E_{\varepsilon,t}(\mathbb{E}_{\mu^E} [\max\{\mathbb{E}_{\psi}[\pi(z, m_t)], 0\}] - f^e) \leq 0 ,\\
        &E_{\varepsilon,t} \geq 0 ,\\
       & \mathbb{E}_{\psi} \pi(\underline{z} + \varepsilon, m_t) = 0,
    \end{align} %
and all other equilibrium equations are as in the baseline model. We note the following properties relative to our baseline model:

\begin{appprop}[Equilibrium with idiosyncratic risk]\label{prop:ext_idrisk_dyn}The following holds
\begin{enumerate}
    \item As in the baseline model, the cutoff, $\underline{z}$ is history-independent in a state of entry.
    
    \item An incumbent with permanent productivity component $z$ is part of the steady-state if and only if the probability that temporary productivity fluctuations cannot make them exit at the peak of the crisis: $
        \mathbb{P}_\psi(z + \varepsilon < \underline{z}) = 0$. The pre-crisis and long-run allocations differ if and only if \begin{align}\label{eq:ext_idrisk_pathdep}
        \int \mathbb{P}_\psi(z + \varepsilon < \underline{z}) \mu^E(z) \d z > 0.
    \end{align} 
    A sufficient condition for both these to be true is that there exists a lower bound $\underline \varepsilon$ high enough.
\end{enumerate}
\end{appprop}

An immediate consequence of Proposition \ref{prop:ext_idrisk_dyn} is that whenever eq. (\ref{eq:ext_idrisk_pathdep}) holds, Propositions \ref{proposition:PE-cycles}--\ref{prop:compreces}, which all exploit the path dependency of our model, still hold true.

An alternative scenario is one in which \textit{every} firm exits in some future, and incumbency plays no role in the steady state. For example, this would be the case whenever there are random exit shocks \citep{melitz_impact_2003,bilbiie2012endogenous}. With exogenous exit, the steady-state distribution becomes history-independent, and all the fixed cost cycle effects we discussed above become transitory. Nonetheless, the dynamics we described throughout would be visible through what \cite{caballero1996timing} refer to as ``echo'' effects. Namely, persistent fluctuations caused by past recessions that only vanish in the very long run.

\subsection{Multiproduct Firms}\label{app:multiprod}An important assumption of our model is that firms produce a single variety. As a consequence, firm exit necessarily implies the loss of a variety. Here, we study an extension of our baseline model where firms can produce multiple varieties. The key force that this extension introduces is that the exit of a firm may induce a surviving competitor to expand its variety set, leaving the number of offered goods in the economy unchanged. Effectively, this allows us to break the link between firm exit and variety losses. 

Consider an economy in which individual firms pay a fixed cost $f^p>0$ per variety they produce, and firm productivity is the same across varieties. Then, a firm with productivity $z$ produces a total number of varieties $N(z)$. To aggregate varieties $j \in [0, N(z)]$ within a firm, we use a CES-aggregator with an elasticity of substitution $\eta>1$: \begin{align}
    y(z) = \left(\int_0^{N(z)} y_{j}(z)^{\frac{\eta-1}{\eta}} \d j\right)^{\frac{\eta}{\eta-1}}.
\end{align} 
Varieties across firms are aggregated by a generalized CES aggregator with elasticity of substitution $\sigma \in (1, \eta)$. This aggregator features, in addition to the aggregate externality running through LoV for different firms, $q_F$, an externality term for LoV through products, $q_V$:%
\begin{align}\label{eq:aggr-multiprod}
    Y = M^{q_F - \frac{1}{\sigma-1}} \tilde{N}^{q_V - \frac{1}{\eta-1}} \left[\int y(z)^{\frac{\sigma}{\sigma-1}}m(z) \d z\right]^{\frac{\sigma-1}{\sigma}}.
\end{align} %
Here, the number of firms is $M:=\int m \d z$, and  $\tilde N$ is the number of products per firm: $\tilde{N} = N/M = \int N(z) \mu(z) \d z$. Since the cost to produce an extra variety is linear while the gains are concave, each firm chooses to operate an optimal mass $N(z)>0$ of varieties:
\begin{align}\label{eq:n(z)-multi-prod}
    N(z) = \frac{L^p}{(\eta-1) f^p} \frac{z^{\phi-1}}{\int  z^{\phi-1} m(z) \d z},
\end{align}
where $\phi -1 := \frac{(\sigma-1)(\eta-1)}{\eta-\sigma}$. Like individual output in our baseline model, the number of varieties of each firm is increasing in its relative productivity $z^{\phi - 1}/\left(\int  z^{\phi-1} m(z) \d z\right)^{\frac{1}{\phi-1}}$ and decreasing in the fixed cost $f^p$. By aggregation, the total number of varieties is given by \begin{align}
    N=\frac{L^p}{(\eta -1)f^p},   
\end{align}%
and therefore varieties scale linearly in the amount of available production labor. Equivalently, the ratio of labor expended on product fixed cost to that used for production, $Nf^p/L^p$ is constant and given by $1/(\eta-1)$. Intuitively, firms trade off allocating labor to producing more varieties versus more output of existing varieties. Since firms behave as if they lived in a world with CES aggregation, this trade-off is independent of scale, and $Nf^p/L^p$  is a constant fraction. To solve for $L^p$, we use the labor market clearing condition: 
\begin{align}\label{eq:mkt-clearing-multiprod}
    L^p + M f^c + E f^e + N f^p = \bar L,
\end{align}%
from which we conclude \begin{align}
    L^p &= \frac{\eta -1}{\eta}  (\bar L - M f^c) \quad\quad \text{ and} \quad\quad
    N =  \frac{(\bar L - M f^c)}{\eta f^p}.
\end{align} %
As before, we consider the steady state, i.e., $E=0$. By substituting the equilibrium production of each firm into the aggregator in eq. (\ref{eq:aggr-multiprod}), we can write aggregate output in equilibrium as
\begin{align}\label{eq:gdp-multiprod}
    Y 
      & \propto \overbrace{M^{q_F - q_V - \frac{1}{\phi - 1}}}^{= M^{\left(q_F - \frac{1}{\sigma-1}\right) - \left(q_V - \frac{1}{\eta-1}\right)}} (\bar L - Mf^c)^{1+q_V} \overbrace{\left( \int  z^{\phi - 1} m(z) \d z \right)^{\frac{1}{\phi - 1}}}^{=: \mathcal{Z}},
\end{align}
suppressing positive constants depending on $\eta$ and $f^p$ only.
Contrary to eq. (\ref{eq:gdp}), both the mass of firms and the mass of varieties matter for aggregate output, and we allow for two separate parameters capturing how much each is valued: $q_F$ and $q_V$, respectively. From eq. \eqref{eq:gdp-multiprod}, it is clear that cleansing effects are qualitatively the same, with a few remarks. First, for $q_V > 0$, there are increasing returns to scale to \textit{uncommitted labor}, $\bar L - M f^c$, because more uncommitted labor increases both the number of varieties per firm $\tilde{N}$, and firm output $y(z)$. Second, an increase in the number of firms $M$ is valued through two channels: a mechanical contraction of varieties per firm $\tilde{N}$ with elasticity $-(q_V - \frac{1}{1-\eta})$ and an expansion of firm-variety with elasticity $(q_F - \frac{1}{\sigma-1})$. 

With appropriately redefined boundaries for the inequalities in Propositions \ref{proposition:PE-cycles}--\ref{prop:compreces}, all our results on cleansing through fixed-cost cycles hold true. In particular, firm-level profits are
\begin{align}\label{eq:prof-multiprod}
    \pi( z, m) &= \frac{L^p}{\phi-1} \frac{z^{\phi-1}}{\int  z^{\phi-1} m(z) \d z} - f^c,
\end{align}
eqs. (\ref{eq:fe-pe}), (\ref{eq:zcp-pe}) and (\ref{eq:mu-lom}) yield exactly the same equilibrium conditions and entry/exit dynamics by replacing $\sigma$ with $\phi$ to account for within-firm variety substitution. Note that $f^p$ does not enter the profits of a firm. This is driven by the fact that changes in the cost of holding products make all firms shrink or expand their number of varieties, thereby not altering the effective relative productivities. As a consequence, movements in $f^p$ do not induce exit or entry since all firms move in lockstep and profits are unchanged. 
We formalize the long-run effects of business cycles driven by firm-level fixed costs $f^c$ in part (i) and driven by product-level fixed costs in part (ii) of Proposition \ref{prop:multi-prod-cycles}.

\begin{appprop}[Fixed-cost Cycles with Multi-Product Firms]\label{prop:multi-prod-cycles}
    In the presence of multi-product firms, the following holds true.
\begin{enumerate}[label=(\roman*)]
    \item For recessions driven by temporary increases in $f^c$, where $q_F^{CES} := \frac{1}{\sigma-1}$ and $q_V^{CES} := \frac{1}{\eta-1}$:
    \begin{enumerate}
        \item an analogue of Proposition \ref{proposition:PE-cycles} holds with 
        \begin{align*}
            M_3 < M_1, \; L_3^p = L_1^p, \; N_3 = N_1, \; \mathcal{Z}_3 = \mathcal{Z}_1
        \end{align*}
        and
        \begin{align*}
            \frac{Y_3}{Y_1} \gtreqqless 1 \Leftrightarrow q_F - q_F^{CES} \lesseqqgtr q_V - q_V^{CES};
        \end{align*}
        \item an analogue of Proposition \ref{prop:output-GE} holds with 
        \begin{align*}
            M_3 < M_1, \; L_3^p > L_1^p, \; N_3 > N_1, \; \mathcal{Z}_3 > \mathcal{Z}_1
        \end{align*}
        and there exists a unique $q_F^\star$ with
        \begin{align*} 
            q_F^\star - q_F^{CES} > q_V - q_V^{CES},
        \end{align*}
        such that
        \begin{align*}
              \frac{Y_3}{Y_1}\gtreqqless 1 \iff q_F \lesseqqgtr  q_F^\star;
        \end{align*}
        \item an analogue of Proposition \ref{prop:compreces} holds with
    \begin{align*}
        \exists! \; (q_\circ, q^\circ) \;\; \text{with} \;\; q_\circ - q_F^{CES} > q_V - q_V^{CES} 
    \end{align*}
    such that long-run effects depend on $f_h^c$ intensity.
    \end{enumerate}
    \item Recessions driven by temporary increases in $f^p$ have no long-run effects.
\end{enumerate}
\end{appprop}

The temporary increase in firm-level fixed costs induces exit on impact. In partial equilibrium, surviving firms temporarily increase their number of products. As the fixed cost reverts to its initial level, incumbents' product set shrinks while entrants join the economy. Overall, because the recession has cleansing effects, the total number of firms is lower, but the number of products per firm is unchanged. The long-run effects on output and welfare depend on whether LoV is high enough to compensate the loss of varieties, yielding a condition similar to that of Proposition \ref{proposition:PE-cycles} but accounting for both within and across firms substitutability. In GE, as fewer firms operate after the recession, labor is relatively cheaper due to smaller firm fixed costs payments. This increases the optimal number of products for surviving firms. Like in Proposition \ref{prop:output-GE}, there is a unique level of LoV that makes agents indifferent between going through a fixed cost cycle or not. This level now accounts for the partial recovery of lost varieties by exiting firms, since incumbents increase their product range after the recession. 

Finally, the setup of this extension also allows us to examine a second kind of fixed cost cycle: a cycle in $f^p$. Part (ii) states that $f^p$ cycles are akin to aggregate TFP cycles in that they have no long-run effects. The intuition is as follows. First, note that, differently from firm entry, product creation does not require a sunk investment, hence there is no notion of `product-incumbency'. Since we consider temporary increases in $f^p$, they only have long-run effects if they affect the firm distribution by inducing entry or exit. Firms reduce their product set while the recession occurs and create new products as $f^p$ returns to the original level; however, since all firms reduce their product set at the same time, there is no effect on market shares and, therefore, on profits. As a consequence, this recession does not feature exit, but rather just a reduction in the number of products per firm. This is because, in our economy, the presence of multiple products operates like a shifter to the productivity of each individual firm: if $z^\prime>z$, the effective productivity difference is even larger since the $z^\prime$ firm chooses to operate more products than the $z$ firm. However, the elasticity of $N(z)$ to $f^p$ is independent of $z$. Hence, when $f^p$ increases, all firms reduce the product set by the same percentage. This operates like a reduction in aggregate productivity and does not affect the relative productivity of firms. As in Proposition \ref{prop:ext_TFP}, this shift does not induce any exit and, therefore, does not alter the firm distribution before and after the recession. As a consequence, the economy reverts to the initial steady state.

\subsection{Fixed costs in output and labor units}\label{ext:fixed-cost-in-final-good} 

To preserve efficiency in the decentralized economy with $q=q^{CES}$, we have assumed that fixed and entry costs are paid entirely in units of labor \citep{barro_economic_2004}. To generalize our model, we let these costs be produced with the following Cobb–Douglas technology combining output (share $\alpha$) and labor (share $1-\alpha$):
\begin{align*}
        B(Y,L) \;&=\;\kappa(\alpha) \;  Y^{\alpha} L^{\,1-\alpha}, \qquad \kappa(\alpha)= \alpha^\alpha (1-\alpha)^{1-\alpha}, \quad  \alpha\in[0,1].
\end{align*}
This technology nests two special cases: $\alpha=0$, in which costs are paid entirely in units of labor as in our model, and $\alpha=1$, in which they are paid entirely in units of output. Throughout, we make the technical assumption that $\sigma>1+\alpha$. The following Proposition extends Proposition \ref{prop:output-GE} to this general setting.

\begin{appprop}[Cleansing Effects of Cycles]\label{prop:output-fin-good}
The change in output after the crisis, when the fixed costs are produced by a Cobb-Douglas bundle of output (share $\alpha$) and labor (share $1-\alpha$), is given by \begin{align}\label{eq:ge_ratio_fc_cb}
\Delta\log Y(q,\alpha)
&= \frac{\sigma-\alpha}{\sigma-1}\, \Delta\log \left(\int z^{\sigma-1} m(z)\d z\right)
+(1-\alpha)\,(\,q-q^{\mathrm{CES}}\,)\,\Delta\log M.
\end{align}
%
For each $\alpha\in[0,1]$ there exists a value of love-of-variety $q^\star(\alpha)$ such that the crisis leaves output unchanged:
$\Delta\log Y(q^\star(\alpha),\alpha)=0$; with $q^\star(\alpha)\ge q^{\mathrm{CES}}$, with equality for $\alpha=1$. 
Furthermore, output is decreasing in $q$ around $q^\star(\alpha)$: $\left.\frac{\partial}{\partial q}\Delta\log Y(q,\alpha)\right|_{q=q^{\star}(\alpha)}<0$. Which implies  
$$
q \gtrless q^{\star}(\alpha) \Rightarrow \Delta\log Y(q) \lessgtr 0 \quad \text{for $q$ in a neighborhood of $q^{\star}(\alpha)$, to first order}.
$$
\end{appprop}

The main results continue to hold when fixed costs are denominated in any combination of labor and output units. When the economy places greater weight on variety than under the CES benchmark ($q>q^\star(\alpha) \ge q^{CES}$), recessions reduce long-run output and welfare.

The intuition behind this result is as follows. The fixed-cost cycle shifts the economy from the original steady state to a new one in which the mass of active firms is lower. Aggregate output adjusts according to equation \eqref{eq:ge_ratio_fc_cb}, which depends on both the number of varieties and the firm productivity distribution. After the recession, the mass of firms shrinks, while the surviving firms are on average more productive. 

When $\alpha=0$, fixed costs are paid entirely in labor units, and thus $q$ does not affect the resources available to firms for production. The equilibrium trade-off then operates purely between variety and selection: for $q>q^\star$, the economy values variety more than average productivity, so the variety effect dominates the selection effect.

When $\alpha>0$, part of the fixed cost is denominated in units of output. A higher $q$ encourages entry, which increases output through the love-of-variety channel but also raises the nominal burden of fixed costs, reducing the scale of production among the most productive firms. When $\sigma$ is large, goods are close substitutes and concentrating production in the most efficient firms is particularly valuable. If $q>q^{\star}(\alpha)$, however, the economy places excessive weight on preserving variety, diverting capacity to sustain additional firms rather than allowing efficient producers to expand. This shift of production away from the strongest margin under high $\sigma$ implies that recessions reduce long-run output and welfare whenever $q>q^{\star}(\alpha)$.

\section{Proofs}\label{app:proofs}

\begin{proof}[Proof of Lemma \ref{lemma:hsa_invariance_A_z}]
    Since both equilibria are with entry, the free-entry condition and the zero-profit condition must both hold with equality. Jointly, they determine $A$ and $\underline{z}$. Since these two equilibrium equations are the same for $m_{t-1}$ and $m_{t-1}'$, the result follows.     
\end{proof}


\begin{proof}[Proof of Proposition \ref{prop:hsa_phase2}]
    Let $\olpsi = \underline{z}^{-1}$. Absence of entry ($E_2=0$) is obvious. From eq. \eqref{eq:hsa_zpc}, given that $\pi$ is decreasing in $\olpsi/A$, we note that $\overline{u}_2 = \olpsi_2/A_2 < \olpsi_1/A_1=\overline{u}_1$. For contradiction, suppose that the marginal cost cutoff has become more permissive, so $\olpsi_2 > \olpsi_1$. This would imply that $A_2 > A_1$, and competition has gone down. Then, consistency (eq. \ref{eq:hsa_lom}) and the fact that the revenue function $r \equiv s\circ X $ is strictly decreasing in its argument imply that the RHS of the adding-up constraint \eqref{eq:hsa_addup} must increase strictly, and that the adding-up constraint is violated. Therefore, $\olpsi_2 < \olpsi_1$. This implies that $m_2$ is a strict truncation of $m_1$, and the adding-up constraint \eqref{eq:hsa_addup} must decrease strictly, since $r > 0$, globally. To counter this, $A_2 > A_1$ increases.
\end{proof}

\begin{proof}[Proof of Theorem \ref{thm:hsa_tradeoff}]
    Using the adding-up constraint in $t=1$ and $t=3$, we obtain the entry mass post cycle: $
        E_3 = \frac{\int_{\ulz_1}^{\ulz_2} r(\frac{1}{zA_1}) m_1(z) \d z}{\int_{\ulz_1}^{\infty} r(\frac{1}{zA_1}) \mu^E(z) \d z}.
    $ %
    
    Then, the difference in the mass of firms, $M_3 - M_1$, is:
    \begin{align*}
        \Delta M &= -\int_{\ulz_1}^{\ulz_2} m_1(z)\d z + E_3 \Pr_{\mu^E}(z \geq \ulz_1) = -\int_{\ulz_1}^{\ulz_2} m_1(z)\d z + \frac{\int_{\ulz_1}^{\ulz_2} r(\frac{1}{zA_1}) m_1(z) \d z}{\E_{f_E}\Big[ r(\frac{1}{zA_1}) \Big]} \\
        &= \int_{\ulz_1}^{\ulz_2} m_1(z)\d z \left(\frac{\E_{\zeta_X}\Big[ r(\frac{1}{zA_1}) \Big]}{\E_{\zeta_E}\Big[ r(\frac{1}{zA_1}) \Big]} - 1 \right),
    \end{align*}
    where $\zeta_E(z)\equiv \frac{\mu^E(z)\,\mathbb{I}\{z\ge \ulz_1\}}{\Pr_{\mu^E}(z\ge \ulz_1)}$ is the successful-entrant probability density and $\zeta_X(z)\equiv \frac{m_1(z)\,\mathbb{I}\{\ulz_1\le z\le \ulz_2\}}{\int_{\ulz_1}^{\ulz_2} m_1(z)\,\d z}$ is the exiter probability density.
    This implies the equivalent characterization
    \begin{align*}
    \Delta M \gtreqqless 0
    \quad \Longleftrightarrow \quad
    \E_{\zeta_E}\Big[ r\Big(\frac{1}{zA_1}\Big) \Big] 	\lesseqqgtr \E_{\zeta_X}\Big[ r\Big(\frac{1}{zA_1}\Big) \Big].
    \end{align*}
    Let us define the pre and post-cycle productivity probability distributions
    $$
    \mu_1(z)\equiv \frac{m_1(z)}{M_1},
    \qquad
    \mu_3(z)\equiv \frac{m_3(z)}{M_3},
    $$
    where $M_t$ denotes the mass of active firms in period $t$.     Because phase 2 removes the slice $[\ulz_1,\ulz_2]$ and phase 3 adds entrants drawn from $\mu^E$, we can write $\mu_1$ and $\mu_3$ as mixtures with a common survivor component:
    \begin{align*}
    \mu_1(z) &= (1-w_X) \zeta_S(z) + w_X \zeta_X(z), 
    \qquad w_X \equiv \frac{\int_{\ulz_1}^{\ulz_2} m_1(z)\,\d z}{M_1},\\
    \mu_3(z) &= (1-w_E) \zeta_S(z) + w_E \zeta_E(z), 
    \qquad w_E \equiv \frac{E_3\Pr_{\mu^E}(z\ge \ulz_1)}{M_3},
    \end{align*}
    where $\zeta_S(z)\equiv \frac{m_1(z)\,\mathbb{I}\{z\ge \ulz_2\}}{\int_{\ulz_2}^{\infty} m_1(z)\,\d z}$ is the survivor density.
    Let the exiters FOSD successful entrants, i.e.,
    $$
    \zeta_X \ \text{FOSD}\ \zeta_E
    \qquad\Longleftrightarrow\qquad
    G_X(z)\le G_E(z)\ \ \ \forall z,
    $$
    where $G$ denotes the corresponding CDF. Since $\zeta_S$ is supported on $[\,\ulz_2,\infty)$ while $\zeta_X$ (and $\zeta_E$) place (weakly) more mass below $\ulz_2$, we have
    $$
    G_S(z)\le G_X(z)\le G_E(z)\qquad \forall z.
    $$
    Next, note that the sign characterization derived above implies $\Delta M>0$ under the assumed dominance. To see this, note that  $\tilde r(z)\equiv r(1/(zA_1))$ is strictly increasing in $z$, and thus FOSD implies
    $$
    \E_{\zeta_X}[\tilde r(z)] \;\geq\; \E_{\zeta_E}[\tilde r(z)],
    $$
    hence $\Delta M \geq 0$ and therefore $M_3\geq M_1$. Using the mixture weights, observe that
    $$
    w_E - w_X
    = \frac{E_3\Pr_{\mu^E}(z\ge \ulz_1)}{M_3} - \frac{\int_{\ulz_1}^{\ulz_2} m_1(z)\,\d z}{M_1}
    = \frac{\Delta M}{M_3} \left( \frac{\int_{\ulz_2}^{\infty} m_1(z) \d z}{M_1} \right),
    $$
    and thus $sign(w_E - w_X) = sign(\Delta M)$. Finally, compare the CDFs of $\mu_1$ and $\mu_3$:
    $$
    G_1(z)=(1-w_X)G_S(z)+w_X G_X(z),
    \qquad
    G_3(z)=(1-w_E)G_S(z)+w_E G_E(z),
    $$
    where $G_1$ and $G_3$ denote the corresponding CDF of $\mu_1$ and $\mu_3$.
    Then for every $z$,
    \begin{align*}
    G_3(z)-G_1(z)
    &= (w_E-w_X)\big(G_E(z)-G_S(z)\big) + w_X\big(G_E(z)-G_X(z)\big)
    \;\ge\; 0,
    \end{align*}
    because $w_E-w_X \geq 0$, $G_E(z)\ge G_S(z)$, and $G_E(z)\ge G_X(z)$ for all $z$. Hence, $G_3(z)\ge G_1(z)$ for all $z$, i.e.,
    $$
    \mu_1 \ \text{FOSD}\ \mu_3.
    $$
    With strict $FOSD$, i.e. $G_X(z) \leq G_E(z)$ for all $z$ with strict inequality on a set of positive measure, we obtain $\Delta M >0$ and $\mu_1$ strictly $FOSD$ $\mu_3$, so the post-cycle productivity distribution shifts left.
    
    The proof in the case in which the successful entrants distribution $FOSD$ the exiters distribution follows by repeating the same steps. 
\end{proof}

\begin{proof}[Proof of Theorem \ref{thm:hsa_cleansing}]
    Consider log output post cycle, recalling $\ulz_1 = \ulz_3$, and $A_1 = A_3$ by Lemma \ref{lemma:hsa_invariance_A_z}: \begin{align*}
        \log Y_3 &= \log \upsilon(M_3) + \log (\mathcal{I}) - \log(A_1) + \\ &\int_{\ulz_1}^\infty (s \circ X)\Big( \frac{1}{zA_1} \Big) \cdot (\Phi \circ X) \Big(\frac{1}{zA_1}\Big) \cdot \big( m_1(z) \mathbb{I}\{z \geq \ulz_2\} + E_3 \mu^E(z) \big) \d z,
    \end{align*}%
    and the difference with the first period log output $\Delta \log Y = \log Y_3 - \log Y_1$: $$
    \Delta \log Y = \Delta \log \upsilon(M) + \int_{\ulz_1}^\infty (s  \circ X)\Big( \frac{1}{zA_1} \Big) \cdot (\Phi \circ X) \Big(\frac{1}{zA_1}\Big) \cdot \big(E_3 \mu^E(z)  - m_1(z)\mathbb{I}\{z\in [\ulz_1, \ulz_2]\}\big) \d z.
    $$ %
    Using the adding-up constraint in $t=1$ and $t=3$, we obtain the entry mass post cycle: $$
        E_3 = \frac{\int_{\ulz_1}^{\ulz_2} r(\frac{1}{zA_1}) m_1(z) \d z}{\int_{\ulz_1}^{\infty} r(\frac{1}{zA_1}) \mu^E(z) \d z} = \frac{\Theta(\ulz_1, \ulz_2)}{\E_{\mu^E}\Big[ r(\frac{1}{zA_1}) \ \big|\ z \geq \ulz_1 \Big]\ \Pr_{\mu^E}(z \geq \ulz_1)}.
    $$ %
    Note that $\theta_1(z) \equiv r(1/(zA_1))m_1(z)$ is a density function with support $[\ulz_1, \infty)$ and that $\theta_1^E(z) \equiv \frac{r( \frac{1}{zA_1} )\mu^E(z)}{\int_{\ulz_1}^{\infty} r(\frac{1}{zA_1}) \mu^E(z) \d z}$ is a density function also with support $[\ulz_1, \infty)$. Therefore, the change in output after the recession is: {\footnotesize\begin{align*}
    \Delta \log Y &= \Delta \log \upsilon(M) + \int_{\ulz_1}^\infty r\Big( \frac{1}{zA_1} \Big) \cdot (\Phi\circ X) \Big(\frac{1}{zA_1}\Big) \cdot \bigg(  \frac{\int_{\ulz_1}^{\ulz_2} r(\frac{1}{zA_1}) m_1(z) \d z}{\int_{\ulz_1}^{\infty} r(\frac{1}{zA_1}) \mu^E(z) \d z} \mu^E(z)  - m_1(z)\mathbb{I}\{z\in [\ulz_1, \ulz_2]\}\bigg) \d z \\
    &= \Delta \log \upsilon(M) + \int_{\ulz_1}^\infty (\Phi\circ X) \Big(\frac{1}{zA_1}\Big) \cdot  \frac{\int_{\ulz_1}^{\ulz_2} r(\frac{1}{zA_1}) m_1(z) \d z}{\int_{\ulz_1}^{\infty} r(\frac{1}{zA_1}) \mu^E(z) \d z} r\Big( \frac{1}{zA_1} \Big)\mu^E(z) \d z
    \\ &\quad - \int_{\ulz_1}^{\ulz_2}  (\Phi\circ X) \Big(\frac{1}{zA_1}\Big) r\Big( \frac{1}{zA_1} \Big)  m_1(z) \d z \\
    &= \Delta \log \upsilon(M) + \Theta_1(\ulz_1, \ulz_2) \cdot \int_{\ulz_1}^\infty (\Phi\circ X) \Big(\frac{1}{zA_1}\Big)   \theta_1^E(z) \d z
    - \int_{\ulz_1}^{\ulz_2}  (\Phi\circ X) \Big(\frac{1}{zA_1}\Big) \theta_1(z)   \d z \\
    &= \Delta \log \upsilon(M) + \Theta_1(\ulz_1, \ulz_2) \bigg[ \E_{\theta_1^E} \Big[(\Phi\circ X) \Big(\frac{1}{zA_1}\Big) \Big]
    - \E_{\theta_1} \Big[  (\Phi\circ X) \Big(\frac{1}{zA_1}\Big) \ \Big|\ \ulz_2 \geq z \geq \ulz_1 \Big] \bigg].
    \end{align*}}
\end{proof}

\begin{proof}[Proof of Proposition \ref{prop:CES_HSA}]
    The first part is immediate from \citet{MatsuyamaUshchev2023_LoveForVariety}, and noting that if $\Phi = c$, then both expectations in \eqref{eq:hsa_cycle_Y} equal $c$, and hence they cancel. For the converse, note that the first expectation is constant while the second expectation varies in $\epsilon$ if $\Phi$ is not constant $\theta_1$-almost everywhere. Therefore, $\Phi$ must be constant $\theta_1$-a.e. and by continuity of $\Phi$, it must be constant overall.
\end{proof}
\begin{proof}[Proof of Proposition \ref{prop:hsa_cleansing_msld}]
    If $m_0=0$, then $$r(1/(zA_1))m_1(z) = r(1/(zA_1))\mu^E(z)E_1= r(1/(zA_1))\mu^E(z)\frac{1}{\int_{\ulz_1}^{\infty} r(\frac{1}{zA_1}) \mu^E(z) \d z},$$ using \eqref{eq:hsa_addup} to solve for $E_1$. Hence, $\theta_1 = \theta_1^E$. 
    
    Next, consult \citet{MatsuyamaUshchev2023_LoveForVariety} for the implication $\text{MSLD} \implies\ (\mathcal{L}^{int})' < 0 \ \forall M$. The equivalence $(\mathcal{L}^{int})' < 0 \ \forall M\iff \Phi'(x) < 0 \ \forall x$ follows from the definition of $\mathcal{L}^{int}$ in eq. \eqref{eq:hsa_lov_def}, where the market share function, $s$, is strictly decreasing (up to the choke price). The choke price $\bar x$ is defined as $\bar x := \inf \{ x>0\, |\, s(x)=0\}$.

    To prove the MLSD case, note that the first expectation in \eqref{eq:hsa_cycle_Y_noinc} is taken over larger values of $z$ than the second. For decreasing $\Phi$, its composition with $X(1/(zA_1))$, which is decreasing in $z$, is increasing. Hence, the first expectation in \eqref{eq:hsa_cycle_Y_noinc} is larger than the second, and $\Delta^{int}>0$ (recall here that $\theta_1^E(z)$ is not supported on values below $\ulz_1$.) Therefore, under no external variety effects, $\Delta \log Y = \Delta^{int} > 0$.

    The second result with regard to the opposite of MLSD follows analogous reasoning.
\end{proof}%
\begin{proof}[Proof of Proposition \ref{prop:hsa_cleansing_external_iso}]
    If $\upsilon(M) = e^{qM}$, then $\Delta \log \upsilon(M) = q \Delta M$. If $\upsilon(M) = M^q$, then $\Delta \log \upsilon(M) = q \Delta \log M$. Substituting these expressions in eq. (\ref{eq:hsa_cycle_Y}), and using the result on $\Delta M$ under Theorem \ref{thm:hsa_cleansing}, yields the result.
\end{proof} 

\begin{proof}[Proof of Remark \ref{rem:q}]
    Suppose wlog that $z_i$ is increasing in the firm index, $i \in [0, M]$. We can then perform the change of measure:
    $$
    V_M(\{y_i\}_{i\in [0,M]}) := M^{q - \frac{1}{\sigma -1}} \left(\int_{0}^M y_i^{\frac{\sigma-1}{\sigma}} \d i\right)^{\frac{\sigma}{\sigma -1}} = M^{q - \frac{1}{\sigma -1}} \left(\int_{z \in \mathcal{Z}} y(z)^{\frac{\sigma-1}{\sigma}} m(z) \d z\right)^{\frac{\sigma}{\sigma -1}},
    $$
    where the right-hand side is our aggregator. As in \cite{benassy1996taste}, we define %
    $
        \upsilon(M) := \frac{V_M(1)}{M} = \frac{M^{q - \frac{1}{\sigma -1}} M^{1 + \frac{1}{\sigma - 1}}}{ M } = M^{q}.
    $
    The function $v(M)$ characterizes the relative ``gain derived from spreading a certain amount of production between $M$ differentiated products instead of concentrating it on a single variety'' \citep{benassy1996taste}. The elasticity of $v(M)$ with respect to $M$ is defined as LoV, which indeed equals $q$.
\end{proof}

\versionshort{
\begin{proof}[Proof of Proposition \ref{proposition:PE-cycles}]
    The Proposition is an immediate Corollary of Theorem \ref{thm:hsa_cleansing} under CES, with $m_{t-1}(z)=c\,\mu^E(z) \; \forall z\in[\ulz_1,\ulz_2]$, for some constant $c > 0$. Recall, this is a simple sufficient condition to guarantee cleansing and $\Delta M < 0$. Since the $\Delta^{int}$ vanishes under CES, one only computes $\Delta^{ext}$ using $\upsilon(M)=M^{q - \frac{1}{\sigma-1}}$.
\end{proof}
}

\versionlong{
\begin{proof}[Proof of Proposition \ref{proposition:PE-cycles}]
    Start by noticing that \begin{align}
    \int z^{\sigma -1 } m_3 (z)\, \d z &= \int z^{\sigma -1 } \big(E_1 \mu^E \mathbf{1}_{\{ z \geq \underline z_2\}} + E_3 \mu^E \mathbf{1}_{\{ z \geq \underline z_3\}}\big) \d z \\
    &= E_1 \int_{\{ z \geq \underline{z}_2 \}}z^{\sigma -1 } \mu^E \d z + E_3 \int_{\{ z \geq \underline{z}_3 \}}z^{\sigma -1 } \mu^E \d z \\
    &= E_1 \int_{\{ z \geq \underline{z}_2 \}}z^{\sigma -1 } \mu^E \d z + E_1 \left(1 - \frac{\int_{\{ z \geq \underline{z}_2 \}}z^{\sigma -1 } \mu^E \d z}{\int_{\{ z \geq \underline{z}_1 \}}z^{\sigma -1 } \mu^E \d z}\right) \int_{\{ z \geq \underline{z}_3 \}}z^{\sigma -1 } \mu^E \d z \\
    &= E_1 \int_{\{ z \geq \underline{z}_1 \}}z^{\sigma -1 } \mu^E \d z \\
    &= \int z^{\sigma -1 } m_1 (z)\, \d z.
\end{align} Therefore, the factor $\int z^{\sigma -1 } m_\tau \, \d z$ remains unchanged in before and after the cycle. Next, consider some $a \geq \underline z_2 > \underline z_1$, then
\begin{align}
\frac{\int_a^{\infty} m_3(z) \d z}{\int m_3(z) \d z} &= \frac{E_1 \int_a^{\infty} \mu^E(z) \d z + E_3 \int_a^{\infty} \mu^E(z) \d z}{\int m_3(z) \d z} \\
&= \frac{E_1 \int_a^{\infty} \mu^E(z) \d z + E_1 \left(1 - \frac{\int_{\{ z \geq \underline{z}_2 \}}z^{\sigma -1 } \mu^E \d z}{\int_{\{ z \geq \underline{z}_1 \}}z^{\sigma -1 } \mu^E \d z}\right) \int_a^{\infty} \mu^E(z) \d z}{\int m_3(z) \d z} \\
&= \frac{E_1 \int_a^{\infty} \mu^E(z) \d z + E_1 \left(1 - \frac{\int_{\{ z \geq \underline{z}_2 \}}z^{\sigma -1 } \mu^E \d z}{\int_{\{ z \geq \underline{z}_1 \}}z^{\sigma -1 } \mu^E \d z}\right) \int_a^{\infty} \mu^E(z) \d z}{E_1 \int_{\underline{z}_2}^{\infty} \mu^E(z) \d z + E_1 \left(1 - \frac{\int_{\{ z \geq \underline{z}_2 \}}z^{\sigma -1 } \mu^E \d z}{\int_{\{ z \geq \underline{z}_1 \}}z^{\sigma -1 } \mu^E \d z}\right) \int_{\underline{z}_1}^{\infty} \mu^E(z) \d z} \\
&> \frac{E_1 \int_a^{\infty} \mu^E(z) \d z + E_1 \left(1 - \frac{\int_{\{ z \geq \underline{z}_2 \}}z^{\sigma -1 } \mu^E \d z}{\int_{\{ z \geq \underline{z}_1 \}}z^{\sigma -1 } \mu^E \d z}\right) \int_a^{\infty} \mu^E(z) \d z}{E_1 \int_{\underline{z}_1}^{\infty} \mu^E(z) \d z + E_1 \left(1 - \frac{\int_{\{ z \geq \underline{z}_2 \}}z^{\sigma -1 } \mu^E \d z}{\int_{\{ z \geq \underline{z}_1 \}}z^{\sigma -1 } \mu^E \d z}\right) \int_{\underline{z}_1}^{\infty} \mu^E(z) \d z} \\
&= \frac{\int_a^{\infty} \mu^E(z) \d z}{\int_{\underline{z}_1}^{\infty} \mu^E(z) \d z} \\
&= \frac{\int_a^{\infty} m_1(z) \d z}{\int m_1(z) \d z}.
\end{align}
Likewise, for $\underline z_1 \leq a < \underline z_2$
\begin{align}
1 - \frac{\int_a^{\infty} m_3(z) \d z}{\int m_3(z) \d z} &= \frac{\int_{0}^{a} m_3(z) \d z}{\int m_3(z) \d z} \\
&= \frac{E_1 \int_{a}^{a} \mu^E(z) \d z + E_1 \left(1 - \frac{\int_{\{ z \geq \underline{z}_2 \}}z^{\sigma -1 } \mu^E \d z}{\int_{\{ z \geq \underline{z}_1 \}}z^{\sigma -1 } \mu^E \d z}\right) \int_{\underline{z}_1}^{a} \mu^E(z) \d z}{E_1 \int_{\underline{z}_2}^{\infty} \mu^E(z) \d z + E_1 \left(1 - \frac{\int_{\{ z \geq \underline{z}_2 \}}z^{\sigma -1 } \mu^E \d z}{\int_{\{ z \geq \underline{z}_1 \}}z^{\sigma -1 } \mu^E \d z}\right) \int_{\underline{z}_1}^{\infty} \mu^E(z) \d z} \\
&\leq \frac{E_1 \left(1 - \frac{\int_{\{ z \geq \underline{z}_2 \}}z^{\sigma -1 } \mu^E \d z}{\int_{\{ z \geq \underline{z}_1 \}}z^{\sigma -1 } \mu^E \d z}\right) \int_{\underline{z}_1}^{a} \mu^E(z) \d z}{E_1 \left(1 - \frac{\int_{\{ z \geq \underline{z}_2 \}}z^{\sigma -1 } \mu^E \d z}{\int_{\{ z \geq \underline{z}_1 \}}z^{\sigma -1 } \mu^E \d z}\right) \int_{\underline{z}_1}^{\infty} \mu^E(z) \d z} \\
&= \frac{\int_{0}^{a} m_1(z) \d z}{\int m_1(z) \d z} = 1 - \frac{\int_a^{\infty} m_1(z) \d z}{\int m_1(z) \d z},
\end{align}
so that again $\frac{\int_a^{\infty} m_3(z) \d z}{\int m_3(z) \d z} \geq \frac{\int_a^{\infty} m_1(z) \d z}{\int m_1(z) \d z}$. Thus, the probability density function $\frac{m_3(z)}{\int m_3(z) \d z}$ strictly first-order dominates the probability density function $\frac{m_1(z)}{\int m_1(z) \d z}$. Since $g(z)=z^{\sigma-1}$ is increasing for $\sigma>1$,
FOSD implies
\begin{align}
    \frac{\int z^{\sigma-1} m_3(z)\,dz}{\int m_3(z)\,dz}
    \;>\;
    \frac{\int z^{\sigma-1} m_1(z)\,dz}{\int m_1(z)\,dz}.
\end{align}
Combining this with the first part of this proof, we obtain $\int m_1(z) \d z > \int m_3(z) \d z$.

Finally, notice that total income is exogenously fixed at $\mathcal{I}$. This implies that $L^p_3 = L^p_1 = \frac{\sigma-1}{\sigma} \mathcal{I}$. Therefore, we have that $\Delta \log L^p = 0$.

\end{proof}
}

\versionshort{
\begin{proof}[Proof of Proposition \ref{prop:ext_TFP}]
    We write the equilibrium conditions for an economy with aggregate productivity shock $\Xi$ as below and index all endogenous quantities by $\Xi$: \begin{align} \label{eq:ext_TFP_1}
        f^e &\geq \int_{\underline z_\Xi} \left[\frac{R_\Xi}{\sigma} \frac{(\Xi z)^{\sigma -1}}{\int_{\underline z_\Xi} (\Xi z)^{\sigma -1} m_\Xi(z) \d z} - f^c \right] \mu^E(z)\d z \\
        0 &= \frac{R_\Xi}{\sigma} \frac{(\Xi\underline z_\Xi)^{\sigma -1}}{\int_{\underline z_\Xi} (\Xi z)^{\sigma -1}m_\Xi(z) \d z} \label{eq:ext_TFP_2}
    \end{align}
    Note that $\Xi$ cancels in both equations above, such that the system is equivalent to the equilibrium equations of our baseline model if and only if $R_\Xi = R$. To establish this equality, note \begin{align}
        R_\Xi &= \bar L + \Pi_\Xi = \bar L + R_\Xi/\sigma - \left(\int m_\Xi(z) \d z\right) f^c \quad \text{(integrate $\pi_\Xi(z, m)$)}. \label{eq:ext_TFP_3}
    \end{align}
    Eqs. (\ref{eq:ext_TFP_1}--\ref{eq:ext_TFP_3}) are then exactly the same as in the baseline equilibrium, and thus any $(m, \underline z, R)$ satisfying the baseline equilibrium also satisfies (\ref{eq:ext_TFP_1}--\ref{eq:ext_TFP_3}).
\end{proof}
}

\versionlong{
\begin{proof}[Proof of Proposition \ref{prop:ext_TFP}]
    We write the equilibrium conditions for an economy with aggregate productivity shock $A$ as below and index all endogenous quantities by $A$: \begin{align} \label{eq:ext_TFP_1}
        f^e &\geq \int_{\underline z_A} \left[\frac{R_A}{\sigma} \frac{(Az)^{\sigma -1}}{\int_{\underline z_A} (Az)^{\sigma -1} m_A(z) \d z} - f^c \right] \mu^E(z)\d z \\
        0 &= \frac{R_A}{\sigma} \frac{(A\underline z_A)^{\sigma -1}}{\int_{\underline z_A} (Az)^{\sigma -1}m_A(z) \d z} \label{eq:ext_TFP_2}
    \end{align}
    Note that $A$ cancels in both equations above, such that the system is equivalent to the equilibrium equations of our baseline model if and only if $R_A = R$. To establish this equality, note \begin{align}
        R_A &= \bar L + \Pi_A = \bar L + R_A/\sigma - \left(\int m_A(z) \d z\right) f^c \quad \text{(integrate $\pi_A(z, m)$)}. \label{eq:ext_TFP_3}
    \end{align}
    Eqs. (\ref{eq:ext_TFP_1}--\ref{eq:ext_TFP_3}) are then exactly the same as in the baseline equilibrium, and thus any $(m, \underline z, R)$ satisfying the baseline equilibrium also satisfies (\ref{eq:ext_TFP_1}--\ref{eq:ext_TFP_3}). We can give more detail. Let $m^I$ be the predetermined incumbent density (in a dynamic context, $m_t^I=m_{t-1}$). In the case with entry, use $m_A = (E_A \mu^E + m^I) \mathbb{I}_{z\geq \underline z_A}$ in eqs. (\ref{eq:ext_TFP_1}--\ref{eq:ext_TFP_2}) to derive eq. (\ref{eq:cutoff}), so $\underline z_A = \underline z$. Use eq. (\ref{eq:ext_TFP_3}) in (\ref{eq:ext_TFP_2}) and the definition of $m_A$ to see that $E_A = E$, hence $m_A = m$. Clearly, then also $R_A = R$. In the case without entry, $E_A = 0 = E$ and eq. (\ref{eq:ext_TFP_2}) alone pins down $\underline z_A$, after substituting for $m_A$. Since this equation is again independent of $A$, $\underline z_A =\underline z$.
    
\end{proof}
}

\begin{proof}[Proof of Proposition \ref{prop:ext_entry}]
   First, note that the presence of a random exit shock $\delta$ implies that the stationary distribution of firms is unique and so is the number of firms $M^{\sst}$. As a consequence, there is a unique steady-state cutoff $\underline z^{\sst}$. 
   
   Without loss of generality, consider a case in which the entry cost increases and then steadily declines. Suppose that, for any arbitrary temporary increase in $f^e$, the cutoff $\underline z$ is non-increasing. Let $m^I$ be the predetermined incumbent density (in a dynamic context, $m_t^I=m_{t-1}$), and $\mu^I$ be the corresponding pdf. Then there are two cases
   \begin{enumerate}
       \item $f^e$ is large enough that the free entry condition is slack, and there is an initial part of the transition where $E=0$, the number of firms steadily declines until some new level $M^\prime$ such that the free entry condition holds exactly. From this point onwards, any further reduction in $f^e$ is associated with a strictly increasing number of firms until the entry cost returns to its initial level and $M\rightarrow M^{\sst}$ from below
       \item $f^e$ increases but not enough to offset the exit rate delta, $E>0$, along the entire transition and $M\rightarrow M^{\sst}$ from below.
   \end{enumerate}
   Consider now the behaviour of the cutoff. Again, we consider the same two cases
   \begin{enumerate}
       \item If the increase is large enough that $E=0$ in the early periods of the transition, we know from the exit case above that the cutoff remains constant, since it is determined by 
       \begin{align*} 
          f^c &= \frac{\mathcal{I}}{\sigma} \frac{\underline z^{\sigma-1}}{\int_{\underline z}^\infty  z^{\sigma -1} \mu^I(z) \d z}
        \end{align*}
       which is independent of $f^e$.
       
       As soon as the entry cost declines enough to have $E>0$, the cutoff is determined by 
       \begin{align*}
            \frac{f^e}{f^c} &= \int_{\underline z}^\infty \left[\left(\frac{z}{\underline z}\right)^{\sigma-1} -1 \right]\mu^E(z) \d z 
       \end{align*}
       which implies that the cutoff jumps to some new level $\underline z^\prime<\underline z$ since the RHS is strictly decreasing in $\underline z$. As the fixed entry cost keeps decreasing back towards the original level, $\underline z^\prime$ converges to $\underline z^{\sst}$ from below.
       \item If the entry cost increase is small enough that $E>0$ throughout the transition, then the cutoff jumps to some level $\underline z^\prime<\underline z$ on impact and then converges back to $\underline z^{\sst}$ from below. 
   \end{enumerate}

   This verifies the original guess and proves that along the transition, the cutoff decreases and so does the number of firms. 

   To prove part \textit{c)}, note that since both $\underline z$ and $M$ decrease along the transition and since exit is \textit{unselected}, the average productivity of firms has to decrease as well. 
   Since $M$ and $\underline z$ both decline, then output has to decrease independently of $q$.
\end{proof}

\begin{proof}[Proof Proposition \ref{prop:output-GE}]
    Note first that in the absence of fixed cost of entry, $f^e$, eq. \eqref{eq:entrymass-ge} determines $E^{\text{ss}}$ directly: there is no additional labor freed up from lower entry expenses. Hence, the entry process terminates immediately and jumps to $E^{\text{ss}}$. Rewrite eq. \eqref{eq:entrymass-ge} now as
\begin{align}
    E_3 &= \frac{R_3}{\sigma f^c} \frac{\underline{z}_1^{\sigma-1}}{\int_{\{ z \geq \underline{z}_1 \}}z^{\sigma -1 } \mu^E(z) \d z} - E_1 \; \frac{\int_{\{ z \geq \underline{z}_2 \}}z^{\sigma -1 } \mu^E(z) \d z}{\int_{\{ z \geq \underline{z}_1 \}}z^{\sigma -1 } \mu^E(z) \d z} \\
    &= \frac{R_1}{\sigma f^c} \frac{\underline{z}_1^{\sigma-1}}{\int_{\{ z \geq \underline{z}_1 \}}z^{\sigma -1 } \mu^E(z) \d z} - E_1 \; \frac{\int_{\{ z \geq \underline{z}_2 \}}z^{\sigma -1 } \mu^E(z) \d z}{\int_{\{ z \geq \underline{z}_1 \}}z^{\sigma -1 } \mu^E(z) \d z} + (R_3 - R_1) \frac{1}{\sigma f^c} \frac{\underline{z}_1^{\sigma-1}}{\int_{\{ z \geq \underline{z}_1 \}}z^{\sigma -1 } \mu^E(z) \d z} \\
    &= E_1 \left(1 - \frac{\int_{\{ z \geq \underline{z}_2 \}}z^{\sigma -1 } \mu^E(z) \d z}{\int_{\{ z \geq \underline{z}_1 \}}z^{\sigma -1 } \mu^E(z) \d z} \right) + (R_3 - R_1) \frac{1}{\sigma f^c} \frac{\underline{z}_1^{\sigma-1}}{\int_{\{ z \geq \underline{z}_1 \}}z^{\sigma -1 } \mu^E(z) \d z} \label{eq:proof_prop2_E3}
\end{align}  
where $R_\tau$ is steady-state GE income (i.e., without fixed cost of entry). Note that \begin{align}
    (R_3 - R_1) &= \frac{\sigma}{\sigma-1}f^c(E_1 p_E(\underline z_1) - (E_1 p_E(\underline z_2) + E_3 p_E(\underline z_1 ))) = \frac{\sigma}{\sigma-1}f^c(M_1 - M_3). \label{eq:proof_prop2_diffR}
\end{align}   %
Suppose that $ (R_3 - R_1) \leq 0$, then $M_3 \geq M_1$. Since $E_3 \leq E_1 \left(1 - \frac{\int_{\{ z \geq \underline{z}_2 \}}z^{\sigma -1 } \mu^E(z) \d z}{\int_{\{ z \geq \underline{z}_1 \}}z^{\sigma -1 } \mu^E(z) \d z} \right)$, following the steps of the proof of Proposition \ref{proposition:PE-cycles}, we know that $\int z^{\sigma -1 } m_3 (z)\d z \leq \int z^{\sigma -1 } m_1 (z)\d z$. But then, it also holds that $\int z^{\sigma -1 } \mu_3 \d z \leq \int z^{\sigma -1 } \mu_1 \d z$, i.e., the average productivity of a firm active in the market is lower post cycle. By construction of the exit and entry process (clipping the left tail of the productivity distribution in period 2 and allowing entry of average firms in period 3) this cannot be. Hence,  $ (R_3 - R_1) > 0$ and thus \begin{align}
    M_3 < M_1,\; \text{and} \; \int z^{\sigma -1 } m_3 (z)\d z > \int z^{\sigma -1 } m_1 (z)\d z.
\end{align} %
Consequently, $L_3^p > L_1^p$ holds, too. The output ratio between pre- and post-crisis steady-state depending on $q$ is then given by \begin{align}
    \frac{Y_3}{Y_1}(q) =  \left[\frac{M_3}{M_1}\right]^{q - \frac{1}{\sigma-1}} \; \left[\frac{L_3^p}{L_1^p}\right] \frac{(\int z^{\sigma -1 } m_3 (z)\d z)^{1/(\sigma-1)}}{(\int z^{\sigma -1 } m_1 (z)\d z)^{1/(\sigma-1)}},
\end{align} and none of the equilibrium quantities on the right-hand-side depend on $q$. It follows that $\frac{Y_3}{Y_1}\big(\frac{1}{\sigma-1}\big) > 1$ for the CES case, and there exists a unique $q^* > \frac{1}{\sigma-1}$ for which $\frac{Y_3}{Y_1}(q^*)=1$.
\end{proof}

\begin{proof}[Proof Proposition \ref{prop:SP_solution}]
We start by characterizing the optimal policy in the case of entry and then study the one of pure exit. For lighter notation, we suppress the time subscripts of \eqref{eq:SP_obj}. Accordingly, let $m^I=m_{t-1}$ be the density of incumbents, and $I=\int m^I$ the incumbent mass. Let also $p_E(\underline z):=\int_{\underline z}^\infty \mu^E(z) dz$. \\

\noindent\textbf{Entry:}\hspace{0.3cm} 
    The interior first-order condition on entry is %
\begin{align}\label{eq:foc_SP_E}
    E^{SP}(\underline z; m^I, q)=\dfrac{\bar L-If^{c}}{p_E(\underline z) f^{c}+f^{e}}\dfrac{1}{\sigma }-\dfrac{\int_{\{ z \geq \underline{z} \}}z^{\sigma -1 } m^I \d z}{\int_{\{ z \geq \underline{z} \}}z^{\sigma -1 } \mu^E(z) \d z}\dfrac{\sigma-1}{\sigma}+ \Delta E^{SP},
\end{align} %
%
which equals the CES GE-entry plus a difference term, $\Delta E^{SP}$. The latter is given by \begin{align}
    \Delta E^{SP} 
    &= \left(q - \frac{1}{\sigma - 1}\right) L^p \frac{\sigma}{\sigma - 1} (p_E( \underline z) f^{c}+f^{e})^{-1} Q(m^I, \underline z),
\end{align} %
where $Q(m^I, \underline z)$ is the ratio between the average productivity of firms post-entry and the average productivity of (successful) entrants. It is immediate that if the household has CES preferences ($q=1/(\sigma-1)$), then $\Delta E^{SP}=0$, and the market outcome is constrained efficient. Over time, a stock of firms builds up (see discussion in Section \ref{sec:fixec_cost_cycles}), so $\sum_{t\geq 0}E_t^{SP} \rightarrow E^{SP, ss}$, where we call $E^{SP, ss}$ the (steady-state) total mass of entrants for the social planner problem. As $E_t^{SP}$ tends to $0$, $Q$ tends to 1, and the incumbent mass building up is, eventually, $\tilde{I} = I + E^{SP, ss}$. (Similarly, the measure of incumbents is stacked up.) Substituting these into eq. (\ref{eq:foc_SP_E}) and rearranging, we find \begin{align}
    E^{SP, ss} = \frac{(\bar L -If^c)q}{f^c p_E(\underline z)(1+q) + f^e} - \frac{\int_{\{ z \geq \underline{z} \}}z^{\sigma -1 } m^I \d z}{\int_{\{ z \geq \underline{z} \}}z^{\sigma -1 } \mu^E(z) \d z}\frac{f^c p_E(\underline z) + f^e}{f^c p_E(\underline z)(1+q) + f^e}.
\end{align} %
The socially optimal number of entrants is increasing in taste-for-varieties, $q$, and decreasing in the mass and productivity of active incumbents. Fixed costs of entry, $f^e$, stop to matter if and only if $q\rightarrow\infty$. Comparing this to eq. (\ref{eq:entrymass-ge}), it so follows that the general equilibrium steady-state mass converges to $E^{SP, ss}$ if and only if $q = \frac{1}{\sigma - 1}$. It also follows that in an economy with high $q > \frac{1}{\sigma - 1}$, the equilibrium number of entrants is too small compared to the short-run optimal number, so $E^{GE} < E^{SP}$. Conversely, for small $q$, the equilibrium number of varieties is larger than the short-run efficient number.  

Using the first order condition on $E^{SP}$, the interior first order condition on $\underline z^{SP} = \underline z^{SP}(E; q, m^I)$ yields: %
{\footnotesize\begin{align} \label{eq:foc_SP_z} 
    & (\bar L - If^c) \left(\int_{\{ z \geq \underline{z} \}}z^{\sigma -1 } \mu^E(z) \d z\right) \left( \frac{f^e}{f^cp_E(\underline z) + f^e} \frac{1}{\sigma}\frac{J}{K} - \left[1 - \frac{\underline z^{\sigma -1}}{\E_{\mu^E}[z^{\sigma -1 } \g  z \geq \underline z]}\right]\frac{1}{\sigma}\frac{\sigma - \frac{J}{K}}{\sigma-1} \right) \notag\\
    &= \left(\int_{\{ z \geq \underline{z} \}}z^{\sigma -1 } m^I \d z\right) (f^cp_E(\underline z) + f^e) \left( \frac{f^e}{f^cp_E(\underline z) + f^e} \left(1 - \frac{\sigma-1}{\sigma}\frac{1}{K}\right) + \frac{1}{\sigma}\frac{1}{K}  \left[1 - \frac{\underline z^{\sigma -1}}{\E_{\mu^E}[z^{\sigma -1 } \g  z \geq \underline z]}\right] \right).
\end{align}} %
The fraction $J/K$ and $J$ and $K$ individually take value $1$ precisely when $q = \frac{1}{\sigma - 1}$. Only in this case, are the terms in the large parentheses both equivalent to {\footnotesize$\dfrac{f^{e} \int_{\{ z \geq \underline{z} \}}z^{\sigma -1 } \mu^E(z) \d z}{\delta_{f}} - \left( \int_{\{ z \geq \underline{z} \}}z^{\sigma -1 } \mu^E(z) \d z- \underline z^{\sigma -1} p_{E}(\underline z) \right)$}. Eq. (\ref{eq:foc_SP_z}) is then equivalent to eq. (\ref{eq:cutoff}), the GE cutoff.

With external variety effect, unpack $J$ and $K$ to determine the cutoff. Note that $J = 1 + [(\sigma - 1)q - 1] Q$ and $K = 1 + \frac{1}{\sigma} [(\sigma - 1)q - 1] Q$. If we approach a steady state (or drop the incumbents), $Q\rightarrow 1$. Then eq. (\ref{eq:foc_SP_z}) simplifies:%
\begin{align}\label{eq:foc_SP_z_noI}
    \frac{f^e}{f^cp_E(\underline z) + f^e} (\sigma - 1) q - \left[1 - \frac{\underline z^{\sigma -1}}{\E_{\mu^E}[z^{\sigma -1 } \g  z \geq \underline z]}\right] = 0,
\end{align}%
and under mild regularity conditions on $\mu^E$, we can conclude that $\frac{\partial \underline z}{\partial q} < 0$. Intuitively, if the economy has a strong preference for varieties, the entry barriers for new establishments should be lower. Eq. (\ref{eq:foc_SP_z_noI}) also holds with incumbents if they were drawn from the same distribution, i.e., if $m^I \propto \mu^E(z) \mathbb{I}( z \geq \underline z)$. \\

\noindent\textbf{Exit:}\hspace{0.3cm} 
Suppose incumbent firms are distributed according to $m^I = M_0 \mu^E(z) \mathbb{I}(z \geq \underline z_0)$, and that fixed costs $f^c$ are such that the optimal cutoff $\underline z$ is larger than $\underline z_0$, and that optimal entry, therefore, is nil. Hence, the relevant production function to maximize is \begin{align}
    Y &= \left( M_0\right)^{q}p_{E}( \underline z)^{q-\frac{1}{\sigma - 1}}\left( \bar L- f^c( M_0p_{E}( \underline z))\right)  \left(\int_{\{ z \geq \underline{z} \}}z^{\sigma -1 } \mu^E(z) \d z\right)^{\frac{1}{\sigma-1}},
\end{align} %
where taking first order conditions yields \begin{align}\label{eq:foc_SP_exit}
    f^c M_0 p_E(\underline z^{SP}) = \frac{\bar L - f^c M_0p_E(\underline z^{SP})}{\sigma -1} \left( [q (\sigma - 1)-1] + \frac{(\underline z^{SP})^{\sigma -1}}{\E_{\mu^E}[z^{\sigma -1 } \g  z \geq \underline z^{SP}] }\right),
\end{align}
which is equal to the market equilibrium outcome only if $q = \frac{1}{\sigma -1}$. For larger $q$, the cutoff is again optimally set lower than in the equilibrium allocation.
\end{proof}

\begin{proof}[Proof of Proposition \ref{prop:steady-state-subsidy}]
    In the steady-state targeted by the social planner, eq. (\ref{eq:foc_SP_exit}) must hold once entry has subsided, i.e., with $E^{SP}=0$ and some $\underline z^{SP}$. We now evaluate eq. (\ref{eq:zcp-pe}) at this $\underline z^{SP}$, whereby we have swapped $f^c$ for $\delta^cf^c$ beforehand. Solving for $\delta^c$ yields the result.
\end{proof}

\begin{proof}[Proof of Lemma \ref{lem:emp_elasticities}]
    Since we are using the entry equation for positive entry, all derivatives are right-derivatives. Let $E^I$ be the number of `old' entrants. That is, the mass of incumbent firms that have tried, and of which $E^I p_E(z^*)$ have succeeded to enter the market. %
    By definition, $M = Ep_E(\underline z) + E^I p_E(z^*)$. The entry condition determines $E = E(\mathcal I)$ such that
    \begin{align}
       Ep_E(\underline z) + E^I p_E(z^*) &= \frac{\mathcal I}{\sigma f^c} \frac{\underline z^{\sigma - 1}}{\E[z^{\sigma - 1} \g z \geq \underline z]} - E^I p_E(z^*)\left( 1 - \frac{\E[z^{\sigma - 1} \g z \geq  z^*]}{\E[z^{\sigma - 1} \g z \geq \underline z]}\right). 
   \end{align} %
       Letting %
   \begin{align}
       k(\underline z, z^*)&= p_E(z^*)\left( 1 - \frac{\E[z^{\sigma - 1} \g z \geq  z^*]}{\E[z^{\sigma - 1} \g z \geq \underline z]}\right),
    \end{align} %
    we obtain $
    k(\underline z, z^*) \gtreqqless 0 \iff \ulz \gtreqqless z^*
    $. %
    One now takes elasticities with respect to $\mathcal I$ to obtain the expression for $\frac{\partial_+\log M}{\partial_+ \log \mathcal I}$. 
    Next, substitute the distribution post-entry \begin{align}
        \frac{E p_E(\underline z)\mu^E(z \g z \geq \underline z) + E^I p_E( z^*)\mu^E(z \g z \geq z^*)}{E p_E(\underline z) + E^I p_E( z^*)}
    \end{align} %
    into the integral of the average productivity, $\log \bar z$. After some algebra to write integrals as expectations, substitute out $E(\mathcal I)$ to obtain \begin{align}
        \log \bar z &= \text{\footnotesize$\frac{1}{\sigma-1}$}\left( \log \mathcal I - \log \sigma f^c + \log \underline z - \log[ p_E(\underline z)E_{\g E^I=0} + E^I\, k(\underline z, z^*)]  \right) \\
         &= \text{\footnotesize$\frac{1}{\sigma-1}$}\left( \log \mathcal I - \log \sigma f^c + \log \underline z - \log M  \right)
    \end{align} with
    \begin{align}
        E_{\g E^I=0} &= \frac{\mathcal I}{\sigma f^c}\frac{\underline z^{\sigma -1}}{\E[z^{\sigma - 1}\g z \geq \underline z]p_E(\underline z)}.
    \end{align} %
    Differentiating with respect to $\log \mathcal I$ yields the result.
\end{proof}

\begin{proof}[Proof of Proposition \ref{prop:identification}] 
    From Lemma  \ref{lem:emp_elasticities}, $k(\underline z, \underline z) = 0$. Hence, $\frac{\partial_+\log M}{\partial_+ \log \mathcal I}=1 $ and $\frac{\partial_+\log \bar z}{\partial_+ \log \mathcal I} = 0$.
\end{proof}

\begin{proof}[Proof of Remark \ref{identification_homog}]
    We prove the general statement, and the special case of constant returns to scale (homogeneity of degree 1) follows. First, we prove that if output changes have a representation like eq. \eqref{eq:empirical_logY} under homogeneity of degree $\chi$, then we identify $\chi+q$. Then we prove that such a representation exists.

    First, we prove the proportionality of input demands.  Consider a production function $f(\mathbf{x})$ homogeneous of degree $\chi$, increasing, strictly concave, and twice differentiable in all inputs. Then, we can write generic implicit input demands as $f_{x_i} = w_i$, where $w_i$ is the unit price of input $x_i$. Taking ratios of the FOCs $f_{x_i}/f_{x_j} = \psi(x_i/x_j)=w_i/w_j$, which we can invert to obtain $x_i = \phi(w_i/w_j)x_j,\,\forall i,j$. This proves that all inputs are perfectly collinear, provided that relative input prices do not change. To prove the second part of the statement, we substitute input demands into the production function, $f(x_1,...,x_N) = (\phi_{1j}x_j,...,\phi_{Nj}x_j)$. By homogeneity of degree $\chi$, we obtain $f(x_1,...,x_N) =x_j^\chi f(\phi_{1j},...,\phi_{Nj})$. This directly implies that the output elasticity to a single input is 
    \begin{align*}
        \frac{\partial \log y}{\partial\log x_j} = \chi.
    \end{align*}
    This proves that if we can write output changes in an equivalent way to eq. \eqref{eq:empirical_logY} for general homogeneous production functions, the regression estimates $\hat\beta = \chi+q$. To prove that such a representation exists, note that by the properties of homogeneous functions, the marginal cost associated to the cost minimizing input bundle takes the form $c(z) = \frac{\kappa}{z \chi}\left(\frac{y(z)}{z}\right)^{\frac{1}{\chi}-1}$, where $\kappa$ is some function of the relative input prices and parameters which the firm takes as given. By the profit maximization, then the firm charges an optimal price $p(z) = \frac{\sigma}{\sigma-1}c(z).$ Using the demand function and solving for profits we obtain $\pi(z,m) = \frac{\mathcal I}{\sigma}\frac{c(z)^{1-\sigma}}{\int_z c(z)^{1-\sigma}m(z)\d z }-f^c.$ From the demand function we can solve output and find that $y(z)\propto z^{\frac{\sigma}{\chi+\sigma(1-\chi)}}$. From this, the steps used in the proof of Lemma \ref{lem:emp_elasticities} apply since we can write profits as $\pi(z,m) =  \frac{\mathcal I}{\sigma}\frac{z^{f(\sigma,\chi)}}{\int_z z^{f(\sigma,\chi)} m(z)\d z }-f^c$, where $f(\sigma,\chi)$ is some constant that depends on the degree of homogeneity of the production function and the elasticity of substitution.
\end{proof}

\begin{proof}[Proof of Remark \ref{rem:ident_pindex}]
    Consider the methodology described in Appendix \ref{sec:salesVoutput} to obtain real output from nominal output and price indices in WIOD. Define sales $S_t = P_tY_t$. The data either reports sales at previous year prices correctly accounting for the variety effect $S_t^{t-1} =P_{t-1}Y_t$ or sales at previous year prices using a deflator that ignores the variety effect $\tilde S^{t-1}_t = \frac{S_t^t P_{t-1}^{CES}}{P_t^{CES}}$, with $P_t^{CES} = P_t M_t^{q-\frac{1}{\sigma-1}}$. When the statistical agencies report deflators that properly account for the variety effect, then real output growth can be decomposed as:
    \begin{align}\label{eq:delta-log-yt-corr-defl}
       \Delta \log Y_t &= \log S_{t}^{t-1} - \log S_{t-1}^{t-1} = \log P_{t-1} + \log Y_t - \log P_{t-1} - \log Y_{t-1} \notag \\
        &= q \Delta \log M_t + \Delta \log L^p_t + \Delta \log \bar z_t.
    \end{align}
    When the statistical agencies ignore the variety effect, then real output growth can be decomposed as:
    \begin{align}\label{eq:delta-log-yt-incorr-defl}
         \Delta \log \tilde Y_t &= \log \tilde S_{t}^{t-1} - \log S_{t-1}^{t-1} = \log P_{t} + \log Y_t + \log P_{t-1}^{CES} - \log P^{CES}_{t} - \log P_{t-1} - \log Y_{t-1} \notag \\
        &= \log Y_t + \left ( q - \frac{1}{\sigma-1} \right) \log M_{t-1} - \left ( q - \frac{1}{\sigma-1} \right) \log M_t - \log Y_{t-1} \notag \\
        &= q \Delta \log M_t + \Delta \log L^p_t + \Delta \log \bar z_t - \left ( q - \frac{1}{\sigma-1} \right) \Delta \log M_{t}\notag \\
        &= \frac{1}{\sigma-1} \Delta \log M_{t} + \Delta \log L^p_t + \Delta \log \bar z_t.
    \end{align}
    Let the fixed-cost technology be Cobb–Douglas in the final good and labor:
    \begin{align*}
            B(Y,L) \;&=\;\kappa(\alpha) \;  Y^{\alpha} L^{\,1-\alpha}, \qquad \kappa(\alpha)= \alpha^\alpha (1-\alpha)^{1-\alpha}, \quad  \alpha\in[0,1].
    \end{align*}
    Let $(P_t,w_t)$ denote the prices of the final good and labor, then the cost of the bundle $B$ is
    \begin{align*}
            \Phi_t(P_t,w_t) \;&=\; P_t^{\alpha} w_t^{\,1-\alpha}.
    \end{align*}
    The firm requires $f^c$ units of $B$ per period. With labor as the numeraire ($w_t=1$), the nominal expenditure on fixed cost payments at time $t$ is
    \begin{align*}
        f^c \cdot \Phi_t(P_t,1)
        \;&=\;
        f^c\, P_t^{\alpha}.
    \end{align*}
    This nests two special cases. If $\alpha=0$, then the fixed costs are paid only in units of labor, and the nominal expenditure is $w_t f^c= f^c$. If $\alpha=1$, then the fixed costs are paid only in final-good units, and the nominal expenditure is $f^c P_t$.

Let the economy be in a state of entry, in line with our estimation strategy that isolates periods of positive demand shocks. Then, the mass of varieties in the economy is
\begin{align}
    &\frac{L^p_{t} \; \underline z_t^{\sigma-1}}{(\sigma-1) \; P_t^\alpha \;  f^c} = \mathcal{Z}_t^{\sigma-1} %
    \iff \frac{L^p_t}{\sigma^\alpha (\sigma-1)^{1-\alpha}} M_t^{\alpha(q-\frac{1}{\sigma-1})} \mathcal{Z}_t^{\alpha+1-\sigma} \underline z_t^{\sigma-1} =  f^c \iff \notag \\
    &\frac{L^p_t \underline z_t^{\sigma-1}}{\sigma^\alpha (\sigma-1)^{1-\alpha}} M_t^{\alpha(q-\frac{1}{\sigma-1})+\frac{\alpha+1-\sigma}{\sigma-1}} \bar z_t^{\alpha+1-\sigma}  =  f^c 
    \iff M_t = \left ( \frac{\sigma^\alpha (\sigma-1)^{1-\alpha}  f^c}{L^p_t \underline z_t^{\sigma-1}} \bar z_t^{\sigma-1-\alpha} \right)^{\frac{1}{\alpha q - 1}}\label{cobb-douglas-fc}
\end{align}
As $\alpha \to 1$, we obtain the mass of varieties in an economy where fixed costs are paid exclusively in final-good units, that is
$
    M_t = \left ( \frac{\sigma  f^c}{L^p_t \underline z_t^{\sigma-1}} \bar z_t^{\sigma-2} \right)^{\frac{1}{q - 1}}.
$
As $\alpha \to 0$, we obtain the mass of varieties in an economy where fixed costs are paid exclusively in units of labor, that is
$
    M_t = \frac{L^p_t}{(\sigma-1)f^c} \frac{\underline z_t^{\sigma-1}}{\bar{z}_t^{\sigma-1}} \iff f^c = \frac{L^p_t}{\sigma-1} \frac{\underline z_t^{\sigma-1}}{\int z^{\sigma-1} m_t(z) dz}.
$

Taking log differences of eq. (\ref{cobb-douglas-fc}), we obtain
\begin{align}\label{eq:delta-log-mt-alpha}
    \Delta \log M_t &= \frac{1}{\alpha q - 1} \left[ - \Delta \log L^p_t - (\sigma-1) \Delta \log \underline z_t + (\sigma-1-\alpha) \Delta \log \bar z_t \right].
\end{align}%
Substituting $ \Delta \log M_t$, eq. (\ref{eq:delta-log-mt-alpha}), into $\Delta \log Y_t$ and $\Delta \log \tilde Y_t$, eq. (\ref{eq:delta-log-yt-corr-defl}) and eq. (\ref{eq:delta-log-yt-incorr-defl}), we obtain
\begin{align}
    \Delta \log Y_t &=  \left(1-\frac{q}{\alpha q -1} \right) \Delta \log L^p_t - \frac{q(\sigma-1)}{\alpha q -1} \Delta \log \underline z_t + \frac{q(\sigma-1)-1}{\alpha q -1} \Delta \log \bar z_t, \notag \\
    \Delta \log \tilde Y_t &= \left(1-\frac{1}{(\sigma-1) (\alpha q -1)} \right) \Delta \log L^p_t - \frac{1}{\alpha q -1} \Delta \log \underline z_t + \frac{\alpha q(\sigma-1)-\alpha}{(\sigma-1) (\alpha q -1)} \Delta \log \bar z_t. \notag
\end{align}
Given that $L^p_t = \frac{\sigma-1}{\sigma} \mathcal{I}_t$, the above equations hold for both $\mathcal{I}_t$ and $L_t^p$. Furthermore, positive income shocks do not move $\underline z_t$, and by Assumption \ref{ass:entrant_incum_dist} they also do not move $\bar z_t$. More specifically, $\Delta \log X_t >0 \implies \Delta \log \underline z_t = 0$ and $\Delta \log \bar z_t = 0$, where $X_{t}\in \{\mathcal I_{t}, L^p_{t}\}$. Substituting these conditions and indexing by industry $i$ yields the regression specifications in the Remark:
\begin{align}
        \Delta\log Y_{i,t} = \beta \Delta\log X_{i,t} +\epsilon_{i,t}, \qquad
        \Delta\log \tilde Y_{i,t} = \tilde \beta \Delta\log X_{i,t}+\epsilon_{i,t}, \notag
\end{align}
with $\beta = 1-\frac{q}{\alpha q -1}$, $\tilde \beta = 1-\frac{1}{(\sigma-1) (\alpha q -1)}$. Using the relationships $\beta = 1-\frac{q}{\alpha q -1}$
and $\tilde\beta = 1-\frac{1}{(\sigma-1)(\alpha q -1)}$, we invert them to obtain, as in the Remark,
\begin{align}
    q(\alpha)=\frac{1-\beta}{\alpha(1-\beta)-1},\qquad
    \tilde q(\alpha)=\frac{1}{\alpha}\!\left[1+\frac{1}{(\sigma-1)(1-\tilde\beta)}\right], \notag
\end{align}
where $\tilde q(\alpha)$ denotes the function that maps $\alpha$ (given $\tilde\beta$ and $\sigma$) to the implied $q$ under incorrectly measured (non–variety-adjusted) deflators, in contrast to $q(\alpha)$ obtained under variety-adjusted deflators.

Identification is immediate from these mappings: with variety-adjusted deflators, $q$ is point-identified for any $\alpha\in[0,1]$; with deflators that ignore variety, $q$ is point-identified for any $\alpha\in(0,1]$ (given $\sigma$), but \emph{not} at $\alpha=0$, where the mapping is undefined and $\lim_{\alpha\downarrow 0}\tilde q(\alpha)=+\infty$.

Applying Proposition \ref{prop:identification}, the corresponding biases can be computed as follows (the biases arise because the economy has a single true, but unobserved, $\alpha$). For the case of deflators which correctly account for the variety effect, the bias is given by 
\begin{align*}
    \text{Bias}(\alpha;\beta)
= \underbrace{(\beta-1)}_{\text{measured q}} \;-\;\underbrace{\frac{1-\beta}{\alpha(1-\beta)-1}}_{\text{true} \;q}
= -\,\frac{\alpha(1-\beta)^2}{\,\alpha(1-\beta)-1\,}.
\end{align*}
For the case of deflators which ignore the variety effect, the bias is given by 
{\small\begin{align*}
    \text{Bias}(\alpha;\tilde\beta,\sigma)
\; =\;\underbrace{(\tilde\beta-1)}_{\text{measured q}}\;-\;\underbrace{\frac{1}{\alpha}\Bigg[1+\frac{1}{(\sigma-1)\bigl(1-\tilde\beta\bigr)}\Bigg]}_{\text{true} \;q}\; =\; \frac{\,\alpha(\sigma-1)(\tilde\beta-1)^2-(\sigma-1)(\tilde\beta-1)+1\,}
        {\alpha(\sigma-1)(\tilde\beta-1)}. 
\end{align*}}
\end{proof}

\begin{proof}[Proof of Proposition \ref{prop:sign_identification}]
Using Lemma \ref{lem:emp_elasticities}, note that
    \begin{align}
        \frac{\partial_+\log Y}{\partial_+ \log \mathcal I} &= q\frac{\partial_+\log M}{\partial_+ \log \mathcal I} + 1 +  \text{\footnotesize$\frac{1}{\sigma-1}$} \left( 1 -  \frac{\partial_+\log M}{\partial_+ \log \mathcal I}\right)
    \end{align}
    The LHS is the estimated coefficient from our regression. Define the implied $\hat q =\beta-1$. Then, we have 
    $
        \hat q = \beta-1  =\text{\footnotesize$\frac{1}{\sigma-1}$}+ \frac{\partial_+\log M}{\partial_+ \log \mathcal I}(q - q^{CES})
    $.
    If we have a measure of $\frac{1}{\sigma-1}$ we can write 
    $
        \hat q -\frac{1}{\sigma-1} =  \beta-1 - \frac{1}{\sigma-1} = \frac{\partial_+\log M}{\partial_+ \log \mathcal I}(q - q^{CES}).
    $ 
    By Lemma \ref{lem:emp_elasticities}, $k(\underline z, z^\star)<0$ if $z^\star>\underline z$, which implies $\frac{\partial_+\log M}{\partial_+ \log \mathcal I} > 1$, and hence we have that
    $
        \mathrm{sgn}\left( \hat q- \text{\footnotesize$\frac{1}{\sigma-1}$}\right) = \mathrm{sgn}\left( q - \text{\footnotesize$\frac{1}{\sigma-1}$}\right).
    $
\end{proof}

\begin{proof}[Proof of Proposition \ref{prop:HSA}]
    By \cite{matsuyama2023love}, the only demand system in HSA, HIIA or HDIA that does not feature an aggregate externality to TFP and has constant LoV is $CES$. Since $
        \frac{d\ln Y^{CES}}{d \ln M}= \frac{1}{\sigma -1}
        $, we immediately get $g(\bm{y}) = Y^{CES}(\bm{y})M^{\frac{-1}{\sigma-1}}$ has LoV of $0$. By iso-elasticity of $h$, we get $h(M)=M^q$. 
\end{proof}

\versionshort{
\begin{proof}[Proof of Proposition \ref{remark:equiv}]

    Let $\mu^E$ be the baseline probability density function from which productivity is drawn. Let $m^I$ be a density of incumbents (in a dynamic context, $m^I = m_{t-1}$ at time $t$), and $I=\int m^I$ the incumbent mass. We then have the following iterative entry procedure: 
    \begin{enumerate}
        \item At iteration $k\geq 0$, a mass of firms, $E^{(k)}$, enter initially until: 
        \begin{align}\label{seq-entry-cond}
            0 &= \E_{\mu^E}[\max\{\pi( z, m^I, m^{(k-1)}, E^{(k)}),0\}] - f^e \\
            &= \int_{\underline z^{(k)}}^{\infty} \left [ \frac{\mathcal{I}}{\sigma} \frac{z^{\sigma -1 }}{\int z^{\sigma -1 } m^\prime \d z + \int z^{\sigma -1 } m^I \d z} - f^c \right ] \mu^E \d z - f^e \notag
        \end{align} where $m^\prime := E^{(0)} \mu^E$ in the first iteration and $m^\prime := m^{(k-1)} + E^{(k)} \mu^E$ in the $k$th iteration.
        \item Currently active firms check their profitability, and all firms with productivity less than $\underline{z}^{(k)}$ leave the market. The cut-off $\underline{z}^{(k)}$ is determined by:
        \begin{align}
            \pi(\underline{z}^{(k)}, m^I, m^{(k-1)}, E^{(k)}) = \frac{\mathcal{I}}{\sigma} \frac{(\underline{z}^{(k)})^{\sigma -1}}{\int z^{\sigma -1 } m^\prime \d z + \int z^{\sigma -1 } m^I \d z} - f^c = 0 \notag \\
            (\underline{z}^{(k)})^{\sigma -1} = f^c \frac{\sigma}{\mathcal{I}} \left [\int z^{\sigma -1 } m^\prime \d z + \int z^{\sigma -1 } m^I \d z \right] \notag
        \end{align}
        \item After exit, the density function corresponding to the firms left in the market is relabeled as the new incumbent density function and is given by:
        \begin{align}
            m^I \leftarrow m^I + m^\prime(z)\mathbb{I}_{\{z\geq \underline z^{(k)}\}} = m^I + m^{(k)} \notag
        \end{align} where $m^I$ is not affected by the truncation since scope for entry implies that the threshold is weakly lower than the one previously incurred by the incumbents.
    \end{enumerate} 

    Steps 1 to 3 are iterated until convergence. To prove equivalence, one follows three steps. First, show that the cutoff is the same as under the orgininal entry protocol for all iterations, $k$. Using this fact derive, second, the partial entry masses $E^{(k)}$ explicitly. Finally, show convergence of total entry in the alternative to entry in the original protocol: $\sum E^{(k)} = E$.
\end{proof}
}

\versionlong{
\begin{proof}[Proof of Proposition \ref{remark:equiv}]

    Let $\mu^E$ be the baseline probability density function from which productivity is drawn. Let $m^I$ be a density of incumbents (in a dynamic context, $m^I = m_{t-1}$ at time $t$), and $I=\int m^I$ the incumbent mass. We then have the following iterative entry procedure: 
    \begin{enumerate}
        \item At iteration $k\geq 0$, a mass of firms, $M^{(k)}$, enter initially until: 
        \begin{align}\label{seq-entry-cond}
            0 &= \E_{\mu^E}[\max\{\pi( z, m^I, m^{(k-1)}, M^{(k)}),0\}] - f^e \\
            &= \int_{\underline z^{(k)}}^{\infty} \left [ \frac{\mathcal{I}}{\sigma} \frac{z^{\sigma -1 }}{\int z^{\sigma -1 } m^\prime \d z + \int z^{\sigma -1 } m^I \d z} - f^c \right ] \mu^E \d z - f^e \notag
        \end{align} where $m^\prime := M^{(0)} \mu^E$ in the first iteration and $m^\prime := m^{(k-1)} + M^{(k)} \mu^E$ in the $k$th iteration.
        \item Currently active firms check their profitability, and all firms with productivity less than $\underline{z}^{(k)}$ leave the market. The cut-off $\underline{z}^{(k)}$ is determined by:
        \begin{align}
            \pi(\underline{z}^{(k)}, m^I, m^{(k-1)}, M^{(k)}) = \frac{\mathcal{I}}{\sigma} \frac{(\underline{z}^{(k)})^{\sigma -1}}{\int z^{\sigma -1 } m^\prime \d z + \int z^{\sigma -1 } m^I \d z} - f^c = 0 \notag \\
            (\underline{z}^{(k)})^{\sigma -1} = f^c \frac{\sigma}{\mathcal{I}} \left [\int z^{\sigma -1 } m^\prime \d z + \int z^{\sigma -1 } m^I \d z \right] \notag
        \end{align}
        \item After exit, the density function corresponding to the firms left in the market is relabeled as the new incumbent density function and is given by:
        \begin{align}
            m^I \leftarrow m^I + m^\prime(z)\mathbb{I}_{\{z\geq \underline z^{(k)}\}} = m^I + m^{(k)} \notag
        \end{align} where $m^I$ is not affected by the truncation since scope for entry implies that the threshold is weakly lower than the one previously incurred by the incumbents.
    \end{enumerate} 

    Steps 1 to 3 are iterated until convergence. \medskip \\

    \textit{Equivalence of threshold $\underline z$} \medskip
    
    By step 1, we have that at iteration $k$ the following holds:
    \begin{align}
        f^e &= \int_{\underline z^{(k)}}^{\infty} \left [ \frac{\mathcal{I}}{\sigma} \frac{z^{\sigma -1 }}{\int z^{\sigma -1 } m^{(k-1)} \d z + M^{(k)} \int z^{\sigma -1 } \mu^E \d z + \int z^{\sigma -1 } m^I \d z} - f_c \right] \mu^E \d z \\
        &= \frac{\mathcal{I}}{\sigma} \frac{p_E(\underline z^{(k)}) \; \E_{\mu^E}[ z^{\sigma -1 } \g  z \geq \underline z^{(k)}]}{\int z^{\sigma -1 } m^{(k-1)} \d z + M^{(k)} \E_{\mu^E}[z^{\sigma -1 }] + \int z^{\sigma -1 } m^I \d z} - p_E(\underline z^{(k)}) f_c
    \end{align}
    where $p_E(\underline z^{(k)}) = \int_{\underline z^{(k)}}^\infty \mu^E(z) \d z$. \smallskip \\
    Solving for $M^{(k)}$, we obtain:
    \begin{align} \label{mass_iter_k}
        M^{(k)} = \frac{1}{\E_{\mu^E}[z^{\sigma -1 }]} \; \left [ \frac{\mathcal{I}}{\sigma} \frac{p_E(\underline z^{(k)}) \; \E_{\mu^E}[z^{\sigma -1 } \g  z \geq \underline z^{(k)}]}{f^e + f^c p_E(\underline z^{(k)})} - \int z^{\sigma -1 } m^{(k-1)}\d z - \int z^{\sigma -1 } m^I \d z \right]
    \end{align}

    Substituting the derived $M^{(k)}$ in the cut-off equation of step 2,  we obtain:
    \begin{align}
        (\underline{z}^{(k)})^{\sigma -1} = f^c \frac{p_E(\underline z^{(k)}) \; \E_{\mu^E}[z^{\sigma -1 } \g  z \geq \underline z^{(k)}]}{f^e + f^c \; p_E(\underline z^{(k)})} \label{cutoff_to_subst} \\
        \frac{f^e}{f^c} = p_E(\underline z^{(k)}) \left[ \frac{\E_{\mu^E}[z^{\sigma -1 } \g  z \geq \underline z^{(k)}]}{(\underline{z}^{(k)})^{\sigma -1}} -1 \right] \label{cutoff_iter}
    \end{align}

    First, note that clearly none of this is dependent on having started at iteration $k$. We would have obtained eq. (\ref{cutoff_iter}) also by having derived $M^{(0)}$ in eq. (\ref{mass_iter_k}). Second, note that $\underline z^{(k)}$ is a constant sequence as it depends only on $f^c$ and $f^e$, which do not change with $k$. Hence, we can write $\underline{z}_k = \underline{z} \; \forall k\in \mathbb{N}$. Third, note that the equation pinning down the threshold in the iterative entry procedure is the same as in the simultaneous entry game eq. (\ref{eq:cutoff}). \medskip \\

    \textit{Equivalence of entry mass $E$} \medskip 
    
    With $\underline{z}^{(k)} = \underline{z} \; \forall k\in \mathbb{N}$, it is clear that all we are doing is stacking scaled versions of the baseline density truncated at $\underline{z}$. Therefore, eq. (\ref{mass_iter_k}) can be written as: 
    \begin{align}
        M^{(k)} = \frac{1}{\E_{\mu^E}[z^{\sigma -1 }]} \; \left [ \frac{\mathcal{I}}{\sigma} \frac{\underline{z}^{\sigma -1}}{f^c} - (M^{(0)} + ... + M^{(k-1)}) \int_{\underline z}^{\infty} z^{\sigma -1 } \mu^E \d z - \int z^{\sigma -1 } m^I \d z \right]
    \end{align} where we used eq. (\ref{cutoff_to_subst}) to substitute in $\frac{\underline{z}^{\sigma -1}}{f^c}$.\\
    This shows that $\{M^{(k)}\}$ is a decreasing sequence. Since $M^{(k)} \geq 0$ holds $\forall k$, $\{M^{(k)}\}$ has a real limit point, and this limit point must be $0$. Hence, taking the limit as $k \rightarrow \infty$, we obtain:
    \begin{align}
        0 =  \frac{\mathcal{I}}{\sigma} \frac{\underline{z}^{\sigma -1}}{f^c} - \sum_{i = 0}^\infty M^{(i)} \int_{\underline z}^{\infty} z^{\sigma -1 } \mu^E \d z - \int z^{\sigma -1 } m^I \d z \\
        \sum_{i = 0}^\infty M^{(i)} = E = \frac{\mathcal{I}}{\sigma f^c} \frac{\underline{z}^{\sigma -1}}{\int_{\underline z}^{\infty} z^{\sigma -1 } \mu^E \d z} - \frac{\int z^{\sigma -1 } m^I \d z}{\int_{\underline z}^{\infty} z^{\sigma -1 } \mu^E \d z} 
    \end{align}
    which is clearly the same entry mass as in the simultaneous entry game (see eq. \ref{eq:entrymass-ge} with $\mathcal{I}$ instead of $R_t$). \medskip \\
    
    \textit{Equivalence of measure $\mu$} \medskip \\
    Assuming that $\mu^E$ is a bounded density function, it is clear that $m^{(k)}:= \sum_{i = 0}^{k} M^{(i)} \mu^E(z)\mathbb{I}_{\{z\geq \underline z\}}$ converges to $m^E:= E \mu^E(z)\mathbb{I}_{\{z\geq \underline z\}}$ as $k \rightarrow \infty$, given that the partial sums of the entry mass series form a Cauchy sequence. The limiting density function $m = m^I +  E \mu^E(z)\mathbb{I}_{\{z\geq \underline z\}}$ is the same density function as in the simultaneous entry game equilibrium.
\end{proof}
}

\begin{proof}[Proof of Proposition \ref{prop:elast_labor}]
    We prove the first statement by proving that factor-saving effects are dampened by the endogenous factor supply response. We prove it for a single primary factor with the understanding that it readily generalizes. 
    
Suppose the household has preferences
$    \mathcal U = \log C-\nu(L)
$
and BC
$    wL +\Pi = PC.
$ Setting labor as the numeraire, we have that the labor supply curve is implicitly defined by $
    \nu^\prime(L) = (PC)^{-1}
$. Assuming that $\nu(\cdot)$ is twice continuously differentiable, we can write the labor market clearing as 
\begin{align}
    L^p +Mf_c+Ef_e = L = {\nu^\prime}^{-1}((PC)^{-1})
\end{align}
From the price index and consumption aggregators, we know that 
   $ PC = L^p\mu$, 
which implies
\begin{align}
    Mf_c+Ef_e = {\nu^\prime}^{-1}((L^p\mu)^{-1})-L^p
\end{align}
Note that RHS is decreasing in $L^p$ for all $\nu^\prime>0$ and $\mu>1$. So, more firms ($M$) or more entrants ($E$) reduce labor used in production, even when supply is elastic. To complete the proof, define appropriately a function $g$, with $g^\prime<0$ such that
\begin{align}
    Mf_c+Ef_e = g(L^p)-L^p
\end{align}
Comparing it with the inelastic case
\begin{align}
    Mf_c+Ef_e = \bar L-L^p
\end{align}
we note that for any given change of the LHS $\Delta LHS$ under inelastic supply, $\Delta L^p=-\Delta LHS$. Under elastic supply, if $\Delta L^p=-\Delta LHS$ then $|\Delta RHS|>|\Delta LHS|$. We conclude that under elastic labor supply, for any given $\Delta LHS$ we need $|\Delta L^p|<|\Delta LHS|$. Hence, elastic labor dampens movements in production labor. Since the elastic factor supply dampens movements in factors used in production, it reduces the factor-saving effects. As a consequence, all else equal, the output change is strictly smaller under elastic factor supply.


\end{proof}

\begin{proof}[Proof of Proposition \ref{prop:compreces} and Figure \ref{fig:decomp}.]
    We prove the characteristics of Figure \ref{fig:decomp}, panels A and B, first. To this end, note that $\underline z_2^{\sigma-1}$ is a strictly increasing, monotonic function of $f_h^c$ defined by the zero profit condition (\ref{eq:zcp-pe}), with $\underline z_2^{\sigma-1} \rightarrow \infty$ as $f_h^c\rightarrow \infty$, and any analysis in terms of $\underline z_2^{\sigma-1}$ carries over in terms of $f_h^c$. Now note that \begin{align}
        &M_3 = E_1 p_E(\underline z_2) + E_3 p_E(\underline z_1) \quad\text{and}\quad
        E_3 = \left( 1 + \frac{\underline z_1^{\sigma-1}}{\E_{\mu^E}[z^{\sigma -1 } \g  z \geq \underline z_1] (\sigma - 1)} \right)^{-1} \notag \\ 
        & \cdot E_1 \left[ 1 - \frac{\int_{\{ z \geq \underline{z}_2 \}}z^{\sigma -1 } \mu^E(z) \d z}{\int_{\{ z \geq \underline{z}_1 \}}z^{\sigma -1 } \mu^E(z) \d z} + (p_E(\underline z_1) - p_E(\underline z_2))\frac{\underline z_1^{\sigma-1}}{(\int_{\{ z \geq \underline{z}_1 \}}z^{\sigma -1 } \mu^E(z) \d z) \; (\sigma-1)}\right], \label{eq:decomp_E3}
    \end{align} which follows from substituting eq. (\ref{eq:proof_prop2_diffR}) into eq. (\ref{eq:proof_prop2_E3}) and rearranging for $E_3$. Using eq. (\ref{eq:decomp_E3}), we calculate \begin{align}
        \label{decomp_derivM3} \frac{\partial}{\partial\underline z_2^{\sigma-1}} M_3 &= E_1 \mu^E(\underline z_2)\left(\frac{\frac{\underline z_2^{\sigma-1} +\frac{1}{\sigma-1}\underline z_1^{\sigma-1}}{\E_{\mu^E}[z^{\sigma -1 } \g  z \geq \underline z_1]}}{1 + \frac{1}{\sigma -1}\frac{\underline z_1^{\sigma-1}}{\E_{\mu^E}[z^{\sigma -1 } \g  z \geq \underline z_1]}} -1\right).
    \end{align} If $\underline z_2 = \underline z_1$, one verifies that the term in parentheses is negative, while for large enough $\underline z_2$, it is positive. By continuity, there exists a $\underline z_2^*$, such that \begin{align} \label{eq:decomp_derivM3}
        \frac{\partial M_3}{\partial\underline z_2^{\sigma-1}} \lesseqgtr 0 \iff \underline z_2 \lesseqgtr \underline z_2^* 
    \end{align}%
    Now, partial derivatives are \begin{align}
        \frac{\partial [M_3/M_1]^{q - q^{CES}}}{\partial\underline z_2^{\sigma-1}} &= (q - q^{CES})\left[\frac{M_3}{M_1}\right]^{q - q^{CES}-1} \frac{1}{M_1} \frac{\partial M_3}{\partial\underline z_2^{\sigma-1}}, \label{eq:decomp_derivFac1} \\
        \frac{\partial (L_3^p/L_1^p)}{\partial\underline z_2^{\sigma-1}} &= \frac{-f_l^c}{L_1^p} \frac{\partial M_3}{\partial\underline z_2^{\sigma-1}}, \label{eq:decomp_derivFac2}\\
        \frac{\partial}{\partial\underline z_2^{\sigma-1}} \left[\frac{\int z^{\sigma -1 } m_3 (z)\d z}{\int z^{\sigma -1 } m_1 (z)\d z}\right]^{1/(\sigma-1 )} &= \left[\frac{\int z^{\sigma -1 } m_3 (z)\d z}{\int z^{\sigma -1 } m_1 (z)\d z}\right]^{1/(\sigma-1 )-1}\notag \\
        &\; \times \frac{1}{\sigma-1}\frac{\underline z_1^{\sigma-1}}{\int_{\{ z \geq \underline{z}_1 \}}z^{\sigma -1 } \mu^E(z) \d z}\frac{-1}{E_1(\sigma-1)}  \frac{\partial M_3}{\partial\underline z_2^{\sigma-1}} \label{eq:decomp_derivFac3},
    \end{align} from which we see that eq. (\ref{eq:decomp_derivFac2}, \ref{eq:decomp_derivFac3}) have the reverse sign order of (\ref{eq:decomp_derivM3}), while eq. (\ref{eq:decomp_derivFac1})  has the same sign order, iff $q > q^{CES}$. Combined with the facts that each factor is 1 at $f_l^c$ (no crisis in phase 2) and at $f_h^c \rightarrow \infty$, the graphs in panels A and B of Figure \ref{fig:decomp} follow. Furthermore, if $q\leq q^{CES}$, then the graph for $q^\prime$ in panel C follows, and $(Y_3/Y_1)(q, f_h^c) \geq 1$ for all $f_h^c \geq f_l^c$.

    Assume now that $q > q^{CES}$. For reference, note that \begin{align} \label{eq:decomp_prodintratio}
        \frac{\int z^{\sigma -1 } m_1 (z)\d z}{\int z^{\sigma -1 } m_3 (z)\d z} &= 1 + \frac{M_1 - M_3}{E_1(\sigma-1)} \frac{\underline z_1^{\sigma-1}}{\int_{\{ z \geq \underline{z}_1 \}}z^{\sigma -1 } \mu^E(z) \d z}
    \end{align} is the ratio of productivity integrals.
    The output ratio only depends on $f_h^c$ through $\underline z_2$, and it depends on $\underline z_2$ only through $M_3$. Consider the output ratio now as a function of $M_3$ given $q$, i.e. $\frac{Y_3}{Y_1}(M_3; q)$, and note that the domain of $M_3(\underline z_2)$ is $[M_{min}, M_1]$ for some $M_{min}$, and that $M_3$ runs from $M_1$ monotonically down to $M_{min} := M_3(\underline z_2^*)$ and back up to $M_1$ in the limit. We write $$\frac{Y_3}{Y_1}(M_3; q) = F_1(M_3, q)F_2(M_3)F_3(M_3),$$ corresponding to its three contributing factors, and take the derivative \begin{align}
        \frac{\partial }{\partial M_3}\frac{Y_3}{Y_1} = \bigg(\frac{F_1^\prime(M_3, q)}{F_1(M_3, q)} + \frac{F_2^\prime(M_3)}{F_2(M_3)} + \frac{F_3^\prime(M_3)}{F_3(M_3)}\bigg)\frac{Y_3}{Y_1}(M_3;q)
        =:g(M_3, q)\frac{Y_3}{Y_1}(M_3;q).
    \end{align} The derivative vanishes iff the first factor vanishes; substituting in:%
    {\small\begin{align}\label{eq:decomp_foc_outratio}
        0 = \frac{(q-q^{CES})}{M_3} \underbrace{- \frac{f_l^c}{L_3^p(M_3)} - \frac{\int z^{\sigma -1 } m_1 (z)\d z}{\int z^{\sigma -1 } m_3 (z)\d z} \frac{1}{\sigma-1} \frac{\underline z_1^{\sigma-1}}{\int_{\{ z \geq \underline{z}_1 \}}z^{\sigma -1 } \mu^E(z) \d z}\frac{1}{E_1(\sigma-1)}}_{:=-s(M_3)/M_3} =: v(M_3, q),
    \end{align}}%
    which is strictly decreasing in $M_3$ (to see that the third summand is decreasing cf. eq. \ref{eq:decomp_prodintratio}). Therefore, if $M_3^*$ satisfying eq. (\ref{eq:decomp_foc_outratio}) exists, it is unique. 
    Since $M_{min} \leq M_3 \leq M_1$, it can only exist if $v(M_{min},q) \geq 0 \geq v(M_1,q)$. Clearly, if $q$ is large enough, then $v$ must be globally positive and if $q$ gets close to $q^{CES}$, then $v$ is negative for some, and eventually all $M_3\in [M_{min}, M_1]$. One checks that an interior critical point exists iff $q\in (\underline q, \overline q)$, where $\underline q = q^{CES} + s(M_1)$ and $\overline q = q^{CES} + s(M_{min})$. On the one hand, if $q\leq \underline q$, even though there is LOV relative to CES, the output ratio is decreasing in the number of post-crisis varieties, and any crisis is cleansing. On the other hand, if $q\geq \overline{q}$, then the output ratio is globally increasing in the number of varieties, and all crises come at an output loss. 
    For the region in between, $q\in (\underline q , \overline{q})$, there is an extremum at some $M_3^{*}$. This is a maximum because $Y_3/Y_1$ satisfies second-order conditions at $M_3^{*}$ for such $q$. See this as follows. As argued, $g(M_3;q)$ is decreasing, so $g^\prime<0$. The second derivative of $Y_1/Y_3$ evaluated at $M_3^{*}$ is \begin{align}
        g^\prime(M_3^*;q) \frac{Y_3}{Y_1}(M_3^*;q) + g(M_3^*;q) \frac{Y_3}{Y_1}^\prime(M_3^*;q) =  g^\prime(M_3^*;q) \frac{Y_3}{Y_1}(M_3^*;q) +  g(M_3^*;q) \cdot 0 <0.
    \end{align} %
    Consequently, there is an optimal number of firms in phase 3, $M_3^*\in (M_{min}, M_1)$, created by an ``optimally cleansing business cycle''. Because the mapping between $f_h^c$ and $M_3$ is parabola-shaped, its inverse image always has two elements, and any optimal post-crisis number of firms can be created by two different values of $f_h^c$. This implies the graph displayed in panel C for $q^{\prime\prime}$. It is easy to check that the output ratio uniformly decreases in $q$, which explains why the graph for $q^{\prime\prime\prime}$ lies below. Therefore, as the curve of $Y_3/Y_1$ shifts down, there must be an interval $(q_\circ, q^\circ)\subseteq (\underline q, \overline{q})$ where the output ratio crosses 1 multiple times, creating two nonempty, disjunct intervals in which crises are cleansing. \\
    
    The partial equilibrium result follows immediately since $F_2$ and $F_3$ are independent of $f_h^c$ and of $q$, by Proposition \ref{proposition:PE-cycles}.
\end{proof}

\begin{proof}[Proof of Remark \ref{rem:indep} and Proposition \ref{prop:het_ext}]
    Since firms are selected on $z$, the average externality contributions in the economy are given by \begin{align}
    \E_{m_1}[\xi] &= \E_{\mu^{E}}[\xi \g z \geq \ulz_1], \\
    \E_{m_2}[\xi] &= \E_{\mu^{E}}[\xi \g z \geq \ulz_2], \\
    \E_{m_3}[\xi] &= \E_{\mu^{E}}[\xi \g z \geq \ulz_2]\times \frac{M_2}{M_3} + \E_{\mu^{E}}[\xi \g z \geq \ulz_1] \times \frac{E \mu_z(z\geq \ulz_1)}{M_3}.
\end{align} where the first equation obtains from initial entry, the second from keeping only surviving firms, and the third from averaging the surviving firms of the cycle with new entrants. As usual, $E>0$ is the post-cycle entry mass of firms. Let $\Delta$ be the difference operator differencing cycle phase 3 and 1. We immediately see that if externality and productivity draws are mean-independent, then $\E_{m_t}[\xi] = \E_{\mu^E}[\xi] \; \forall t = 1, 2, 3$. Thus, $\mathcal{M}_t = M_t$, and the model is identical to our baseline model in aggregates. This proves Remark \ref{rem:indep}.

Substituting into the aggregate externality $\mathcal{M}$, we obtain \eqref{eq:extension_hetvar_deltacurlyM}. This tells us that the welfare effect running through $\mathcal{M}$ is positive if and only if the selection in $\xi$ given by $\Delta \log \E_{m}[\xi]$ is at least some positive number $- \Delta \log M$. This is a result analogous to Proposition \ref{proposition:PE-cycles}. To relate the cleansing effect in this extension to the baseline model, and thus to complete the proof of Proposition \ref{prop:het_ext}, compute $\Delta\log Y$ under the alternative aggregator. This yields \begin{align}
    \Delta \log Y = q\ \Delta \log \E_{m}[\xi] + \underbrace{\bigg(q - \frac{1}{\sigma-1}\bigg)\Delta \log M}_{\Delta Y^{hom}}.
\end{align} One then notices that $\Delta \log \E_{m}[\xi] > 0$ if there is positive positive correlation between $\xi$ and $z$ in the baseline distribution $\mu^E$. The stated cases follow.
\end{proof}

\versionshort{
\begin{proof}[Proof of Proposition \ref{prop:smooth_reversion_pareto_sequential_fc} and Proposition \ref{prop:smooth_reversion_pareto_smooth_fc}]

    Consider a Pareto density $\mu^E(z)$ with minimum $x$ and shape $\beta$. For brevity, we write $x$ instead of $z_{min}$ in this proof. Note that $\bm{\mathrm{I}}(z) := \int_{\underline{z}}^\infty z^{\sigma-1} \mu^E(z) dz = x^\beta \frac{\beta}{\beta - (\sigma-1)} \underline{z}^{\sigma-1-\beta}$ with $\beta - (\sigma-1) > 0$. \\

    \noindent\emph{Step 1, calculate $E_{-1}, \underline z_{-1}$:} The cutoff $\underline{z}_{-1}$ satisfies $\frac{f^e}{1} = \frac{\bm{\mathrm{I}}(\underline{z}_{-1})}{\underline{z}_{-1}^{\sigma -1}} - \int_{\underline z_{-1}}^\infty \mu^E(z) dz$. From this follows
    \[
        \underline z_{-1}
        = x\left\{\frac{1}{f^e}\frac{\sigma-1}{\beta-(\sigma-1)}\right\}^{\!1/\beta},\quad E_{-1}
    =\frac{\mathcal I}{\sigma}\,\frac{(\sigma-1)}{\beta}\cdot\frac{1}{f^e}.
    \]

    \noindent\emph{Step 2, truncate the firm measure to find the crisis measure (on impact):} Immediately, $E_0 = 0$. From $\pi(\underline{z}_0, m_0)=0 \Leftrightarrow \phi = \frac{\mathcal{I}}{\sigma} \frac{\underline{z}_{0}^{\sigma-1}}{E_{-1} \bm{\mathrm{I}}(\underline{z}_{0})}$ we solve for $\underline{z}_{0}$: 
    \[
    \underline z_0
    = x\left\{\frac{\sigma-1}{\beta-(\sigma-1)}\cdot \frac{1}{f^e}\cdot \phi\right\}^{\!1/\beta}.
    \]


    \noindent \emph{Step 3, calculate reentry:} Now, in $t = 1, 2, ..., T^*$, fixed cost go down relative to the previous period. From the cutoff equation for entry, obtain directly 

    $$
        \underline{z}_t = x \bigg\{ \frac{f_{t}}{f^e} \frac{\sigma -1}{\beta - (\sigma -1)} \bigg\}^{1/\beta}
    $$
    where $f_{t} = \phi^{1-t/T^*}$. Hence, the entry is given by 
    $$
        E_t = \frac{\mathcal{I}}{\sigma} \frac{1}{f_t^c} \frac{\underline z_t^{\sigma -1}}{\bm{\mathrm{I}}(\underline z_t)} - \frac{
            E_{-1} \bm{\mathrm{I}}(\underline z_0) + \sum_{\tau = 1}^{t-1} E_\tau \bm{\mathrm{I}}(\underline{z}_\tau)
            }{
                \bm{\mathrm{I}}(\underline z_t)
            }.
    $$
    Note that we can simplify the first term using $\underline z_t^\beta
    =x^\beta\Big\{\frac{f_t^c}{f^e}\frac{\sigma-1}{\beta-(\sigma-1)}\Big\}$ to $
    \frac{\mathcal{I}}{\sigma}\frac{1}{f_t^c}\frac{\underline z_t^{\sigma-1}}{\bm{\mathrm{I}}(\underline z_t)}
    =\frac{\mathcal{I}}{\beta}\frac{\sigma-1}{\sigma}\frac{1}{f^e}.$
    Furthermore, the effect of pre-cycle incumbents is %
    \[
    E_{-1} \bm{\mathrm{I}}(\underline z_0)
    =\frac{\mathcal{I}}{\sigma}\,\frac{\underline z_0^{\sigma-1}}{\phi}
    =\frac{\mathcal{I}}{\sigma}\,\frac{1}{\phi}\left[x^\beta\Big\{\frac{\phi}{f^e}\frac{\sigma-1}{\beta-(\sigma-1)}\Big\}\right]^{\frac{\sigma-1}{\beta}}
    =\frac{\mathcal{I}}{\sigma}\bigg\{x^\beta\frac{1}{f^e}\frac{\sigma-1}{\beta-(\sigma-1)}\bigg\}^{\frac{\sigma-1}{\beta}}\phi^{\frac{\sigma-1-\beta}{\beta}}.
    \]
    In the summation above, we have $\bm{\mathrm{I}}(\underline{z}_\tau)$. We solve for this object explicitly %
    \[
    \begin{aligned}
    \bm{\mathrm{I}}(\underline z_t) = \phi^{\frac{\sigma-1-\beta}{\beta}\left(1-\frac{t}{T^*}\right)}
    \frac{\beta}{\sigma-1}\,f^e\,
    \Big\{x^\beta\frac{1}{f^e}\frac{\sigma-1}{\beta-(\sigma-1)}\Big\}^{\frac{\sigma-1}{\beta}}.
    \end{aligned}
    \]
    Therefore, 
    $$
    \begin{aligned}
    \frac{E_{-1} \bm{\mathrm{I}}(\underline z_0)}{\bm{\mathrm{I}}(\underline z_t)}
    = \frac{\mathcal I}{\sigma}\,\frac{\sigma-1}{\beta}\,\frac{1}{f^e}\,
    \phi^{\frac{\sigma-1-\beta}{\beta}\frac{t}{T^*}}\quad \text{and}\quad E_\tau \frac{\bm{\mathrm{I}}(\underline{z}_\tau)}{\bm{\mathrm{I}}(\underline{z}_t)} = E_\tau \phi^{\frac{\sigma -1-\beta}{\beta} \frac{t-\tau}{T^*}} \quad \text{for}\quad t = 1, \dots, T^*.
    \end{aligned}
    $$
    Then, denoting $\pi :=  \phi^{\frac{\sigma -1-\beta}{\beta} \frac{1}{T^*}} $, we can write %
    $
        E_t = E_{-1} - E_{-1} \pi^t - \sum_{\tau=1}^{t-1}  E_\tau \pi^{t-\tau}.
    $
    \noindent \emph{Step 4, solve the recursion and find $m_\infty$:} We can verify the guess %
    $
        E_t = E_{-1}(1-\pi).
    $
    From this, we can put together the measure $m_{\tau}$: 
    \begin{align*}
        m_\tau(z)  %
                    &= E_{-1}\left( \indic{\ulz_0} + \Big(1-\phi^{\frac{\sigma -1-\beta}{\beta} \frac{1}{T^*}} \Big) \sum_{l=1}^{\tau}  \indic{\ulz_\tau}\right)\mu^E(z)
    \end{align*}
    Substituting out the cutoffs $\ulz_0, \ulz_\tau$ and $E_{-1}$ yields the expression for $m^\tau$ in the proposition.\\

    \noindent \emph{Step 5, continuous-time convergence:}
    Next, we consider what happens as the crisis approaches its continuous time limit. First we look at the tail of $m_\infty$, i.e. at a value $z\geq \max \ulz_\tau = \ulz_0$ (recall that the cutoff monotone decreases after $t=0$). Note the following limit: $N(1-u^{1/N}) \rightarrow -\log u$ for any $u > 0$ as $N\rightarrow \infty$. Applying this to our problem, we have
    \begin{align*}
        m_T^*(z)&= E_{-1}\left( \indic{\ulz_0} + \Big(1-\phi^{\frac{\sigma -1-\beta}{\beta} \frac{1}{T^*}} \Big) \sum_{l=1}^{T^*}  \indic{\ulz_\tau}\right)\mu^E(z)\\
                &\longrightarrow E_{-1}\cdot \mu^E(z)\cdot \left[ 1 - \frac{\sigma -1-\beta}{\beta} \log \phi \right] \quad \text{as} \quad T^*\rightarrow\infty
    \end{align*}%
    as claimed. To analyze the left tail, let $\alpha \in [0, 1]$ and define 
    $
        z_\alpha = x\left[ \frac{\phi^{1-\alpha}}{f^e} \frac{\sigma -1}{\beta - (\sigma -1)} \right]^{1/\beta}.
    $
    Note that we can reach every $z \in [\ulz_{-1}, \ulz_0]$ in the left tail by varying $\alpha$. We investigate the convergence of $m_{T^*}(z_\alpha)$. We have \begin{align*}
        m_T^*(z_\alpha)&= E_{-1}\left( \indictwo{\ulz_0}{z_\alpha} + \Big(1-\phi^{\frac{\sigma -1-\beta}{\beta} \frac{1}{T^*}} \Big) \sum_{l=1}^{T^*}  \indictwo{\ulz_\tau}{z_\alpha}\right)\mu^E(z_\alpha)\\
                &\longrightarrow (1-\alpha)\cdot E_{-1}\cdot \mu^E(z_\alpha)\cdot \left[ \frac{\beta - (\sigma -1)}{\beta} \log \phi \right] \quad \text{as} \quad T^*\rightarrow\infty.
    \end{align*} %
    Finally get rid of $\alpha$ in the limit. Inverting $z_\alpha$ for $\alpha$ yields \begin{align*}
        \alpha = \left[ 1-\frac{\log\left( \big(\frac{z_\alpha}{x}\big)^\beta f^e \frac{\beta - (\sigma-1)}{\sigma -1} \right)}{\log \phi} \right] =  \left[ 1-\beta\frac{ \log\left( \frac{z_\alpha}{x}\right) - \log \left( \frac{\ulz_{-1}}{x} \right)}{\log \phi} \right] 
    \end{align*}%
    Where we have used the formula for $\ulz_{-1}$. Substitute this expression back into the limit and the result follows.\\

    \noindent \emph{Step 6, firm mass after the cycle $M_{T^*}$:} Recall that the anti-cumulative of $\mu^E$ is denoted by $p_E$, and $p_E(a) = \left(\frac{x}{a}\right)^\beta$. Hence, we are looking for \begin{align*}
        M_{T^*} &= E_{-1} p_E(\ulz_0) + \sum_{\tau = 1}^{T^*} E_\tau p_E(\ulz_\tau) 
                &= \frac{\mathcal{I}}{\sigma} \frac{\beta - (\sigma -1)}{\beta} \cdot \phi^{-1} \left\{1 + 
                 \left[1-\phi^{\frac{\sigma -1 -\beta}{\beta} \frac{1}{T^*}}\right]\cdot \phi^{\frac{1}{T^*}}\frac{1-\phi^{\frac{T^*}{T^*}}}{1-\phi^{\frac{1}{T^*}}}\right\}
    \end{align*}
    Noting that 
    $
        \frac{\left[1-\phi^{\frac{\sigma -1 -\beta}{\beta} \frac{1}{T^*}}\right]}{1-\phi^{\frac{1}{T^*}}} \rightarrow \frac{\sigma -1 -\beta}{\beta}
    $
    we get, as required: \begin{align*}
        M_{T^*} &= \frac{\mathcal{I}}{\sigma} \frac{\beta - (\sigma -1)}{\beta} \cdot \phi^{-1} \left\{1 + 
                 (\phi - 1) \frac{ \beta - (\sigma -1)}{\beta} \right\}.
    \end{align*}

\end{proof}
}

\versionlong{
\begin{proof}[Proof of Proposition \ref{prop:smooth_reversion_pareto_sequential_fc} and Proposition \ref{prop:smooth_reversion_pareto_smooth_fc}]

    Consider a Pareto density $\mu^E(z)$ with minimum $x$ and shape $\beta$. For brevity, we write $x$ instead of $z_{min}$ in this proof. Note that $\bm{\mathrm{I}}(z) := \int_{\underline{z}}^\infty z^{\sigma-1} \mu^E(z) dz = x^\beta \frac{\beta}{\beta - (\sigma-1)} \underline{z}^{\sigma-1-\beta}$ with $\beta - (\sigma-1) > 0$. \\

    \noindent\emph{Step 1, calculate $E_{-1}, \underline z_{-1}$:} The cutoff $\underline{z}_{-1}$ satisfies $\frac{f^e}{1} = \frac{\bm{\mathrm{I}}(\underline{z}_{-1})}{\underline{z}_{-1}^{\sigma -1}} - \int_{\underline z_{-1}}^\infty \mu^E(z) dz$. From this follows

    \[
        \frac{f^e}{1}
        =\frac{\bm{\mathrm{I}}(\underline z_{-1})}{\underline z_{-1}^{\sigma-1}}
        -\int_{\underline z_{-1}}^\infty \mu^E(z)\,dz
        =\frac{\beta x^\beta}{\beta-(\sigma-1)}\,\underline z_{-1}^{-\beta}
        - x^\beta \underline z_{-1}^{-\beta}
        = x^\beta \underline z_{-1}^{-\beta}\,\frac{\sigma-1}{\beta-(\sigma-1)}.
    \]
    \[
        \Rightarrow\quad 
        \underline z_{-1}^{\beta}
        = x^{\beta}\,\frac{\sigma-1}{\beta-(\sigma-1)}\cdot \frac{1}{f^e}
        \quad\Rightarrow\quad
        \underline z_{-1}
        = x\left\{\frac{1}{f^e}\frac{\sigma-1}{\beta-(\sigma-1)}\right\}^{\!1/\beta}.
    \]

    \[
    E_{-1}
    =\frac{\mathcal I}{\sigma}\,\frac{\underline z_{-1}^{\sigma-1}}{\bm{\mathrm{I}}(\underline z_{-1})}
    =\frac{\mathcal I}{\sigma}\left(\frac{\bm{\mathrm{I}}(\underline z_{-1})}{\underline z_{-1}^{\sigma-1}}\right)^{-1}
    =\frac{\mathcal I}{\sigma}\left[\frac{\beta x^\beta}{\beta-(\sigma-1)}\,\underline z_{-1}^{-\beta}\right]^{-1}
    =\frac{\mathcal I}{\sigma}\,\frac{\beta-(\sigma-1)}{\beta}\,x^{-\beta}\underline z_{-1}^{\beta}.
    \]
    \[
    \text{Using }x^{-\beta}\underline z_{-1}^{\beta}
    =\frac{\sigma-1}{\beta-(\sigma-1)}\cdot\frac{1}{f^e}
    \quad\Rightarrow\quad
    E_{-1}
    =\frac{\mathcal I}{\sigma}\,\frac{(\sigma-1)}{\beta}\cdot\frac{1}{f^e}.
    \]

    \noindent\emph{Step 2, truncate the firm measure to find the crisis measure (on impact):} Immediately, $E_0 = 0$. From $\pi(\underline{z}_0, m_0)=0 \Leftrightarrow \phi = \frac{\mathcal{I}}{\sigma} \frac{\underline{z}_{0}^{\sigma-1}}{E_{-1} \bm{\mathrm{I}}(\underline{z}_{0})}$ we solve for $\underline{z}_{0}$:

    \[
    \phi=\frac{\mathcal I}{\sigma}\,\frac{\underline z_0^{\sigma-1}}{E_{-1} \bm{\mathrm{I}}(\underline z_0)}
    =\frac{\mathcal I}{\sigma}\frac{1}{E_{-1}}
    \left(\frac{\bm{\mathrm{I}}(\underline z_0)}{\underline z_0^{\sigma-1}}\right)^{-1}
    =\frac{\mathcal I}{\sigma}\frac{1}{E_{-1}}\left[\frac{\beta x^\beta}{\beta-(\sigma-1)}\,\underline z_0^{-\beta}\right]^{-1}
    = \frac{\underline{z}_0^\beta}{\underline{z}_{-1}^\beta}
    \]
    \[
    \Rightarrow\quad
    \underline z_0^{\beta}
    = x^{\beta}\,\frac{\sigma-1}{\beta-(\sigma-1)}\cdot \frac{1}{f^e}\cdot \phi
    \quad\Rightarrow\quad
    \underline z_0
    = x\left\{\frac{\sigma-1}{\beta-(\sigma-1)}\cdot \frac{1}{f^e}\cdot \phi\right\}^{\!1/\beta}.
    \]


    \noindent \emph{Step 3, calculate reentry:} Now, in $t = 1, 2, ..., T^*$, fixed cost go down relative to the previous period. From the cutoff equation for entry, obtain directly 

    $$
        \underline{z}_t = x \bigg\{ \frac{f_{t}}{f^e} \frac{\sigma -1}{\beta - (\sigma -1)} \bigg\}^{1/\beta}
    $$
    where $f_{t} = \phi^{1-t/T^*}$. Hence, the entry is given by 
    $$
        E_t = \frac{\mathcal{I}}{\sigma} \frac{1}{f_t^c} \frac{\underline z_t^{\sigma -1}}{\bm{\mathrm{I}}(\underline z_t)} - \frac{
            E_{-1} \bm{\mathrm{I}}(\underline z_0) + \sum_{\tau = 1}^{t-1} E_\tau \bm{\mathrm{I}}(\underline{z}_\tau)
            }{
                \bm{\mathrm{I}}(\underline z_t)
            }.
    $$
    Note that we can simplify the first term:
    \[
    \frac{\mathcal{I}}{\sigma}\frac{1}{f_t^c}\frac{\underline z_t^{\sigma-1}}{\bm{\mathrm{I}}(\underline z_t)}
    =\frac{\mathcal{I}}{\sigma}\frac{1}{f_t^c}\,
    \frac{\underline z_t^{\sigma-1}}{x^\beta\frac{\beta}{\beta-(\sigma-1)}\underline z_t^{\sigma-1-\beta}}
    =\frac{\mathcal{I}}{\sigma}\frac{1}{f_t^c}\,\frac{\beta-(\sigma-1)}{\beta x^\beta}\,\underline z_t^\beta.
    \]
    Using $\displaystyle \underline z_t^\beta
    =x^\beta\Big\{\frac{f_t^c}{f^e}\frac{\sigma-1}{\beta-(\sigma-1)}\Big\}$ 
    \[
    \Rightarrow\
    \frac{\mathcal{I}}{\sigma}\frac{1}{f_t^c}\frac{\underline z_t^{\sigma-1}}{\bm{\mathrm{I}}(\underline z_t)}
    =\frac{\mathcal{I}}{\sigma}\cdot\frac{1}{\beta}\cdot\frac{\sigma-1}{f^e}
    =\frac{\mathcal{I}}{\beta}\frac{\sigma-1}{\sigma}\frac{1}{f^e}.
    \]
    Furthermore, the effect of pre-cycle incumbents is %
    \[
    E_{-1} \bm{\mathrm{I}}(\underline z_0)
    =\frac{\mathcal{I}}{\sigma}\,\frac{\underline z_0^{\sigma-1}}{\phi}
    =\frac{\mathcal{I}}{\sigma}\,\frac{1}{\phi}\left[x^\beta\Big\{\frac{\phi}{f^e}\frac{\sigma-1}{\beta-(\sigma-1)}\Big\}\right]^{\frac{\sigma-1}{\beta}}
    =\frac{\mathcal{I}}{\sigma}\bigg\{x^\beta\frac{1}{f^e}\frac{\sigma-1}{\beta-(\sigma-1)}\bigg\}^{\frac{\sigma-1}{\beta}}\phi^{\frac{\sigma-1-\beta}{\beta}}.
    \]
    In the summation above, we have $\bm{\mathrm{I}}(\underline{z}_\tau)$. We solve for this object explicitly %
    \[
    \begin{aligned}
    \bm{\mathrm{I}}(\underline z_t)
    &= x^\beta \frac{\beta}{\beta-(\sigma-1)}\,\underline z_t^{\,\sigma-1-\beta}
    = x^\beta \frac{\beta}{\beta-(\sigma-1)}\,\Big(\underline z_t^\beta\Big)^{\frac{\sigma-1}{\beta}-1} \\[2pt]
    &= x^\beta \frac{\beta}{\beta-(\sigma-1)}\,
    \Big[x^\beta\Big\{\frac{f_t^c}{f^e}\frac{\sigma-1}{\beta-(\sigma-1)}\Big\}\Big]^{\frac{\sigma-1}{\beta}-1} \\[2pt]
    &= \frac{\beta}{\beta-(\sigma-1)}\,x^{\sigma-1}
    \Big\{\frac{f_t^c}{f^e}\frac{\sigma-1}{\beta-(\sigma-1)}\Big\}^{\frac{\sigma-1-\beta}{\beta}} \\[2pt]
    &= \frac{\beta}{\sigma-1}\,x^{\sigma-1}
    \Big(\frac{\sigma-1}{\beta-(\sigma-1)}\Big)^{\frac{\sigma-1}{\beta}}
    \Big(\frac{f_t^c}{f^e}\Big)^{\frac{\sigma-1-\beta}{\beta}} \\[2pt]
    &= \phi^{\frac{\sigma-1-\beta}{\beta}\left(1-\frac{t}{T^*}\right)}
    \frac{\beta}{\sigma-1}\,f^e\,
    \Big\{x^\beta\frac{1}{f^e}\frac{\sigma-1}{\beta-(\sigma-1)}\Big\}^{\frac{\sigma-1}{\beta}}.
    \end{aligned}
    \]
    Therefore, 
    \[
    \begin{aligned}
    \frac{E_{-1} \bm{\mathrm{I}}(\underline z_0)}{\bm{\mathrm{I}}(\underline z_t)}
    &= \frac{\displaystyle \frac{\mathcal I}{\sigma}
    \Big\{x^\beta \frac{1}{f^e}\frac{\sigma-1}{\beta-(\sigma-1)}\Big\}^{\frac{\sigma-1}{\beta}}
    \phi^{\frac{\sigma-1-\beta}{\beta}}}
    {\displaystyle \phi^{\frac{\sigma-1-\beta}{\beta}\left(1-\frac{t}{T^*}\right)}
    \frac{\beta}{\sigma-1} f^e
    \Big\{x^\beta \frac{1}{f^e}\frac{\sigma-1}{\beta-(\sigma-1)}\Big\}^{\frac{\sigma-1}{\beta}}} \\[4pt]
    &= \frac{\mathcal I}{\sigma}\,\frac{\sigma-1}{\beta}\,\frac{1}{f^e}\,
    \phi^{\frac{\sigma-1-\beta}{\beta}\frac{t}{T^*}}.
    \end{aligned}
    \]
    Also, 
    $$
        E_\tau \frac{\bm{\mathrm{I}}(\underline{z}_\tau)}{\bm{\mathrm{I}}(\underline{z}_t)} = E_\tau \phi^{\frac{\sigma -1-\beta}{\beta} \frac{t-\tau}{T^*}} \quad \text{for}\quad t = 1, \dots, T^*.
    $$
    Then, denoting $\pi :=  \phi^{\frac{\sigma -1-\beta}{\beta} \frac{1}{T^*}} $, we can write %
    $$
        E_t = E_{-1} - E_{-1} \pi^t - \sum_{\tau=1}^{t-1}  E_\tau \pi^{t-\tau}
    $$
    \noindent \emph{Step 4, solve the recursion and find $m_\infty$:} We can verify the following guess %
    $$
        E_t = E_{-1}(1-\pi)
    $$
    From this, we can put together the measure $m_{\tau}$: 
    \begin{align*}
        m_\tau(z)   &= \left(E_{-1} \indic{\ulz_0} + \sum_{l=1}^{\tau} E_\tau \indic{\ulz_\tau}\right)\mu^E(z)\\
                    &= E_{-1}\left( \indic{\ulz_0} + \Big(1-\phi^{\frac{\sigma -1-\beta}{\beta} \frac{1}{T^*}} \Big) \sum_{l=1}^{\tau}  \indic{\ulz_\tau}\right)\mu^E(z)
    \end{align*}
    Substituting out the cutoffs $\ulz_0, \ulz_\tau$ and $E_{-1}$ yields the expression for $m^\tau$ in the proposition.\\

    \noindent \emph{Step 5, continuous-time convergence:}
    Next, we consider what happens as the crisis approaches its continuous time limit. First we look at the tail of $m_\infty$, i.e. at a value $z\geq \max \ulz_\tau = \ulz_0$ (recall that the cutoff monotone decreases after $t=0$). Note the following limit: $N(1-u^{1/N}) \rightarrow -\log u$ for any $u > 0$ as $N\rightarrow \infty$. Applying this to our problem, we have
    \begin{align*}
        m_T^*(z)&= E_{-1}\left( \indic{\ulz_0} + \Big(1-\phi^{\frac{\sigma -1-\beta}{\beta} \frac{1}{T^*}} \Big) \sum_{l=1}^{T^*}  \indic{\ulz_\tau}\right)\mu^E(z)\\
                &= E_{-1}\left( 1 + \Big(1-\phi^{\frac{\sigma -1-\beta}{\beta} \frac{1}{T^*}} \Big) T^*\right)\mu^E(z)\\
                &\longrightarrow E_{-1}\cdot \mu^E(z)\cdot \left[ 1 - \frac{\sigma -1-\beta}{\beta} \log \phi \right] \quad \text{as } T^*\rightarrow\infty
    \end{align*}%
    as claimed. Analyzing the left tail is harder. To this end, let $\alpha \in [0, 1]$ and define 
    $$
        z_\alpha = x\left[ \frac{\phi^{1-\alpha}}{f^e} \frac{\sigma -1}{\beta - (\sigma -1)} \right]^{1/\beta}.
    $$
    Note that we can reach every $z \in [\ulz_{-1}, \ulz_0]$ in the left tail by varying $\alpha$. We investigate the convergence of $m_{T^*}(z_\alpha)$. We have \begin{align*}
        m_T^*(z_\alpha)&= E_{-1}\left( \indictwo{\ulz_0}{z_\alpha} + \Big(1-\phi^{\frac{\sigma -1-\beta}{\beta} \frac{1}{T^*}} \Big) \sum_{l=1}^{T^*}  \indictwo{\ulz_\tau}{z_\alpha}\right)\mu^E(z_\alpha)\\
                &= E_{-1}\left( 0 + \Big(1-\phi^{\frac{\sigma -1-\beta}{\beta} \frac{1}{T^*}} \Big) \sum_{l=\left\lceil \alpha T^* \right\rceil  }^{T^*}  \right)\mu^E(z_\alpha)\\
                &\longrightarrow (1-\alpha)\cdot E_{-1}\cdot \mu^E(z_\alpha)\cdot \left[ \frac{\beta - (\sigma -1)}{\beta} \log \phi \right] \quad \text{as} \quad T^*\rightarrow\infty.
    \end{align*} %
    Finally get rid of $\alpha$ in the limit. Inverting $z_\alpha$ for $\alpha$ yields \begin{align*}
        \alpha = \left[ 1-\frac{\log\left( \big(\frac{z_\alpha}{x}\big)^\beta f^e \frac{\beta - (\sigma-1)}{\sigma -1} \right)}{\log \phi} \right] =  \left[ 1-\beta\frac{ \log\left( \frac{z_\alpha}{x}\right) - \log \left( \frac{\ulz_{-1}}{x} \right)}{\log \phi} \right] 
    \end{align*}%
    Where we have used the formula for $\ulz_{-1}$. Substitute this expression back into the limit and the result follows.\\

    \noindent \emph{Step 6, firm mass after the cycle $M_{T^*}$:} Recall that the anti-cumulative of $\mu^E$ is denoted by $p_E$, and $p_E(a) = \left(\frac{x}{a}\right)^\beta$. Hence, we are looking for \begin{align*}
        M_{T^*} &= E_{-1} p_E(\ulz_0) + \sum_{\tau = 1}^{T^*} E_\tau p_E(\ulz_\tau) \\
                &= \frac{\mathcal{I}}{\sigma} \frac{\beta - (\sigma -1)}{\beta} \cdot \phi^{-1} + 
                 \sum_\tau \frac{\mathcal{I}}{\sigma} \frac{\beta - (\sigma -1)}{\beta} \left[1-\phi^{\frac{\sigma -1 -\beta}{\beta} \frac{1}{T^*}}\right]\cdot \phi^{\frac{\tau}{T^*}-1} \\
                &= \frac{\mathcal{I}}{\sigma} \frac{\beta - (\sigma -1)}{\beta} \cdot \phi^{-1} \left\{1 + 
                 \left[1-\phi^{\frac{\sigma -1 -\beta}{\beta} \frac{1}{T^*}}\right]\cdot \sum_\tau \phi^{\frac{\tau}{T^*}}\right\}\\
                &= \frac{\mathcal{I}}{\sigma} \frac{\beta - (\sigma -1)}{\beta} \cdot \phi^{-1} \left\{1 + 
                 \left[1-\phi^{\frac{\sigma -1 -\beta}{\beta} \frac{1}{T^*}}\right]\cdot \phi^{\frac{1}{T^*}}\frac{1-\phi^{\frac{T^*}{T^*}}}{1-\phi^{\frac{1}{T^*}}}\right\}
    \end{align*}
    Noting that 
    $$
        \frac{\left[1-\phi^{\frac{\sigma -1 -\beta}{\beta} \frac{1}{T^*}}\right]}{1-\phi^{\frac{1}{T^*}}} \rightarrow \frac{\sigma -1 -\beta}{\beta}
    $$
    we get \begin{align*}
        M_{T^*} &= \frac{\mathcal{I}}{\sigma} \frac{\beta - (\sigma -1)}{\beta} \cdot \phi^{-1} \left\{1 + 
                 (\phi - 1) \frac{ \beta - (\sigma -1)}{\beta} \right\}
    \end{align*}  which equals the expression in the proposition.

\end{proof}
}

\begin{proof}[Proof of Proposition \ref{remark:GE_exponential_decay_entry}]
    Consider an economy in which both the distribution of entrants and of incumbents are kept general. Suppose, that the fixed cost have jumped up in the close past, hence the initial measure of incumbents, $m_0$ is truncated at some level $z^*$, but have since reverted back to some fixed level, $f^c$. This lower level of fixed cost is associated with cutoff $\ulz < z^*$. The mass of active firms is given by $M_t=\int m_t$, of which $M_t^I = \int m_{t-1}$ are incumbents. Like before, entrant masses are denoted by $E_t$. The distributional law of motion is given by \begin{align}
    m_t(z) = m_{t-1}(z) + \sum_{\tau = 1}^{t} E_t \indic{\ulz} \mu^E(z),
\end{align} %
where we can ignore the truncation of last period firms, since the cutoff does not move up over time. We impose a general equilibrium condition, labor markets must clear through \begin{align}
    \mathcal{I}_t = \bar{L} + \Pi = \frac{\sigma}{\sigma-1} (\bar{L} - M_t f^c - E_t f^e).
\end{align}
Generally, let $\E_f(x) = \int x f(x) \d x / \int_{0}^{\infty} f(x) \d x$ be the expectation with respect to any non-negative function $f$, and let $p_f(a) = \int_{a}^{\infty} f(x) \d x / \int_{0}^{\infty} f(x) \d x$ be the anti-cumulative distribution function corresponding to $f$. $E$ subscripts i.e., $\E_E$ and $p_E$, refer to the entry distribution, $\mu^E$. The entry equation reads %
{\footnotesize
\begin{align}
    E_t &= \frac{1}{f^c}\frac{1}{\sigma-1} \left( \bar{L} - \sum_{1\leq \tau \leq t-1} p_E(\ulz) E_\tau f^c - E_t(f^e + f^c p_E(\ulz)) - M_0^I f^c \right) \underbrace{\frac{\ulz^{\sigma-1}}{ \int_{\ulz} z^{\sigma-1} \mu^E(z) \d z}}_{\tilde{Z} \, \equiv} \notag \\ 
    & \quad - \sum_{1\leq \tau \leq t-1} E_\tau - \underbrace{ \frac{ \int_{\ulz} z^{\sigma-1} m_0(z) \d z }{ \int_{\ulz} z^{\sigma-1} \mu^E(z) \d z} }_{\tilde{Z}_0 \, \equiv} \\
&\iff \notag \\
    E_t &= \underbrace{\frac{
            \frac{\tilde{Z}}{\sigma-1} \frac{\bar{L} - M_0 f^c }{f^c} - \tilde{Z}_0
        }{
            1 + \frac{\tilde{Z}}{\sigma-1} (f^e / f^c + p_E(\ulz))
        }}_{A\, \equiv} - %
        \left[\sum_{1\leq \tau\leq t-1} E_\tau\right] 
        \underbrace{
            \left[ \frac{1 + p_E(\ulz) \frac{\tilde{Z}}{\sigma-1} }{ 1 + \frac{\tilde{Z}}{\sigma-1} (f^e/f^c + p_E(\ulz)) } \right]
        }_{b\, \equiv}
\end{align}%
}%
from which one deduces that $E_t = (1-b)^{t-1} A$. %
We can rewrite the coefficients as (letting $p_{0}$ be the anti-cumulative of $m_0$) {\footnotesize\begin{align*}
    A = \frac{
        (\bar{L}/f^c - M_0) - \E_{m_{0}}\left[ \big(\frac{z }{\ulz}\big)^{\sigma-1} \vert z \geq \ulz \right] (\sigma-1)\, p_{0}(\ulz)
    }{
        \frac{f^e}{f^c} + \Big(\E_{E}\left[ \big(\frac{z }{\ulz}\big)^{\sigma-1} \vert z \geq \ulz \right] (\sigma-1) + 1\Big)\, p_E(\ulz)
    }, \quad
    1 - b = \frac{
        \frac{f^e}{f^c}
    }{
        \frac{f^e}{f^c} + \Big(\E_{E}\left[ \big(\frac{z }{\ulz}\big)^{\sigma-1} \vert z \geq \ulz \right] (\sigma-1) + 1\Big)\, p_E(\ulz)
    }.
\end{align*}} %
Therefore, finally %
\begin{align}
    E_t = \frac{f^c}{f^e} \left( 
        1-b
     \right)^t 
     \cdot \left[(\bar{L}/f^c - M_0) - \E_{m_{0}}\left[ \bigg(\frac{z }{\ulz}\bigg)^{\sigma-1} \vert z \geq \ulz \right] (\sigma-1)\, p_{0}(\ulz)\right].
\end{align}
\end{proof}

\begin{proof}[Proof of Proposition \ref{prop:ext_idrisk_dyn}]
    In a state of entry, we obtain 
    \begin{align}
        \frac{f^{e}}{f^{c}}
        &= \int_{\underline z}^{\infty}
        \left[
        \frac{\mathbb{E}_{\psi}\!\big[(z+\varepsilon)^{\sigma-1}\big]}
             {\mathbb{E}_{\psi}\!\big[(\underline z+\varepsilon)^{\sigma-1}\big]}
        - 1
        \right]\mu^{E}(z)\,dz = \frac{\tilde Z(\mu,\psi)}{\mathbb{E}_{\psi}\!\big[(\underline z+\varepsilon)^{\sigma-1}\big]}
          - p_E(\underline z),
    \end{align} %
from which we see immediately that cutoffs are independent of incumbents. Like in the baseline, equilibrium is established through entry $E_\varepsilon$ which shifts $\tilde{Z}$. For the second part of the proposition, note that a firm which enters with productivity $z$ in period $k$ of a cycle phase $\tau \in \{1, 2, 3\}$ is still alive in period $t>k$ with probability $ \mathbb{P}_\psi(z + \varepsilon < \underline{z}_\varepsilon)^{t-k}, $ which converges to 0 if and only if $\mathbb{P}_\psi(z + \varepsilon < \underline{z}_\varepsilon) > 0$. Therefore, all incumbents at the beginning of the period for whom $\mathbb{P}_\psi(z + \varepsilon < \underline{z}_\varepsilon) > 0$ holds true will also exit eventually, and only if there is a positive measure of incumbents who never exit, path dependency holds.
\end{proof}

\begin{proof}[Proof of Proposition \ref{prop:multi-prod-cycles}]

    For $(ii)$, a very similar argument to Proposition \ref{prop:ext_TFP} applies. Note that individual firms' profits in eq. (\ref{eq:prof-multiprod}) are independent of $f^p$. To see this, note that the only element in eq. (\ref{eq:prof-multiprod}) that could depend on $f^p$ is $L^p$. Using eq. (\ref{eq:mkt-clearing-multiprod}) in steady state, i.e. with $E=0$, we can write $L^p$ as:
    \begin{align}
        L^p = \bar L - M f^c - N f^p
    \end{align}
    Integrating eq. (\ref{eq:n(z)-multi-prod}), and substituting it in the above, we obtain:
    \begin{align}
        L^p = \bar L - M f^c - \frac{L^p}{(\eta-1) f^p} f^p
    \implies
        L^p = \frac{\eta-1}{\eta} (\bar L - Mf^c)
    \end{align}
    As firm profits are independent of $f^p$, entry and exit decisions are unaffected by $f^p$. Therefore, movements in $f^p$ do not change $(m, \underline z)$, and there is no path-dependency and, thus, no long-run effects.

    For $(i) \; (a)$, firstly, note that in PE $\mathcal{I}$ is exogenously fixed. Therefore, $\mathcal{I}_3 = \mathcal{I}_1$, and given that $L_\tau^p = \frac{\sigma-1}{\sigma} \mathcal{I}_\tau$, also $L_3^p = L_1^p$. Secondly, note that the proof of Proposition \ref{proposition:PE-cycles} applies with $z^{\phi-1}$ instead of $z^{\sigma-1}$, which has no material effect. Therefore, following the steps of Proposition \ref{proposition:PE-cycles}, we obtain that  $\mathcal{Z}_3 = \mathcal{Z}_1$ and $M_3 < M_1$. In addition, given that by integrating eq. (\ref{eq:n(z)-multi-prod}) we have $N = \frac{L^p}{(\eta-1) f^p}$, we obtain that $N_3 = N_1$. The condition on the $q$'s immediately follows.
    
    For $(i) \; (b)$, note that the proof of Proposition \ref{prop:output-GE} applies with $z^{\phi-1}$ instead of $z^{\sigma-1}$, which has no material effect, and with:
    \begin{align}
        R_3 - R_1 &= \frac{\sigma}{\sigma-1} \frac{\eta-1}{\eta} f^c (M_1 - M_3)
    \end{align}
    where $R_3 - R_1$ is obtained by using that $R_\tau = \frac{\sigma}{\sigma-1} L_\tau^p$. Using the argumentation in the proof of Proposition \ref{prop:output-GE}, we obtain that $ M_3 < M_1$ and $\mathcal{Z}_3 > \mathcal{Z}_1$. For the multi-product setting, we have that:
    \begin{align}
        L^p = \frac{\eta-1}{\eta} (\bar L - Mf^c),\quad N = \frac{\bar L - Mf^c}{\eta f^p}
    \end{align}
    where to obtain $N$, we substitute $L^p$ in $N=\frac{L^p}{(\eta-1) f^p}$. Therefore, we can conclude that: $ L_3^p > L_1^p, \; N_3 > N_1$. The condition on the $q$'s immediately follows.

    For $(i) \; (c)$, firstly note that: $Y_\tau^{CES}:= N_\tau^{q_V} L_\tau^p \mathcal{Z}_\tau$ and thus
    $Y_\tau^{CES} \propto (\bar L - Mf^c)^{1+q_V} \mathcal{Z}_\tau$.
    Therefore, the only changes to eq. (\ref{eq:elast}) are that within the parenthesis:
    \begin{enumerate}
        \item instead of $(q-q^{CES})$ we have $(q_F - \frac{1}{\sigma-1}) - (q_V - \frac{1}{\eta-1})$
        \item for $\frac{\partial \log Y^{CES}_3}{\partial \log  M_3 }$ we just have that the labor cleansing effect is amplified. Recall that in the single-product firm case, we had $Y_\tau^{CES} = (\bar L - Mf^c) \mathcal{Z}_\tau.$ However, the elasticity is still increasing (in absolute terms) in $M_3$.
    \end{enumerate}
    Therefore, the proof of Proposition \ref{prop:compreces} applies, with the only difference that the $q$-interval in which the long-run effects are dependent on $f_h^c$ needs to have that $q_\circ - \frac{1}{\sigma-1} > q_V - \frac{1}{\eta-1}$, as stated in the proposition.
    
    
    
    
\end{proof}

\begin{proof}[Proof of Proposition \ref{prop:output-fin-good}]

We define the following proof-specific objects:
\begin{flalign*}
& \tilde Z_\tau := \int z^{\sigma-1} m_\tau(z) dz, \quad a := \alpha(q-q^{\text{CES}}), \quad b := \frac{\alpha+1-\sigma}{\sigma-1}, \quad \tilde \sigma := \sigma^\alpha (\sigma-1)^{1-\alpha}.
\end{flalign*}
Let the technology for entry and per-period fixed operating costs be Cobb–Douglas in the final good and labor:
    \begin{align*}
            B(Y,L) \;&=\;\kappa(\alpha) \;  Y^{\alpha} L^{\,1-\alpha}, \qquad \kappa(\alpha)= \alpha^\alpha (1-\alpha)^{1-\alpha}, \quad  \alpha\in[0,1].
    \end{align*}
    Let $(P_t,w_t)$ denote the prices of the final good and labor, then the cost of the bundle $B$ is
    \begin{align*}
            \Phi_t(P_t,w_t) \;&=\; P_t^{\alpha} w_t^{\,1-\alpha}.
    \end{align*}
    The firm requires $f^e$ units of $B$ to enter and $f^c$ units of $B$, per period, to produce. As in Proposition \ref{prop:output-GE}, we analyze steady-state output, therefore, there is no payment of entry costs. With labor as the numeraire ($w_t=1$), the nominal expenditure on fixed cost payments at time $\tau$ is
    \begin{align*}
        f^c \cdot \Phi_\tau(P_\tau,1)
        \;&=\;
        f^c\, P_\tau^{\alpha}.
    \end{align*}
    This nests two special cases. If $\alpha=0$, then the fixed costs are paid only in units of labor, and the nominal expenditure is $w_\tau f^c= f^c$. If $\alpha=1$, then the fixed costs are paid only in final-good units, and the nominal expenditure is $f^c P_\tau$.

\paragraph{Change in output after the crisis} The zero-profit condition at time $\tau$ can be written as
\begin{align}
    &\frac{L^p_{\tau} \; \underline z_\tau^{\sigma-1}}{(\sigma-1) \; P_t^\alpha \;  f^c_\tau} = \mathcal{Z}_\tau^{\sigma-1} \iff
    &\frac{L^p_\tau}{\sigma^\alpha (\sigma-1)^{1-\alpha} f^c_\tau} M_\tau^{\alpha(q-\frac{1}{\sigma-1})} \underline z_\tau^{\sigma-1} =  \mathcal{Z}_\tau^{\sigma-\alpha-1}. \label{zpc-cobb-fc-v1}
\end{align}
Dividing the zero-profit conditions at $\tau=3$ and $\tau=1$, and using that $f^c_3 = f^c_1$, we obtain:
\begin{align}
    \frac{L^p_3}{L^p_1} \left(\frac{M_3}{M_1}\right)^{a} \left(\frac{\underline z_3}{\underline z_1}\right)^{\sigma-1} = \left(\frac{\mathcal{Z}_3}{\mathcal{Z}_1}\right)^{\sigma-\alpha-1} \notag
\end{align}
Using Lemma \ref{lemma:hsa_invariance_A_z}, which also applies in this case, and implies $\underline z_3=\underline z_1$, we obtain:
\begin{align}
    \frac{L^p_3}{L^p_1} \left(\frac{M_3}{M_1}\right)^{a} \frac{\mathcal{Z}_3}{\mathcal{Z}_1} = \left(\frac{\mathcal{Z}_3}{\mathcal{Z}_1}\right)^{\sigma-\alpha} \iff
    \frac{Y_3}{Y_1} \left(\frac{M_3}{M_1}\right)^{(\alpha-1)(q-q^{CES})} = \left(\frac{\mathcal{Z}_3}{\mathcal{Z}_1}\right)^{\sigma-\alpha} \notag
\end{align}
Taking logs, we obtain
\begin{align}\label{out-post-crisis-cb}
    \Delta \log Y = (\sigma-\alpha) \Delta\log \mathcal{Z} + (1-\alpha)(q-q^{CES}) \Delta \log M.
\end{align}
Note that we can also rewrite eq. (\ref{out-post-crisis-cb}) like in the proposition:
\begin{align}\label{out-post-crisis-cb-rewrt}
    \Delta \log Y = \frac{\sigma-\alpha}{\sigma-1} \Delta\log \tilde{Z} + (1-\alpha)(q-q^{CES}) \Delta \log M.
\end{align}

\paragraph{$\bm{q^{\star}}$} When fixed costs are entirely paid in units of output, i.e. $\alpha=1$, we have that
\begin{align}
    \Delta\log Y = (\sigma-1)\,\Delta\log\mathcal Z = \Delta \log \tilde Z, \notag
\end{align}
thus $\Delta\log Y=0 \iff \Delta \log \tilde Z=0$, which holds at $q=q^{\text{CES}}$. Hence $q^\star(1)=q^{\text{CES}}$. To see that $\Delta \log \mathcal{Z}=0$ at $q=q^{CES}$ when $\alpha=1$, start from eq. (\ref{zpc-cobb-fc-v1}), and substitute $q=q^{CES}=\frac{1}{\sigma-1}$ and $\alpha=1$, then
\begin{align}
    \frac{L^p_\tau}{\sigma f^c_\tau} \underline z_\tau^{\sigma-1} =  \mathcal{Z}_\tau^{\sigma-2} \notag
\end{align}
Lemma \ref{lemma:hsa_invariance_A_z}, which still applies here, together with $f_3 = f_1$ and $L^p_\tau = \bar L \; \forall \tau$, as fixed costs are not paid in labor with $\alpha=1$, immediately give the result. Starting from eq. (\ref{out-post-crisis-cb}), for $\alpha\in[0,1)$, the $q^\star(\alpha)$ solving $\Delta\log Y=0$ satisfies
\begin{align}
q^\star(\alpha)
\;=\;
q^{\text{CES}}
\;-\;
\frac{\sigma-\alpha}{1-\alpha}\,
\frac{\Delta\log \mathcal Z\big(q^\star(\alpha),\alpha\big)}
     {\Delta\log M\big(q^\star(\alpha),\alpha\big)}.\label{eq:qstar-alpha}
\end{align}
Equation \eqref{eq:qstar-alpha} implies that $q^\star(\alpha)> q^{\text{CES}}$ when $\alpha<1$. The statement follows from $\Delta\log M<0$ and $\Delta\log \mathcal Z > 0$ after the crisis, because of the labor-saving effect as soon as $\alpha<1$. 

Therefore, as per the proposition, we have that for each $\alpha\in[0,1]$, there exists a $q^\star(\alpha)$ such that $\Delta\log Y=0$. Moreover, $q^\star(\alpha)\ge q^{\text{CES}}$ for all $\alpha$, with $q^\star(1)=q^{\text{CES}}$.

\paragraph{$\mathbf{q \neq q^{\star}}$} When $\alpha=0$, we characterize the change in output after the crisis globally, see Proposition \ref{prop:output-GE}. For $\alpha \in (0, 1]$, we can characterize the change in output locally around $q^\star(\alpha)$, to first order. The zero-profit condition at $t=\tau$, i.e. eq. (\ref{zpc-cobb-fc-v1}), can be expressed as
$\frac{L_\tau^p \; \underline z_\tau^{\sigma-1}}{\tilde \sigma} \tilde Z_\tau^b M_\tau^a =  f^c_\tau$. 
By taking logs, we obtain:
\begin{align}
    \log L_\tau^p + (\sigma-1) \log \underline z_\tau - \log \tilde \sigma + b \log \tilde Z_\tau + a \log M_\tau = \log  f^c_\tau. \label{zpc-fixed-out-logs}
\end{align}
By subtracting eq. (\ref{zpc-fixed-out-logs}) at $t=1$ to $t=3$, we obtain:
\begin{align}\label{log-diff-zpc-cb}
    \Delta \log L^p + b \Delta \log \tilde Z + a \Delta \log M = 0,
\end{align}
where $\Delta \log X := \log X_3 - \log X_1$, and we used that $f^c_3 = f^c_1$ and Lemma \ref{lemma:hsa_invariance_A_z}. We can solve for $\Delta \log \tilde Z$ in eq. (\ref{log-diff-zpc-cb}), which gives us:
\begin{align}
    \Delta \log \tilde Z = -\frac{1}{b} \left[a(q) \Delta \log M +  \Delta \log L^p \right],\notag
\end{align}
where we explicitly note that $a$ is a function of $q$.
We can now substitute for $\Delta \log \tilde Z$ in eq. (\ref{out-post-crisis-cb-rewrt}), obtaining:
\begin{align}
    \Delta \log Y = \left[(1-\alpha)(q-q^{CES})- \frac{\sigma-\alpha}{\sigma-1} \frac{a(q)}{b}   \right] \Delta \log M - \frac{\sigma-\alpha}{\sigma-1} \frac{1}{b} \Delta\log L^p.
\end{align}
Then, the first-order effect of $q$ on $\Delta \log Y$ locally around $q^\star(\alpha)$ is:
\begin{align}
\left.\frac{\partial}{\partial q}\Delta\log Y(q,\alpha)\right|_{q=q^\star(\alpha)}
&=\Big[(1-\alpha)-\frac{(\sigma-\alpha)}{(\sigma-1)} \frac{\alpha}{b}\Big]\,
\Delta\log M\big|_{q^\star(\alpha)}.
\end{align}
Therefore, given that $\Delta \log M <0$, we have that if $\sigma > \alpha + 1$, and thus $b = \frac{\alpha + 1 -\sigma}{\sigma-1}<0$, then
$
    \left.\frac{\partial}{\partial q}\Delta\log Y(q,\alpha)\right|_{q=q^\star(\alpha)}<0,
$
implying that 
$
q \gtrless q^{\star}(\alpha) \Rightarrow \Delta\log Y(q) \lessgtr 0$ for $q$ in a neighborhood of $q^{\star}(\alpha)$, to first order.

\end{proof}

\section{Data and Empirical Analysis}

\subsection{WIOD Data} \label{app:wiod}
\subsubsection{WIOD Structure and Coverage}
\begin{figure}[H]
\caption{World Input Output Table }
\label{wiod}\includegraphics[width=1\textwidth]{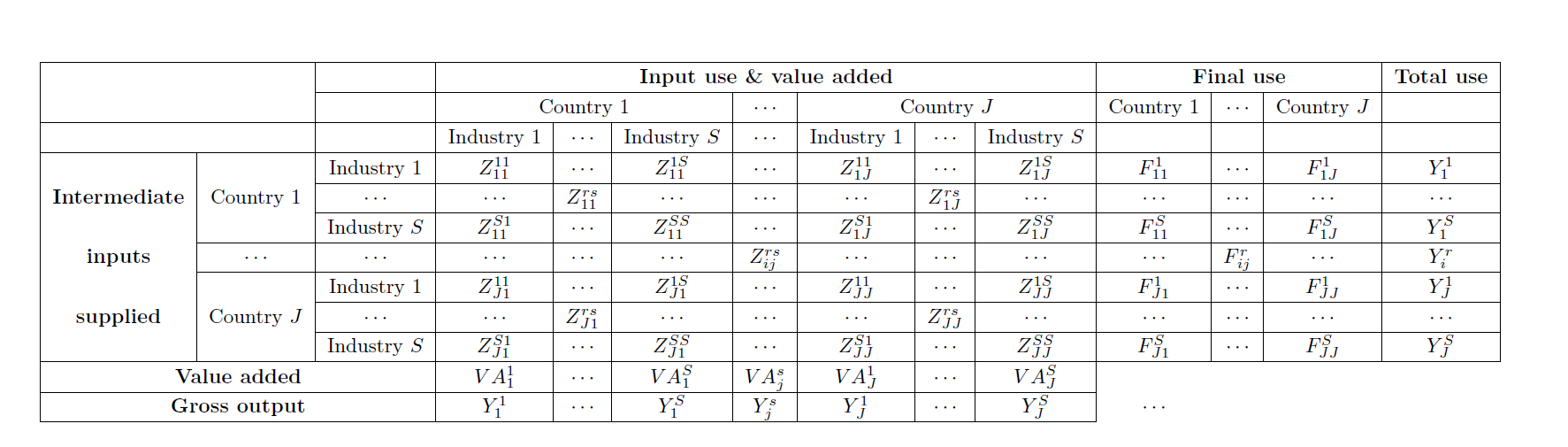} 
\end{figure}
\noindent\justifying
 \footnotesize{Note: The figure shows a schematic of the structure of the World Input-Output Database.}
 
\begin{table}[htp!]
\caption{Countries}
\label{country}
\center\scalebox{0.85}{

\begin{tabular}{llll} \hline
Country   \\ \hline
Australia &Denmark & Ireland & Poland\\
Austria &Spain & Italy &Portugal\\
Belgium &Estonia & Japan&Romania \\
Bulgaria &Finland & Republic of Korea&Russian Federation  \\
Brazil &France & Lithuania& Slovakia\\
Canada &United Kingdom & Luxembourg&Slovenia \\
Switzerland &Greece & Latvia &Sweden\\
China &Croatia & Mexico &Turkey\\
Cyprus &Hungary & Malta &Taiwan\\
Czech Republic &Indonesia & Netherlands&United States \\
Germany &India & Norway &\\
 \hline
\end{tabular}}
\end{table}

\begin{table}[H]
\caption{Industries}
\label{industry}
\hspace{-0cm}{\scalebox{.63}{\begin{tabular}{ll}
\hline{Industry}&{Industry}\\\hline
Crop and animal production&Wholesale trade\\
Forestry and logging&Retail trade\\
Fishing and aquaculture&Land transport and transport via pipelines\\
Mining and quarrying&Water transport\\
Manufacture of food products&Air transport\\
Manufacture of textiles&Warehousing and support activities for transportation\\
Manufacture of wood and of products of wood and cork&Postal and courier activities\\
Manufacture of paper and paper products&Accommodation and food service activities\\
Printing and reproduction of recorded media&Publishing activities\\
Manufacture of coke and refined petroleum products&Motion picture\\
Manufacture of chemicals and chemical products&Telecommunications\\
Manufacture of basic pharmaceutical products and pharmaceutical preparations&Computer programming\\
Manufacture of rubber and plastic products&Financial service activities\\
Manufacture of other non-metallic mineral products&Insurance\\
Manufacture of basic metals&Activities auxiliary to financial services and insurance activities\\
Manufacture of fabricated metal products&Real estate activities\\
Manufacture of computer&Legal and accounting activities\\
Manufacture of electrical equipment&Architectural and engineering activities\\
Manufacture of machinery and equipment n.e.c.&Scientific research and development\\
Manufacture of motor vehicles&Advertising and market research\\
Manufacture of other transport equipment&Other professiona activities\\
Manufacture of furniture&Administrative and support service activities\\
Repair and installation of machinery and equipment&Public administration and defence\\
Electricity&Education\\
Water collection&Human health and social work activities\\
Sewerage&Other service activities\\
Construction&Activities of households as employers\\
Wholesale and retail trade and repair of motor vehicles and motorcycles&Activities of extraterritorial organizations and bodies\\
\hline\end{tabular}
}}
\end{table}
\subsubsection{Output Changes}\label{sec:salesVoutput}
\normalsize
The WIOD data contains information on gross output at current prices for all industries in a given year. It also provides a version of the same information at previous year prices (PYP).
Define $S_t = P_tY_t$ and $S_t^{t-1} =P_{t-1}Y_t$. We can write 
\begin{align}\nonumber
    \log\frac{ S_t^{t-1}}{S_{t-1}^{t-1}} =\log P_{t-1}+ \log Y_t-\log P_{t-1}-\log Y_{t-1} =\log Y_t-\log Y_{t-1} =\Delta\log Y_t.
\end{align}

\subsection{SABI Data}

\subsubsection{TFP Estimation}\label{app:sabi_tfp_estim}

We estimate firm-level revenue productivity (TFPR) using the control-function approach of \cite{ackerberg_identification_2015}, separately by 2-digit NACE industry. We start from the revenue production function
$$
y_{it}=\beta_0+\beta_m m_{it}+\beta_\ell \ell_{it}+\beta_k k_{it}+\omega_{it}+\varepsilon_{it},
$$
where $\omega_{it}$ is firm productivity and evolves as a first-order Markov process $\omega_{it}=h(\omega_{i,t-1})+\xi_{it}$. In the first stage, we estimate $\iota_{it}\equiv\mathbb{E}[y_{it}\mid m_{it},\ell_{it},k_{it},t]$ by OLS, regressing $y_{it}$ on a flexible third-order polynomial in current inputs (including interactions) and year fixed effects, and we set $\hat\iota_{it}$ to the fitted value. In the second stage, for any candidate $\beta=(\beta_0,\beta_m,\beta_\ell,\beta_k)$ we form implied productivity $\hat\omega_{it}(\beta)=\hat\iota_{it}-[1,m_{it},\ell_{it},k_{it}]\beta$ and $\hat\omega_{i,t-1}(\beta)$ analogously, approximate $h(\cdot)$ by a cubic polynomial in $\hat\omega_{i,t-1}(\beta)$ estimated by OLS, and define the innovation $\hat\varkappa_{it}(\beta)=\hat\omega_{it}(\beta)-\widehat h(\hat\omega_{i,t-1}(\beta))$. Finally, we estimate $\hat\beta$ by GMM using the orthogonality conditions $\mathbb{E}[\hat\varkappa_{it}(\beta)\varrho_{it}]=0$ with instruments $\varrho_{it}=(1,m_{it},m_{i,t-1},\ell_{it},\ell_{i,t-1},k_{it})$. We then compute $\widehat{\mathrm{TFPR}}_{it}=y_{it}-\hat\beta_m m_{it}-\hat\beta_\ell \ell_{it}-\hat\beta_k k_{it}$.

\subsection{Robustness analysis}\label{appendix: e-rob}

\paragraph{Weighting by industry}
We now estimate $q$ using weights based on the industry median log total sales across countries and years. This approach gives greater influence to industries with higher median total sales. As it can be seen from Table \ref{tab:q_est_ind_weights}, the point estimates are very similar to section \ref{sec:empirics}, and $q$ is significantly higher than $q^{CES}$.

\begin{table}[H]
  \centering
  
  \caption{Estimation Results with Industry-Size Weights}\label{tab:q_est_ind_weights}
  
   \scalebox{.8}{{\def\sym#1{\ifmmode^{#1}\else\(^{#1}\)\fi} \begin{tabular}{l c c | c c | c} \hline\hline & \multicolumn{2}{c|}{(D)} & \multicolumn{2}{c|}{(E)} & \multicolumn{1}{c}{(F)}\\
                    &\multicolumn{1}{c}{(1)}&\multicolumn{1}{c|}{(2)}&\multicolumn{1}{c}{(3)}&\multicolumn{1}{c|}{(4)}&\multicolumn{1}{c}{(5)}\\
                    &\multicolumn{1}{c}{$\Delta \log \mathcal{I}^i_{c,t} $}&\multicolumn{1}{c|}{$\Delta \log Y^i_{c,t}$}&\multicolumn{1}{c}{$\Delta \log L^{p,i}_{c,t} $}&\multicolumn{1}{c|}{$\Delta \log Y^i_{c,t}$}&\multicolumn{1}{c}{$\Delta \log Y^i_{c,t}$}\\
\hline
                &                     &                     &                     &                     &                     \\
 $ \eta_{c,t}^{+,i} $ &      0.0890\sym{***}&                     &       0.101\sym{***}&                     &                     \\
                    &    (0.0209)         &                     &    (0.0192)         &                     &                     \\
[.5em]
 $ \eta_{c,t-1}^{+,i} $ &      0.0967\sym{***}&                     &      0.0646\sym{***}&                     &                     \\
                    &    (0.0225)         &                     &    (0.0191)         &                     &                     \\
[.5em]
$\Delta \log \mathcal{I}^i_{c,t} $&                     &       1.460\sym{***}&                     &                     &       1.535\sym{***}\\
                    &                     &     (0.174)         &                     &                     &     (0.141)         \\
[.5em]
$\Delta \log L^{p,i}_{c,t} $&                     &                     &                     &       1.648\sym{***}&                    1.535\sym{***} \\
                    &                     &                     &                     &     (0.236)         &                  (0.141)    \\
\hline
$ P(q \leq q^{CES}) $&                     &        .074         &                     &        .031         &        .010         \\
Year FE             &         Yes         &         Yes         &         Yes         &         Yes         &         Yes         \\
Industry FE         &         Yes         &         Yes         &         Yes         &         Yes         &         Yes         \\
N                   &       10067         &       10067         &       10083         &       10083         &       10067         \\
F-Stat              &          25         &          70         &          25         &          49         &                \\
$ R^2 $             &      0.0323         &      0.0864         &      0.0426         &      0.0881         &               \\
F-Stat eq.(1)       &                     &                     &                     &                     &        1300         \\
F-Stat eq.(2)       &                     &                     &                     &                     &         127         \\
$ R^2 $ eq.(1)      &                     &                     &                     &                     &      0.0866         \\
$ R^2 $ eq.(2)      &                     &                     &                     &                     &      0.0875         \\
\hline\hline
\multicolumn{6}{l}{\footnotesize Robust Standard Errors (HC1).}\\
\multicolumn{6}{l}{\footnotesize \sym{*} \(p<0.10\), \sym{**} \(p<0.05\), \sym{***} \(p<0.01\)}\\
\end{tabular}
}
}

    \vspace{0.2cm}
    
    \begin{minipage}{\textwidth}
    \centering
    \footnotesize
     \textbf{Notes:} See Table \ref{tab:q_est}.
    \end{minipage}

\end{table}

\paragraph{Single instrument}

In our main analysis, we estimate $q$ using both lagged and contemporary positive demand shocks as instruments for income and labor. As shown in Tables \ref{tab:q_est_lag_only} and \ref{tab:q_est_curr_only}, our results are robust to the choice of instruments.

\begin{table}[H]
  \centering
  
  \caption{Estimation Results with only lagged positive demand shocks}\label{tab:q_est_lag_only}
  
   \scalebox{.8}{{\def\sym#1{\ifmmode^{#1}\else\(^{#1}\)\fi} \begin{tabular}{l c c | c c | c} \hline\hline & \multicolumn{2}{c|}{(G)} & \multicolumn{2}{c|}{(H)} & \multicolumn{1}{c}{(I)}\\
                    &\multicolumn{1}{c}{(1)}&\multicolumn{1}{c|}{(2)}&\multicolumn{1}{c}{(3)}&\multicolumn{1}{c|}{(4)}&\multicolumn{1}{c}{(5)}\\
                    &\multicolumn{1}{c}{$\Delta \log \mathcal{I}^i_{c,t} $}&\multicolumn{1}{c|}{$\Delta \log Y^i_{c,t}$}&\multicolumn{1}{c}{$\Delta \log L^{p,i}_{c,t} $}&\multicolumn{1}{c|}{$\Delta \log Y^i_{c,t}$}&\multicolumn{1}{c}{$\Delta \log Y^i_{c,t}$}\\
\hline
                &                     &                     &                     &                     &                     \\
 $ \eta_{c,t-1}^{+,i} $ &       0.104\sym{***}&                     &       0.112\sym{***}&                     &                     \\
                    &    (0.0188)         &                     &    (0.0154)         &                     &                     \\
[.5em]
$\Delta \log \mathcal{I}^i_{c,t} $&                     &       1.680\sym{***}&                     &                     &       1.622\sym{***}\\
                    &                     &     (0.255)         &                     &                     &     (0.165)         \\
[.5em]
$\Delta \log L^{p,i}_{c,t} $&                     &                     &                     &       1.568\sym{***}&                   1.622\sym{***}  \\
                    &                     &                     &                     &     (0.216)         &                   (0.165)  \\
\hline
$ P(q \leq q^{CES}) $&                     &        .032         &                     &        .047         &        .006         \\
Year FE             &         Yes         &         Yes         &         Yes         &         Yes         &         Yes         \\
Industry FE         &         Yes         &         Yes         &         Yes         &         Yes         &         Yes         \\
N                   &       15507         &       15507         &       15526         &       15526         &       15505         \\
F-Stat              &          31         &          44         &          53         &          53         &                 \\
$ R^2 $             &      0.0254         &       0.1235         &      0.0353         &       0.1237         &                \\
F-Stat eq.(1)       &                     &                     &                     &                     &           7         \\
F-Stat eq.(2)       &                     &                     &                     &                     &          10         \\
$ R^2 $ eq.(1)      &                     &                     &                     &                     &       0.12353         \\
$ R^2 $ eq.(2)      &                     &                     &                     &                     &       0.12354         \\
\hline\hline
\multicolumn{6}{l}{\footnotesize Robust Standard Errors (HC1).}\\
\multicolumn{6}{l}{\footnotesize \sym{*} \(p<0.10\), \sym{**} \(p<0.05\), \sym{***} \(p<0.01\)}\\
\end{tabular}
}
}

    \vspace{0.2cm}
    
    \begin{minipage}{\textwidth}
    \centering
    \footnotesize
     \textbf{Notes:} See Table \ref{tab:q_est}.
    \end{minipage}

\end{table}

\begin{table}[H]
  \centering
  
  \caption{Estimation Results with only contemporary positive demand shocks}\label{tab:q_est_curr_only}
  
   \scalebox{.8}{{\def\sym#1{\ifmmode^{#1}\else\(^{#1}\)\fi} \begin{tabular}{l c c | c c | c} \hline\hline & \multicolumn{2}{c|}{(J)} & \multicolumn{2}{c|}{(K)} & \multicolumn{1}{c}{(L)}\\
                    &\multicolumn{1}{c}{(1)}&\multicolumn{1}{c|}{(2)}&\multicolumn{1}{c}{(3)}&\multicolumn{1}{c|}{(4)}&\multicolumn{1}{c}{(5)}\\
                    &\multicolumn{1}{c}{$\Delta \log \mathcal{I}^i_{c,t} $}&\multicolumn{1}{c|}{$\Delta \log Y^i_{c,t}$}&\multicolumn{1}{c}{$\Delta \log L^{p,i}_{c,t} $}&\multicolumn{1}{c|}{$\Delta \log Y^i_{c,t}$}&\multicolumn{1}{c}{$\Delta \log Y^i_{c,t}$}\\
\hline
                &                     &                     &                     &                     &                     \\
 $ \eta_{c,t}^{+,i} $ &       0.114\sym{***}&                     &       0.104\sym{***}&                     &                     \\
                    &    (0.0179)         &                     &    (0.0169)         &                     &                     \\
[.5em]
$\Delta \log \mathcal{I}^i_{c,t} $&                     &       1.745\sym{***}&                     &                     &       1.816\sym{***}\\
                    &                     &     (0.235)         &                     &                     &     (0.186)         \\
[.5em]
$\Delta \log L^{p,i}_{c,t} $&                     &                     &                     &       1.919\sym{***}&                   1.816\sym{***}  \\
                    &                     &                     &                     &     (0.316)         &              (0.186)        \\
\hline
$ P(q \leq q^{CES}) $&                     &        .011         &                     &        .012         &        .001         \\
Year FE             &         Yes         &         Yes         &         Yes         &         Yes         &         Yes         \\
Industry FE         &         Yes         &         Yes         &         Yes         &         Yes         &         Yes         \\
N                   &       15029         &       15029         &       15048         &       15048         &       15028         \\
F-Stat              &          41         &          55         &          38         &          37         &               \\
$ R^2 $             &      0.0264         &      0.0689         &      0.0317         &      0.0693         &              \\
F-Stat eq.(1)       &                     &                     &                     &                     &           6         \\
F-Stat eq.(2)       &                     &                     &                     &                     &           5         \\
$ R^2 $ eq.(1)      &                     &                     &                     &                     &      0.0690         \\
$ R^2 $ eq.(2)      &                     &                     &                     &                     &      0.0689         \\
\hline\hline
\multicolumn{6}{l}{\footnotesize Robust Standard Errors (HC1).}\\
\multicolumn{6}{l}{\footnotesize \sym{*} \(p<0.10\), \sym{**} \(p<0.05\), \sym{***} \(p<0.01\)}\\
\end{tabular}
}
}

    \vspace{0.2cm}
    
    \begin{minipage}{\textwidth}
    \centering
    {   \footnotesize{
     \textbf{Notes:} See Table \ref{tab:q_est}.}}
    \end{minipage}

\end{table}

\section{Additional Results Quantitative Model} \label{sec:appendix_quant}
\normalsize
In addition, with Table \ref{tab:policy_baqaee}, Figure \ref{fig:quant_grid-lsf-baqaee}, and Figure \ref{fig:quant_grid-ssp-baqaee} we provide an analogue of Table \ref{tab:policy}, Figure \ref{fig:quant_grid-lsf}, and Figure \ref{fig:quant_grid-ssp}, respectively, for $q=0.3$, which is the love of variety in production estimated by \cite{baqaee2023supplier}. As in the main figures, random exit shocks ensure a unique long-run steady state. Any residual in the panels reflects the 15-year plotting window.

\begin{table}[H]
    \centering
    \caption{Steady-state welfare gains and recession welfare costs (CEV, \%)}
    \label{tab:policy_baqaee}
    \scalebox{.9}{\begin{threeparttable}
    \begin{tabular}{l c c c}
        \hline\hline\addlinespace
        \multicolumn{4}{l}{\textbf{Panel A: Welfare Gains from Steady-State Policy}}\\
        \addlinespace
        Steady-State Welfare & $+3.40\%$ &  &  \\
        \midrule
        \addlinespace
        \multicolumn{4}{l}{\textbf{Panel B: Welfare Costs of Recessions}}\\
        \addlinespace
        Component & CEV relative to Steady State & $\theta_{ss}$ & $\theta_{cyc}$\\
        \midrule
        Total Output & $-1.56\%$ & \multirow{4}{*}{\centering inactive} & \multirow{4}{*}{\centering inactive} \\
        \hspace{6pt}Varieties & \hspace{6pt} $-4.06\%$ &  &  \\
        \hspace{6pt}Production Labor & \hspace{6pt} $-0.39\%$ &  &  \\
        \hspace{6pt}Avg. Productivity & \hspace{6pt} $+2.88\%$ &  &  \\
        \addlinespace
        Total Output & $-0.62\%$ & \multirow{4}{*}{\centering inactive} & \multirow{4}{*}{\centering active} \\
        \hspace{6pt}Varieties & \hspace{6pt} $+0.59\%$ &  &  \\
        \hspace{6pt}Production Labor & \hspace{6pt} $-1.01\%$ &  &  \\
        \hspace{6pt}Avg. Productivity & \hspace{6pt} $-0.20\%$ &  &  \\
        \addlinespace
        Total Output & $-1.53\%$ & \multirow{4}{*}{\centering active} & \multirow{4}{*}{\centering inactive} \\
        \hspace{6pt}Varieties & \hspace{6pt} $-3.93\%$ &  &  \\
        \hspace{6pt}Production Labor & \hspace{6pt} $-0.40\%$ &  &  \\
        \hspace{6pt}Avg. Productivity & \hspace{6pt} $+2.80\%$ &  &  \\
        \addlinespace
        Total Output & $-1.50\%$ & \multirow{4}{*}{\centering active} & \multirow{4}{*}{\centering active} \\
        \hspace{6pt}Varieties & \hspace{6pt} $-1.85\%$ &  &  \\
        \hspace{6pt}Production Labor & \hspace{6pt} $-0.85\%$ &  &  \\
        \hspace{6pt}Avg. Productivity & \hspace{6pt} $+1.20\%$ &  &  \\
        \addlinespace
        \hline\hline
    \end{tabular}
    \end{threeparttable}}\vspace{6pt}
    \begin{minipage}{\textwidth}\footnotesize
    \textbf{Notes: }Results are reported using the $q$ from \cite{baqaee2023supplier}. See notes of Table \ref{tab:policy} for additional details.
    \end{minipage}
\end{table}
\vspace{0.4cm}
\begin{figure}[htbp]
    \centering
    \begin{tabular}{@{\hskip 0pt}c@{\hskip 0pt}c@{\hskip 0pt}c@{\hskip 0pt}}
        \subfloat{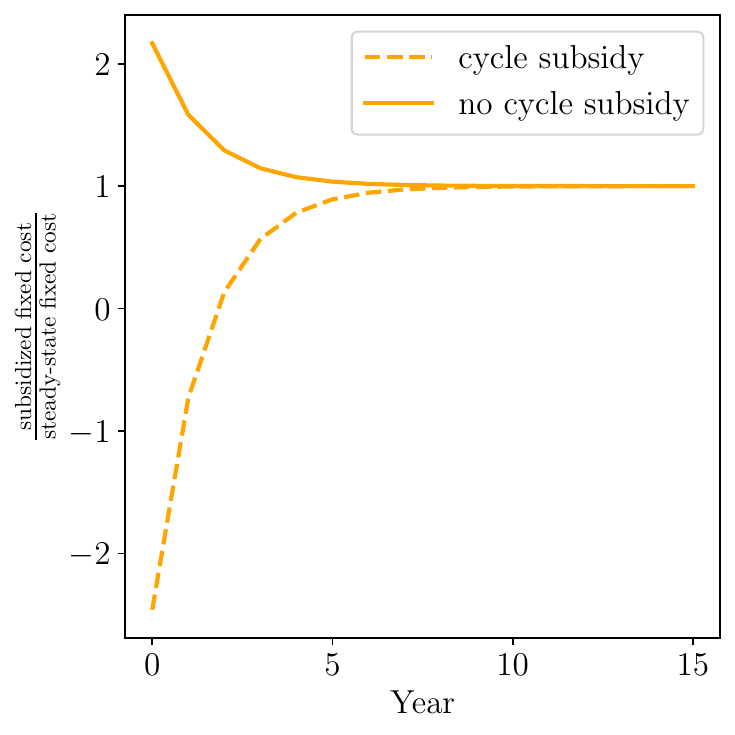}{A}
        \subfloat{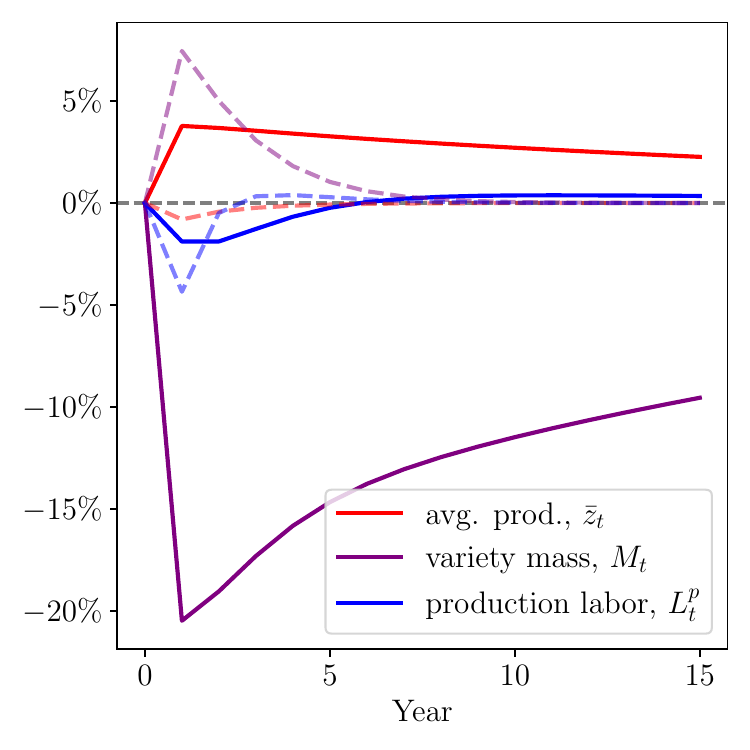}{B}
        \subfloat{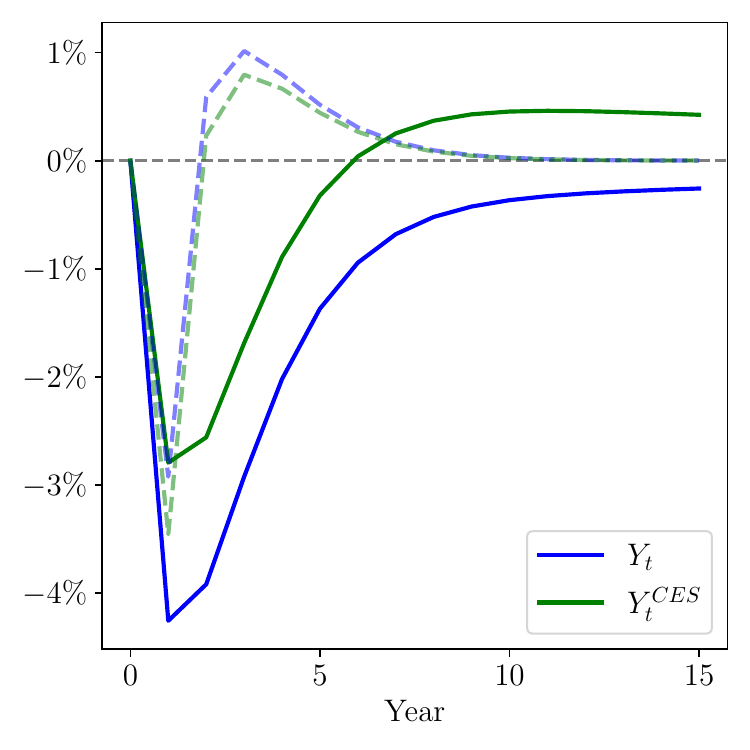}{C}
    \end{tabular}
    \caption{Fixed-cost shock path and impulse responses in a laissez-faire economy, using the $q$ from \cite{baqaee2023supplier}. See the notes of Figure \ref{fig:quant_grid-lsf} for additional details.}
    \label{fig:quant_grid-lsf-baqaee}
\end{figure}

\begin{figure}[htbp]
    \centering
    \begin{tabular}{@{\hskip 0pt}c@{\hskip 0pt}c@{\hskip 0pt}c@{\hskip 0pt}}
        \subfloat{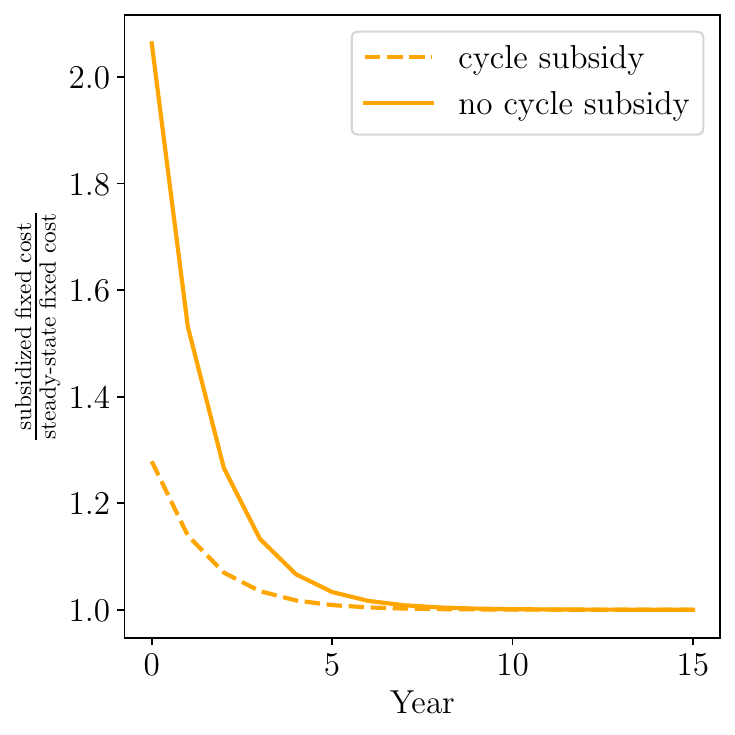}{A}
        \subfloat{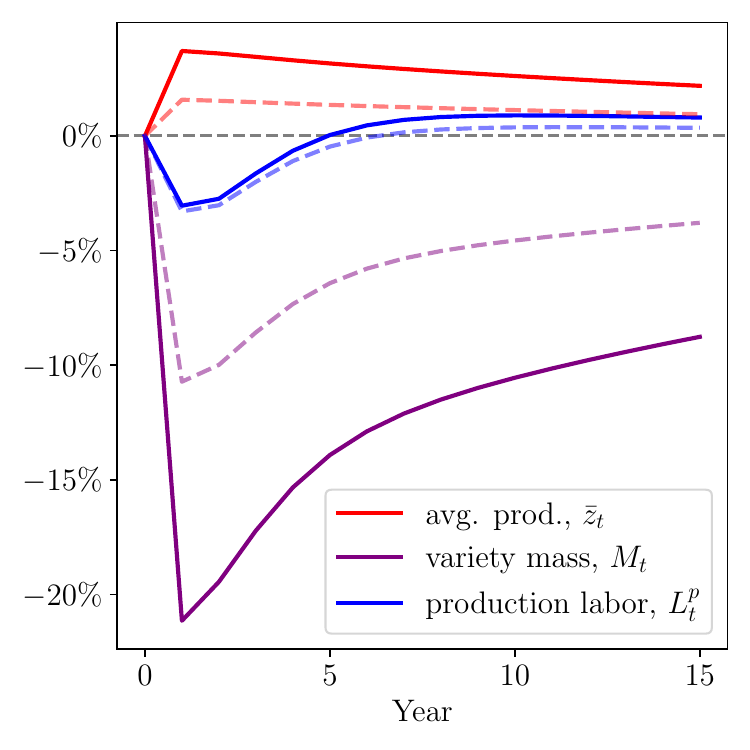}{B}
        \subfloat{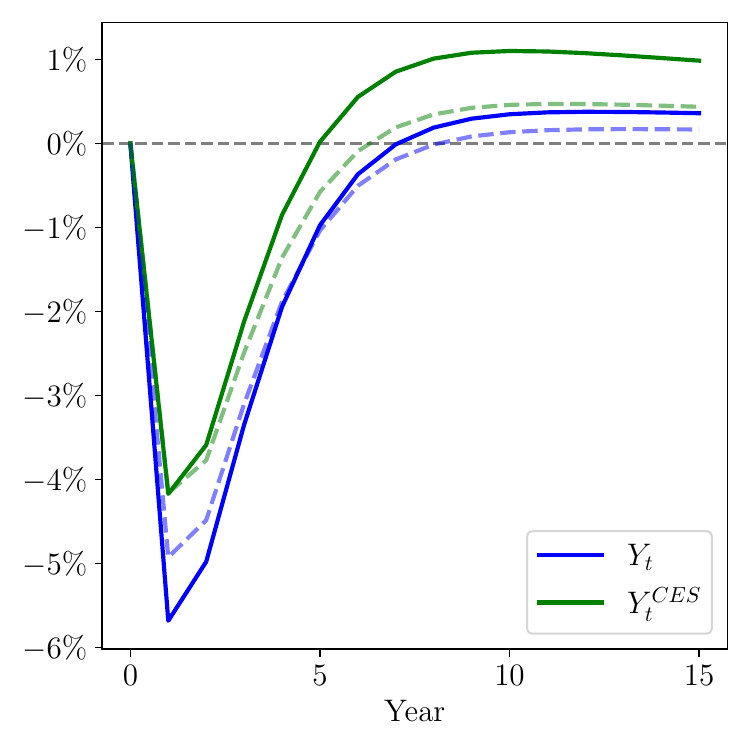}{C}
    \end{tabular}
    \caption{Fixed-cost shock path and impulse responses in an economy with active steady-state policy, using the $q$ from \cite{baqaee2023supplier}. See the notes of Figure \ref{fig:quant_grid-lsf} for additional details.}
    \label{fig:quant_grid-ssp-baqaee}
\end{figure}

\end{appendices}

\end{document}